\newif\ifec\ectrue
\documentclass[letterpaper,english,11pt]{scrartcl}
\usepackage{booktabs} 
\usepackage[ruled]{algorithm2e} 

\def\ECversion{false}
 
\newif\ifreview\reviewfalse

\usepackage{standalone}

\usepackage{amsmath,amsfonts,amsthm,amssymb}

\allowdisplaybreaks

\usepackage{relsize}
\usepackage{mathtools}
\usepackage{braket}
\usepackage{aligned-overset}

\usepackage[margin=1in,bottom=1.3in]{geometry}
\usepackage{authblk}
\usepackage{enumitem,framed}

\usepackage{longtable}

\usepackage{graphicx}

\usepackage{tikz}
\usetikzlibrary{calc}

\usepackage{pgfplots}
\pgfplotsset{compat=1.18}
\usepgfplotslibrary{fillbetween}

\tikzstyle{vertex} = [shape=circle,draw=black]
\tikzstyle{namedVertex} = [shape=circle,draw=black]
\tikzstyle{namedVertexF} = [shape=circle,draw=black,fill=white]
\tikzstyle{namedVertexW} = [shape=circle,draw=white,fill=white]
\tikzstyle{edge} = [draw,->,ultra thick]
\tikzstyle{labeledNodeS}=[circle, color=black!75!white, draw, inner sep = 0.1em, minimum size = 1.5em, scale=1.25]
\tikzstyle{normalEdge}=[very thick, >=stealth]

\definecolor{colA}{rgb}{.8,0,0}
\definecolor{colB}{rgb}{0,0,.8}
\definecolor{colC}{rgb}{0,.5,.3}
\definecolor{colD}{rgb}{.8,0,0}
\definecolor{colE}{rgb}{.8,.8,0}
\definecolor{colF}{rgb}{.8,0,.8}

\colorlet{fcolAr}{colA!70!white}
\colorlet{fcolBr}{colB!70!white}
\colorlet{fcolCr}{colC!70!white}
\colorlet{fcolDr}{colD!70!white}
\colorlet{fcolEr}{colE!70!white}
\colorlet{fcolFr}{colF!70!white}

\newcommand{\colName}[1]{{%
    \ifthenelse{\equal{#1}{A}}{\color{red}XXX}{}%
    \ifthenelse{\equal{#1}{B}}{\color{colB}blue}{}%
    \ifthenelse{\equal{#1}{C}}{\color{colC}green}{}%
    \ifthenelse{\equal{#1}{D}}{\color{colD}red}{}%
    \ifthenelse{\equal{#1}{E}}{\color{colE}yellow}{}%
    \ifthenelse{\equal{#1}{F}}{\color{colF}pink}{}%
}}

\newcommand{\colNo}[1]{{%
    \ifthenelse{\equal{#1}{A}}{\color{red}XXX}{}%
    \ifthenelse{\equal{#1}{B}}{\color{colB}1}{}%
    \ifthenelse{\equal{#1}{C}}{\color{colC}2}{}%
    \ifthenelse{\equal{#1}{D}}{\color{colD}3}{}%
    \ifthenelse{\equal{#1}{E}}{\color{colE}4}{}%
    \ifthenelse{\equal{#1}{F}}{\color{colF}5}{}%
}}

\tikzstyle{fcolA}=[color=colA,opacity=.7]
\tikzstyle{fcolB}=[color=colB,opacity=.7]
\tikzstyle{fcolC}=[color=colC,opacity=.7]
\tikzstyle{fcolD}=[color=colD,opacity=.7]
\tikzstyle{fcolE}=[color=colE,opacity=.7]
\tikzstyle{fcolF}=[color=colF,opacity=.7]

\providecommand{\QUEUEWIDTH}{18pt}
\providecommand{\QUEUEUHEIGHT}{10pt}
\providecommand{\FLOWUWIDTH}{3.5pt}

\newif\ifdebug\debugfalse

\usepackage{ifthen}

\usepackage{inputenc}
\usepackage{listofitems}

\definecolor{DFcolA}{rgb}{.8,0,0}
\definecolor{DFcolB}{rgb}{0,0,.8}
\definecolor{DFcolC}{rgb}{0,.8,0}
\definecolor{DFcolD}{rgb}{.8,.8,0}
\definecolor{DFcolE}{rgb}{.8,0,.8}
\definecolor{DFcolF}{rgb}{0,.8,.8}

\colorlet{DFfcolAr}{DFcolA!70!white}
\colorlet{DFfcolBr}{DFcolB!70!white}
\colorlet{DFfcolCr}{DFcolC!70!white}
\colorlet{DFfcolDr}{DFcolD!70!white}
\colorlet{DFfcolEr}{DFcolE!70!white}
\colorlet{DFfcolFr}{DFcolF!70!white}

\tikzstyle{DFfcolA}=[color=DFcolA,opacity=.7]
\tikzstyle{DFfcolB}=[color=DFcolB,opacity=.7]
\tikzstyle{DFfcolC}=[color=DFcolC,opacity=.7]
\tikzstyle{DFfcolD}=[color=DFcolD,opacity=.7]
\tikzstyle{DFfcolE}=[color=DFcolE,opacity=.7]
\tikzstyle{DFfcolF}=[color=DFcolF,opacity=.7]

\providecommand{\flowcolorlist}{DFfcolCr,DFfcolFr,DFfcolCr,DFfcolDr,DFfcolEr,DFfcolFr}

\providecommand{\QUEUEWIDTH}{.3}
\providecommand{\QUEUEUHEIGHT}{1}
\providecommand{\FLOWUWIDTH}{4pt}

\providecommand{\FLOWSTYLE}{parallel}

\newcommand{\drawFlow}[7][]{%
	\setsepchar[/]{,}
	\readlist\flowcoms{#7}
	\readlist\collist\flowcolorlist
	
	\def\totwidth{0}
	\foreachitem \width \in \flowcoms{\pgfmathadd{\totwidth}{\width}\xdef\totwidth{\pgfmathresult}}
	\def\cumwidth{\totwidth}
	
	\foreachitem \width \in \flowcoms{
		\ifthenelse{\equal{\width}{0.0} \OR \equal{\width}{0}}{}{
			\itemtomacro\collist[\widthcnt]\currentcol
            \ifthenelse{\equal{\FLOWSTYLE}{parallel}}{
                \path(#2)to[#1]coordinate[pos=0](tempEdgeStart)coordinate[pos=1](tempEdgeEnd) (#3);
                \draw[\currentcol,line width=\width*\FLOWUWIDTH,dash pattern={on 1pt off #5*#4 on #6*#4 off #4}, dash phase=1pt]($(tempEdgeStart)!{(.5*\totwidth-\cumwidth+.5*\width)*\FLOWUWIDTH}!90:(tempEdgeEnd)$) to[#1] ($(tempEdgeEnd)!{(.5*\totwidth-\cumwidth+.5*\width)*\FLOWUWIDTH}!-90:(tempEdgeStart)$);
            }{
                \draw[\currentcol,line width=\cumwidth*\FLOWUWIDTH,dash pattern={on 1pt off #5*#4 on #6*#4 off #4}, dash phase=1pt](#2) to[#1] (#3);
            }
			\pgfmathsubtract{\cumwidth}{\width}
			\xdef\cumwidth{\pgfmathresult}
		}
	}
}

\newcommand{\drawEdgeFlow}[5][]{%
	\ifthenelse{\equal{#5}{}}{}{
		\setsepchar[,]{I}
		\readlist\flowblocks#5 
		\foreachitem\flowblock\in\flowblocks{%
			\setsepchar[,]{/}
			\readlist\flowparts\flowblock 
			\itemtomacro\flowparts[3]\flowsplit
			\drawFlow[#1]{#2}{#3}{#4}{\flowparts[1]}{\flowparts[2]}{\flowsplit}
		}
	}
}

\newcommand{\drawQueueSegment}[3]{
	\setsepchar[/]{,}
	\readlist\flowcoms{#3}
	\readlist\collist\flowcolorlist
	
	\coordinate(tempA)at($(#1)+(-.5*\QUEUEWIDTH,0)$);
	\foreachitem \width \in \flowcoms {
		\itemtomacro\collist[\widthcnt]\currentcol
		\fill[\currentcol](tempA) rectangle +(\width*\QUEUEWIDTH,#2*\QUEUEUHEIGHT);        
		\coordinate(tempA)at($(tempA)+(\width*\QUEUEWIDTH,0)$);
	}
}

\newenvironment{placeQueue}[2]{%
	\begin{scope}[rotate around={#2:(#1)}]
		\coordinate(queuebase)at(#1);%
	}{%
	\end{scope}%
}

\newcommand{\drawEdgeQueue}[4][]{%
	\ifthenelse{\equal{#4}{}}{}{
		\begin{placeQueue}{#2}{#3}
			\setsepchar[,]{I}
			\readlist\queueblocks#4 
			\foreachitem\queueblock\in\queueblocks{%
				\setsepchar[,]{/}
				\readlist\queueparts\queueblock 
				\itemtomacro\queueparts[2]\queuesplit
				\drawQueueSegment{queuebase}{\queueparts[1]}{\queuesplit}
				\coordinate (queuebase)at($(queuebase)+(0,\queueparts[1]*\QUEUEUHEIGHT)$);
			}                      
		\end{placeQueue}
	}
}

\ifdebug
\renewcommand{\drawEdgeFlow}[5][]{%
	\drawFlow[#1]{#2}{#3}{#4}{0}{1}{1}%
}
\fi

\usepackage{csvsimple}

\NewDocumentCommand{\CharF}{O{}}{%
	\ifthenelse{\equal{#1}{}}%
	{1}%
	{1_{#1}}%
}

\usepackage{thm-restate}
\usepackage{xspace}
\usepackage{color, colortbl}
\usepackage{hyperref}
\usepackage{cleveref}

\newcommand{\refsym}[1]{%
	\ensuremath{(%
		\ifthenelse{\equal{#1}{1}}{\ast}%
		{\ifthenelse{\equal{#1}{2}}{\#}%
		{\ifthenelse{\equal{#1}{3}}{\triangle}%
		{\ifthenelse{\equal{#1}{4}}{\bigcirc}%
		{\ifthenelse{\equal{#1}{5}}{\diamond}
		{\color{red}NaN}}}}}%
	)}%
}
\newcommand{\symoverset}[2]{\toverset{\refsym{#1}}{#2}}
\newcommand{\toverset}[2]{\overset{\text{#1}}{#2}}

\newcommand{\Croverset}[2]{\overset{\text{\Crefshort{#1}}}{#2}}
\newcommand{\Crefshort}[1]{%
	{%
		\Crefname{definition}{Def.}{Def.}%
		\Crefname{lemma}{Lem.}{Lem.}%
		\Crefname{proposition}{Prop.}{Prop.}%
		\Crefname{corollary}{Cor.}{Cor.}%
		\Crefname{theorem}{Thm.}{Thm.}%
		\Crefname{observation}{Obs.}{Obs.}%
		\Crefname{claim}{Cl.}{Cl.}%
		\Crefname{section}{Sec.}{Sec.}%
		\Crefname{subsection}{Sec.}{Sec.}%
		\Crefname{example}{Ex.}{Ex.}%
		\Crefname{equation}{Eq.}{Eq.}%
		\Cref{#1}%
	}%
}

\DeclareRobustCommand{\Crefnosort}[1]{%
  \begingroup\@cref@sortfalse\Cref{#1}\endgroup
}

\newcommand{\R}{\mathbb{R}}
\newcommand{\Rnn}{\mathbb{R}_{\geq 0}}
\newcommand{\N}{\mathbb{N}}

\newcommand{\objfunc}{\vartheta}

\newcommand{\eflow}{\varphi}
\newcommand{\wflow}{\psi}
\NewDocumentCommand{\sysop}{O{\ }}{system optimal{#1}}
\NewDocumentCommand{\umini}{O{\ }}{\ensuremath{\ell^u}-minimal{#1}}
\NewDocumentCommand{\mini}{O{\ }}{\ensuremath{\ell}-minimal{#1}}
\newcommand{\tmin}{\tau^{\min}}
\newcommand{\tmax}{\tau^{\max}}

\newcommand{\GAS}{{\tilde{\GA}}}
\newcommand{\GVS}{{\tilde{\GV}}}
\newcommand{\ssource}{{\tilde{\source}}}
\newcommand{\sdest}{{\tilde{\dest}}}
\newcommand{\swir}{\tilde{\wir}}
\newcommand{\sRoutes}{\tilde{\Routes}}
\newcommand{\swa}{\tilde{\wa}}
\newcommand{\gop}{\g^*}
\newcommand{\Gop}{\G^*}
\newcommand{\M}{M}
\newcommand{\q}{q}
\newcommand{\G}{G}
\newcommand{\opp}{\nabla}
\newcommand{\hfu}{\xi}

\usepackage{bm}

\definecolor{LightCyan}{rgb}{0.88,1,1}

\def\mybig#1{{\hbox{$\left#1\vbox to23\p@{}\right.\n@space$}}}

\newcommand{\dup}[2]{\langle#1,#2\rangle}

\theoremstyle{definition}
\newtheorem{definition}{Definition}[section]
\newtheorem{assumption}[definition]{Assumption}

\theoremstyle{plain}
\newtheorem{theorem}[definition]{Theorem}

\newtheorem{lemma}[definition]{Lemma}
\newtheorem{corollary}[definition]{Corollary}
\newtheorem{claim}{Claim}
\newtheorem{subclaim}{Subclaim}[claim]

\Crefname{claim}{Claim}{Claims}
 \theoremstyle{remark}
\newtheorem{remark}[definition]{Remark}
\newtheorem{example}[definition]{Example}

\newenvironment{introthm}[2][]{\begin{framed}\textbf{\Cref{#2}}\ifthenelse{\equal{#1}{}}{}{ (#1)}:\\}{\end{framed}}

\newlist{thmparts}{enumerate}{1}
\setlist[thmparts]{
	label=\alph*)
}

\usepackage{xstring}
\usepackage{etoolbox}
\apptocmd{\cref@getref}{\xdef\@lastusedlabel{#1}}{}{error}

\Crefformat{thmpart}{%
	\StrCount{\@lastusedlabel}{:}[\LastColonPos]%
	\StrBefore[\LastColonPos]{\@lastusedlabel}{:}[\ThmLabel]%
	\nameCref{\ThmLabel}~#2\ref*{\ThmLabel}#1#3%
}
\Crefmultiformat{thmpart}{%
	\StrCount{\@lastusedlabel}{:}[\LastColonPos]%
	\StrBefore[\LastColonPos]{\@lastusedlabel}{:}[\ThmLabel]%
	\nameCref{\ThmLabel}~#2\ref*{\ThmLabel}#1#3%
}{ and~#2#1#3}{, #2#1#3}{ and~#2#1#3}
\Crefrangeformat{thmpart}{%
	\StrCount{\@lastusedlabel}{:}[\LastColonPos]%
	\StrBefore[\LastColonPos]{\@lastusedlabel}{:}[\ThmLabel]%
	\nameCref{\ThmLabel}~#3\ref*{\ThmLabel}#1#4 to~#5#2#6%
}

\newenvironment{proofClaim}[1][]{\ifthenelse{\equal{#1}{}}{\begin{proof}}{\begin{proof}[#1]}}{\end{proof}}

\newcommand{\caseitem}[1]{\def\casedescr{#1}%
	\item}
\newlist{proofbycases}{enumerate}{1}
\setlist[proofbycases]{
	leftmargin=0em,
	labelwidth=-.5em,
    parsep=0pt,
    listparindent=\parindent,
	label=\boldmath\bfseries\sffamily\arabic*. Case: \protect\casedescr:,
	ref=\arabic*,
	align=left
}

\newcommand{\proofitem}[1]{\def\pidescr{#1}%
	\item}
\newlist{structuredproof}{enumerate}{3}
\setlist[structuredproof]{
	leftmargin=1em,
    labelwidth=.3em,
    parsep=0pt,
    listparindent=\parindent,
	label=\boldmath\bfseries\sffamily\protect\pidescr:,
	align=left
}
 \setlist[structuredproof,2]{
     leftmargin=1em
 }
 \setlist[structuredproof,3]{
     leftmargin=1em
 }

\newenvironment{proofbyinduction}{\begin{description}[leftmargin=0em,parsep=0pt,listparindent=\parindent]}{\end{description}}
\newcommand{\inductionclaim}{\item[Induction Claim: ]}
\newcommand{\basecase}[1]{\item[Base Case ({\boldmath\bfseries#1}):]}
\newcommand{\inductionstep}[1]{\item[Induction Step ({\boldmath\bfseries#1}):]}

\newcommand{\myparagraph}[1]{\paragraph{#1.}}

\newcommand{\fixed@sra}{$\vrule height 2\fontdimen22\textfont2 width 0pt\shortrightarrow$}
\newcommand{\shortarrow}[1]{%
  \mathrel{\text{\rotatebox[origin=c]{\numexpr#1*45}{\fixed@sra}}}
}

\newcommand{\norm}[1]{\lVert #1 \rVert}
\newcommand{\abs}[1]{\lvert#1\rvert}

\newcommand{\wlg}{w.l.o.g.\ }

\newcommand{\wrt}{w.r.t.\ }

\usepackage{comment}

\newcommand{\di}{\;\mathrm{d}}

\newcommand{\prices}{p}

\newcommand{\arc}{e}

\newcommand{\f}{f}
\newcommand{\g}{g}
\newcommand{\vot}{\gamma}
\newcommand{\exit}{T}
\newcommand{\arr}{A}
\newcommand{\trav}{D}

\newcommand{\hori}{H}

\newcommand{\Routes}{\mathcal{W}}
\newcommand{\wa}{w}
\newcommand{\inflow}{r}

\newcommand{\wir}{\Lambda}

\DeclareMathOperator{\id}{id}

\newcommand{\GA}{E}
\newcommand{\GV}{V}
\newcommand{\source}{s}
\newcommand{\dest}{{d}}
\newcommand{\edgesFrom}[1]{\delta^+(#1)}
\newcommand{\edgesTo}[1]{\delta^-(#1)}

\NewDocumentCommand{\stwalk}{O{\source}O{\dest}}{\ensuremath{#1,#2}-walk}
\NewDocumentCommand{\stpath}{O{\source}O{\dest}}{\ensuremath{#1,#2}-path}

\newcommand{\refrunind}[2]{%
\ifthenelse{\equal{#2}{1}}{
	\ensuremath{%
		\ifthenelse{\equal{#1}{1}}{n}%
		{\ifthenelse{\equal{#1}{2}}{m}%
		{\ifthenelse{\equal{#1}{3}}{j}%
		{\ifthenelse{\equal{#1}{4}}{l}%
		{\ifthenelse{\equal{#1}{5}}{i}
		{\color{red}NaN}}}}}%
	}%
 }
{
\ensuremath{%
		\ifthenelse{\equal{#1}{1}}{N}%
		{\ifthenelse{\equal{#1}{2}}{M}%
		{\ifthenelse{\equal{#1}{3}}{J}%
		{\ifthenelse{\equal{#1}{4}}{L}%
		{\ifthenelse{\equal{#1}{5}}{I}
		{\color{red}NaN}}}}}%
	}%
}
}
 
\newcommand{\n}[1]{\ifthenelse{\equal{#1}{}}{\refrunind{1}}{\refrunind{#1}{1}}}
\newcommand{\capn}[1]{\ifthenelse{\equal{#1}{}}{\refrunind{1}}{\refrunind{#1}{2}}}

\newcommand{\Pf}{P}
\newcommand{\op}{\nabla^u}

\newcommand{\wto}{\rightharpoonup}

\NewDocumentCommand{\seql}{O{1}O{\Routes}O{L(\hori)}}{\otimes^{#1}_{#2}#3}

\newcommand{\ofeas}{\mathcal{M}}
\newcommand{\edom}[1]{\mathcal D_{#1}}
\newcommand{\minim}[1]{$\ell^{#1}$-minimal}
\NewDocumentCommand{\sink}{O{\ }}{destination#1}
\newcommand{\Sink}{Destination }
\Crefname{@theorem}{Theorem}{Theorems}
\Crefname{assumption}{Assumption}{Assumptions}

\usepackage{xr}
\newcommand{\oref}[1]{%
  \@ifundefined{r@#1}{%
   {\cite[\Cref*{FD-ext-#1}]{GHS24FD}}%
  }{%
   \Cref{#1}%
  }%
}

\newcommand{\FDref}[1]{\cite[\Cref*{FD-ext-#1}]{GHS24FD}}

\newcommand{\Tollref}[1]{\cite[\Cref*{Tolls-ext-#1}]{GHS24TollSODA}}

 \usepackage{apptools}

\ifreview
    \usepackage[disable]{todonotes}
\else
    \usepackage[textsize=tiny,textwidth=2cm,shadow,loadshadowlibrary]{todonotes}
\fi
\usepackage{booktabs}
\newcommand{\lgcom}[2][]{\ifthenelse{\equal{#1}{journal}}{\relax}{\todo[color=green!70!blue!60]{LG: #2}}} 
\newcommand{\lgcomI}[2]{\todo[color=green!70!blue!60,inline,caption={#1},size=\small]{LG: #2}} 
\newcommand{\jscom}[1]{\todo[color=blue!50!white]{JS: #1}} 

\ifreview
    \newcommand{\BigPicture}[2][0]{#2}
\else
    
    \newcommand{\BigPicture}[2][0]{#2}
\fi

\title{Are System Optimal Dynamic Flows Implementable by Tolls?}
    \author{Lukas Graf, Tobias Harks and Julian Schwarz} 
\affil{\small University of Passau, Faculty of Computer Science and Mathematics, 94032 Passau\\
    \href{mailto:julian.schwarz@uni-passau.de}{\{\texttt{lukas.graf,tobias.harks,julian.schwarz\}@uni-passau.de}}}

\begin{document}

\maketitle
 \begin{abstract} 
 
 A seminal result of [Fleischer et al.~\cite{Fleischer04} and Karakostas and Kolliopulos~\cite{Karakostas04}, FOCS 2004] states that system optimal multi-commodity static network
flows are always implementable as tolled Wardrop equilibrium flows even if users have heterogeneous value-of-time sensitivities. Their proof uses LP-duality to characterize the general implementability of network flows by tolls.
For the much more complex setting of \emph{dynamic flows},  [Graf et al.~\cite{GHS24TollSODA}, SODA 2025] identified  necessary and sufficient conditions for a dynamic $\source$-$\dest$ flow to be implementable as a tolled dynamic equilibrium.  They used the machinery of (infinite-dimensional) strong duality
to obtain their characterizations. Their work, however, does not answer the question of whether system optimal dynamic network flows are implementable by tolls. 

We consider this question for a general dynamic flow model involving multiple commodities with individual source-\sink pairs, fixed inflow rates and heterogeneous valuations of travel time and money spent.
We present both a positive and a, perhaps surprising, negative result: 
For the negative result, we provide a network with multiple source and \sink pairs in which under the Vickrey queuing model no system optimal flow is implementable -- even if all users  value travel times and spent money the same. 
Our counter-example even shows that the ratio of the achievable equilibrium travel times by using tolls and of the system optimal travel times can be unbounded. 
For the single-source, single-\sink case, we show that if the traversal time functions are suitably well-behaved (as is the case, for example, in the Vickrey queuing model),   
any system optimal flow is implementable. 
  \end{abstract}

\clearpage

\section{Introduction}   

Toll pricing has a long history in the 
economics literature as a mechanism to implement traffic equilibria with desirable properties (see Pigou~\cite{Pigou20} and Knight~\cite{Knight1924} for early works on the subject). In the context of traffic flows it has so far been mainly applied 
to static traffic assignment problems aka the static Wardrop flow model. 
For this model, the results by Cole, Dodis and Roughgarden~\cite{ColeDR03}, Fleischer, Jain and Mahdian~\cite{Fleischer04}, Karakostas and Kollioupous~\cite{Karakostas04} and Yang and Yuang~\cite{Yang04}
give complete characterizations of static edge flows that are implementable as toll-based Wardrop equilibria even if the users have heterogeneous preferences over travel times and paid tolls. In all of these works, the implementability of \emph{system optimal} (static) flows is of particular concern, that is, flows that minimize the total travel time experienced by the users. 
In the most general form, Fleischer et al.~\cite{Fleischer04},  
showed that the implementability of a flow via 
anonymous tolls is characterized by the ``minimality'' of a flow.
A flow $g$  is minimally feasible if there does not exist a flow $g'$ with $g'_\arc \leq g_\arc$ for all edges $\arc$ of the network with the inequality being strict for at least one edge. Under a monotonicity assumption of travel time functions it follows that every static system optimal multi-commodity  flow is minimally feasible and, thus, implementable. This universally positive result even holds for \emph{non-separable} and \emph{player-specific} travel time functions (Harks and Schwarz~\cite[Corollary~6.2]{HarksS23}). 

For dynamic network flows, much less is known regarding the implementability of dynamic system optimal flows using anonymous dynamic tolls. 
 Recently, Graf, Harks and Schwarz~\cite{GHS24TollSODA} derived 
necessary and sufficient conditions for the implementability of dynamic $\source$-$\dest$ flows by means of strong duality of an associated  infinite-dimensional optimization problem. This characterization might give hope that the above positive results for static flows might carry over to the dynamic setting, however, 
 it does not answer the question of whether
or not system optimal dynamic network flows are implementable by tolls. This is the topic of this paper.

\subsection{Related Work}
The inefficiency of equilibria as measured by the price of anarchy and price of stability is a central topic in the algorithmic game theory literature. In the realm of the Vickrey flow model, Koch and Skutella~\cite{Koch11} showed that the price of anarchy with respect to the earliest arrival time-objective  is unbounded for dynamic equilibrium flows.
Correa et al.~\cite{CCOPoAforNashFlows} proved that the price of anarchy with respect to the makespan objective is bounded by $e/(e-1)$ under a  monotonicity conjecture.

A well-established mechanism to improve the quality of equilibria is to use dynamic (anonymous) tolls on the edges of the network.
This approach is known as \emph{road pricing} or \emph{congestion pricing} 
 in the traffic engineering literature and there is a vast literature on the subject, see Lombardi et al.~\cite{Lombardi21} and Yang and Huang~\cite{Yang2005} for a bibliographic overview. 
Frascaria and Olver~\cite{FrascariaO22} proved that for the Vickrey queuing model,  $\source$,$\dest$ networks and departure time choice of particles, there exist tolls implementing an earliest arrival flow, that is, a flow that maximizes the total volume of users arriving at the \sink for every point in time. 
The proof relies on an infinite dimensional linear program for which feasible dual variables are derived as prices which  then implement the earliest arrival flow as a dynamic equilibrium. They assume that the underlying earliest arrival flow has no queues and thus, the underlying travel times are \emph{time-independent}. 

Wie and Tobin~\cite{WieTobin1998} considered a dynamic traffic assignment model based on an optimal control formulation. They derived \emph{non-anonymous} 
tolls implementing a system optimum, that is, these tolls depend on the users' \sink and only apply to the case of homogeneous users. Moreover, the Vickrey model does not satisfy their required differentiability conditions for their
used flow model, and, thus, their results only apply to their specific model.

Ma, Ban and Szeto~\cite{Ma2017} considered a so-called double queue formulation (queues at the entry and exit of an edge) and also argued
 that an optimal solution with respect to a slightly different objective (accounting also for emissions) can be implemented by tolls. They argue that a corresponding optimal control formulation can be used to derive dual variables. However, no formal proof of the existence of an optimal control solution, nor
of any regularity condition  is given to support these claims.
Yang and Meng~\cite{Yang1998}  and also Yang and Huang~\cite[Chapter 13]{Yang2005} studied a model where time is discretized leading to a space-time expanded network. This way,  the toll problem is reduced to a static problem and they derived tolls implementing an optimal (discrete-time) flow using the standard 
finite-dimensional linear programming formulations.

Other mechanisms to improve the quality of dynamic equilibrium flows include the design of travel information (cf.~Griesbach et al.\cite{GHKKInformationDesign})
and the design of the network capacities (cf.~Bhaskar et al.~\cite{BhaskarFA15}).
\subsection{Our Results} 
In this work, we consider the question of whether or not dynamic system optimal flows are implementable as tolled dynamic equilibria. We study this question for the heterogeneous user model and a quite general
physical flow model that includes the Vickrey queuing model~\cite{Vickrey69} and the linear edge-delay model~\cite{Friesz93} as special cases. We further consider both, multi-source, multi-\sink instances and the special case of single-source, single-\sink instances as studied before in Graf et al.~\cite{GHS24TollSODA}.
Our main results consist of two complementary theorems: A negative result on implementability in multi-source, multi-\sink networks and a positive result on implementability in single-source, single-\sink networks. 
These results and their proof ideas are outlined in the following: 
\begin{introthm}[Informal]{thm: CounterSysopImpl}
There exists a multi-source, multi-\sink network utilizing the Vickrey queuing model, in which \emph{every dynamic system optimal flow} is not implementable by dynamic tolls, even if users have homogeneous valuations of travel time and money spent. 
\end{introthm}
As a byproduct of the above theorem, we can even show that the ratio of
the achievable equilibrium travel times by using tolls and of the system optimal travel
times can be unbounded. This translates to the 
following result: The \emph{price of stability achievable by tolls in multi-source, multi-\sink network} is unbounded.
This negative result shows for the first time
that there is a real gap between 
the implementability of system optimal flows
for static and dynamic flow models.

The construction of the counter-example  is based on a specific five commodity multi-source, multi-\sink network (cf.~\Cref{fig: CE_SystemOptimumImplementable}). 
The task then is to prove both, a suitable lower bound on the total travel time of \emph{all} implementable flows, and,  an upper bound on the total travel time of system optimal flows. 
For the lower bound, we derive a general necessary condition (\Cref{thm: ImplementableImpliesOptimal}) for the implementability of a flow in a multi-source, multi-\sink network using an associated infinite dimensional linear program. This optimization problem and the resulting necessary condition is a generalization of the ones proposed in~\cite{GHS24TollSODA}
for the single-source, single-\sink case. 
For deriving the upper bound, we identify a particular flow with strictly lower total travel time than the lower bound for implementable flows. Perhaps surprisingly, this 
flow has a commodity which sends particles along an outgoing edge from its corresponding \sink[.] 

Afterwards, by adjusting the edge travel times and commodities' inflow rates suitably, we obtain a whole family of instances for which the price of stability achievable by tolls is unbounded.

Our second main theorem then considers the case of single-source, single-\sink networks.  
For such networks, Graf et al.\ characterized the implementability of a flow  by the property that the flow  does not send positive flow along edges leaving the \sink (cf.~\cite[Thm.~6.2]{GHS24TollSODA}). This condition alone, however, is not enough to guarantee the implementability of system optimal flows,  because it is not clear whether or not sending flow along an outgoing edge from a commodity's \sink could be beneficial for the overall total travel time (as is the case in the multi-source, multi-\sink counter-example). In fact, we show in \Cref{exa: CE_SystemOptimumWithoutMonotonicity} that this effect is still possible while keeping the assumptions of~\cite[Thm.~6.2]{GHS24TollSODA} intact. 
This example, however, crucially relies on a non-monotonic behavior of travel time functions (roughly speaking, higher flow volume on an edge leads to lower travel times). 
Under suitable \emph{monotonicity conditions} regarding the travel time functions (which are, in particular, fulfilled by the Vickrey queuing model  and the linear edge-delay model) 
prove, then, the following positive result:
\begin{introthm}[Informal]{cor: SsSysOptImpl}
    For single-source, single-\sink networks, every dynamic system optimal flow is implementable by dynamic tolls, even if users have heterogeneous valuations of travel time and money spent. 
\end{introthm}
The proof is rather intricate and builds on two key insights that we establish: First, we show that any flow traveling along a cycle that includes the \sink can be rerouted so that the particles remain at the \sink[,] resulting in a strictly better but \emph{infeasible flow} where particles may wait at nodes. Second, we demonstrate that such infeasible flows can always be transformed into feasible flows, that is, flows without waiting at nodes, while preserving or even further reducing the total travel time.

\section{Model}\label{sec: Model} 

We consider multi-source, multi-\sink networks given by a directed graph $G=(\GV,\GA)$ with nodes $V$ and edges $\GA \subseteq V \times V$. 
We use $\edgesFrom{v}$ to refer to the set of edges leaving a node $v$ and $\edgesTo{v}$ for the set of edges entering $v$. Furthermore, 
we denote  for any pair $v_1,v_2 \in \GV$ by $\hat{\Routes}_{v_1,v_2}$ the countable set of (finite) \stwalk[v_1][v_2]s in~$G$. Here, a walk $\wa$ is a tuple 
of edges $\hat{\wa} = (\arc_1,\ldots,\arc_k) \in \hat{\Routes}$ with $\arc_j=(v_j,v_{j+1})\in \GA$ for all $j\in [k]:=\{1,\ldots,k\}$ for some $(v_j)_{j\in[k+1]}\in \GV^{k+1}$. We use $\hat\wa[j] \coloneqq e_j$ to refer to the $j$-th edge on walk~$\hat\wa$. Note, that we explicitly allow walks to contain cycles and include the same edge multiple times. 
We call a walk $\hat c=(\gamma_1,\ldots,\gamma_m)$ a cycle if $\gamma_1 \in \edgesFrom{v}$ and $\gamma_m \in \edgesTo{v}$ for some node  $v \in \GV$. In case that $v= \dest$, we call $\hat c$ a $\dest$-cycle.  Moreover, for a walk $\hat\wa$ and $j\leq |\hat\wa|$, we denote by  $\hat\wa_{\geq j}$  
the sub-walk of $\hat\wa$ starting with $\hat\wa[j]$.

Next, we are given a fixed planning horizon $\hori'=[0,t_f'] \subseteq \R$ during which flow particles can traverse the network. 
For technical reasons, we assume that flow functions are defined on some larger (but still finite) time horizon $\hori=[0,t_f]$ with $t_f > t_f'$ while still only being supported on~$\hori'$. 
Since dynamic flows will be described by Lebesgue-integrable functions on~$\hori$, we equip $\hori$ with its Borel $\sigma$-algebra $\mathcal{B}(\hori)$. We denote by $\sigma$ the Lebesgue measure on $(\hori,\mathcal{B}(\hori))$ 
and by $L(\hori)$ and $L^\infty(\hori)$ the space of ($\sigma$-equivalence classes of)  $\sigma$-integrable and essentially bounded, respectively, real-valued functions over $\hori$  equipped with the standard norm induced topology and the partial order induced by $L_+(\hori)$, respectively, $L_+^\infty(\hori)$, i.e.\ the subsets of nonnegative integrable functions. 
For any countable set~$M$, we denote by $\seql[1][M][L(\hori)]$ the set of 
vectors $(h_m)_{m \in M} \in L(\hori)^M$ whose sum $\sum_{m \in M}h_m \in L(\hori)$ is well-defined and exists, i.e.
\begin{align*}
    \seql[1][M][L(\hori)] &:= \big\{ h \in L(\hori)^{M} \mid \norm{h} := \sum_{m \in M} \norm{h_m} < \infty \big\}.
\end{align*}
 This defines again a Banach space (cf.~\cite[Section 16.11]{guide2006infinite}) whose topological dual is 
 \begin{align*}
     \seql[\infty][M][L^\infty(\hori)]:= \big\{ f \in L^\infty(\hori)^M \mid \norm{f}_\infty := \sup_{m \in M} \norm{f_m}_\infty < \infty \big\}. 
\end{align*} 
We denote the bilinear form between this dual pair  by $\dup{f}{h} \coloneqq \sum_{m \in M}\int_\hori f_m\cdot h_m\di\sigma$ for $f\in L^\infty(\hori)^M,h\in L(\hori)^M$. Here, we use $\int_\hori f\di\sigma$ to denote the integral of $f$ over~$\hori$ with respect to the Lebesgue measure~$\sigma$.\footnote{We use this notation instead of writing $\int_{\hori} f(t)\di t$ to stay consistent with the proofs of some of the more technical lemmas where we also have to consider integrals with respect to other measures.}
Analogously, we define $\seql[1][M][L(\hori)^\GA]$ where we use $\norm{\g} := \sum_{\arc \in \GA}\norm{\g_\arc}$ for $\g \in L(\hori)^\GA$.

Finally, we have a finite set of commodities/populations~$I$ where each commodity comes with its own value-of-time (VoT) parameter $\gamma_i > 0$ denoting this commodity's conversion factor of (travel) time to  money. Each commodity $i \in I$ has a fixed network inflow rate $\inflow_i\in L_+(\hori)$ with $\inflow_i \neq 0$ specifying for (almost) every point in time at what rate particles of that commodity enter the network at the commodity specific source node~$\source_i$. The particles then aim to traverse the network by choosing an \stwalk[\source_i][\dest_i] where $\dest_i$ is the commodity specific \sink node of which we assume that it is connected to $\source_i$. We denote by $\Routes \coloneqq \cup_{i \in I} \Routes_i \coloneqq \bigcup_{i \in I}\hat{\Routes}_{\source_i,\dest_i} \times\{i\}$ the set of walk-commodity pairs. By some abuse of notation we will also refer to elements $\wa=(\hat\wa,i) \in \Routes$ as walks and use $\wa[j] \coloneqq \hat\wa[j]$ to denote the $j$-th edge of~$\wa$, write $v\in \wa$ and $\arc \in \wa$ to say that there exists  some $j \in [k]$ and $\hat{v},v'\in \GV$ with $(\hat{v},v') = \wa[j]$ and $v \in \{\hat{v},v'\}$, respectively $\arc=\wa[j]$, and use $\abs{\wa} \in \N_0$ for the length (=number of edges) of~$\wa$.

\subsection{Dynamic Flows}

The main concept underlying dynamic flows are the traversal time functions:  

\myparagraph{Traversal time functions} Within our model any vector of edge inflow rates $\g\in L_+(\hori)^\GA$ induces    corresponding nonnegative  \emph{edge traversal time} functions
$\trav_\arc(\g_\arc,\cdot):\hori \to \R_+,\arc \in \GA$ 
 with 
$\trav_\arc(\g_\arc,t)$ denoting the time needed to traverse $\arc$ 
when entering the latter at time~$t$. 
For the sake of readability, we will often just write $\trav_\arc(\g,\cdot)$ instead of $\trav_\arc(\g_\arc,\cdot)$.  
To any such edge traversal time function, we also define two related functions: 
Firstly, the \emph{edge exit time} function $\exit
_\arc(\g,t):= t + \trav_\arc(\g,t)$ denoting the time a particle exits edge~$\arc$ when entering at~$t$. 
Secondly, the \emph{edge arrival time} function $\arr_{\wa,j}(\g,\cdot)$ denoting the time a particle arrives at the tail of the $j$-th edge of some walk~$\wa$ when entering~$\wa$ at time~$t$. More precisely, for an arbitrary walk $\wa$ we define $\arr_{\wa,1}(\g,\cdot):= \id$ and then, recursively, $\arr_{\wa,j}(\g,\cdot):= \exit_{\wa[{j-1}]}(\g,\cdot) \circ \arr_{\wa,j-1}(\g,\cdot)$ for $j\in\{2,\ldots,|\wa|\}$. Additionally, we define $\arr_{\wa,|\wa|+1}(\g,\cdot):= \exit_{\wa[|\wa|]}(\g,\cdot)\circ\arr_{\wa,|\wa|}(\g,\cdot)$
 denoting the arrival time at the \sink[.]

With this, we can now formally describe dynamic flows. We will use two types of these flows:
Aggregated edge flows and multi-commodity walk flows: 

\myparagraph{Multi-Commodity Walk Flows}
A  \emph{walk flow} or walk-inflow function is a vector $h\in L_+(  \hori)^{{\Routes}}$
with $h_{{\wa}}(t)$ representing the walk inflow rate at time $t\in \hori$ into the walk ${\wa}\in  {\Routes}$. 
 In this regard, for a   walk  flow   $h\in L_+(  \hori)^\Routes$,   
 $h_{(\hat{\wa},i)}(t)$ represents the walk inflow rate of  users of commodity $i$ at time $t\in \hori$ into the walk $(\hat{\wa},i)\in \Routes_i$.  

\myparagraph{Aggregated Edge Flows} 
An (aggregated) \emph{edge flow} is a vector $\g \in L_+(\hori)^\GA$, where $\g_\arc(t)$ denotes the rate at which particles (of all commodities combined) enter edge~$\arc$ at time~$t$. For any such edge flow $\g \in L_+(\hori)^\GA$ and edge~$\arc \in \GA$ we define the \emph{cumulative edge inflow and outflow} $\G^+_\arc,\G^-_\arc:\hori\to \R$  via 
\begin{align*}
    \G^+_\arc (t) := \int_{[0,t]}\g_\arc \di\sigma \text{ and }     \G^-_\arc (t) := \int_{\exit_\arc(\g_\arc,\cdot)^{-1}([0,t])}\g_\arc \di\sigma.
\end{align*}
Furthermore, we define for any node $v \in \GV$ the \emph{cumulative node balance} $\opp_v \G:\hori\to \R$   via
\begin{align}\label{eq: DefCumuFlowBala}
    \opp_v \G(t) := \sum_{\arc \in \edgesFrom{v}} \G^+_\arc (t) -  \sum_{\arc \in \edgesTo{v}} \G^-_\arc (t).
\end{align}
 
We then connect the two types of flow by saying that an edge flow~$\g$ is \emph{induceable} by  a walk flow $h \in L_+(\hori)^\Routes$  if: 
  \begin{align}\label{eq: DefEdgeFlow}
    \int_0^t \g_\arc\di\sigma = \sum_{\wa \in \Routes}\sum_{j: \wa[j] = \arc}\int_{\arr_{\wa,j}(\g,\cdot)^{-1}([0,t])} h_\wa\di\sigma \text{  for all  $\arc \in \GA$ and $t \in \hori$.} 
\end{align}  
We denote by 
\begin{align*}
    \wir:= \Big\{ h\in L_+(  \hori)^\Routes \Big\vert\,\sum_{\wa\in \Routes_i} h_\wa=\inflow_i,i\in I,   \text{\eqref{eq: DefEdgeFlow} has a solution supported on $\hori'$.} \Big\} \subseteq \seql
\end{align*}
the set of \emph{admissible walk flows}. 
We assume that we have a network loading operator $\ell:\wir \to L_+(\hori)^\GA, h\mapsto \g$ sending each $h \in \wir$ to a solution $\g$  of \eqref{eq: DefEdgeFlow} supported on $\hori'$. We then say that $h \in \wir$ induces $\g \in L_+(\hori)^\GA$ if $\ell(h) = \g$  holds and call $\ell(\wir) := \{\ell(h) \in L_+(\hori)^\GA \mid h \in \wir\}$ the set of induced edge flows.  

With respect to these flows, 
we require the following assumptions on the travel time function: 
\begin{assumption}\label{ass: TravelGeneral} 
We assume that there exists a constant $\tmin>0$ such that 
for all $\g \in L_+(\hori)^\GA$ and $\arc \in \GA$
\begin{thmparts}
    \item $\trav_\arc(\g,\cdot):\hori \to \R_+$ is absolutely continuous and adheres to the first-in first-out principle (FIFO), that is, $\exit_\arc(\g,\cdot)$ is nondecreasing.\label[thmpart]{ass: TravelGeneral: ConFIFO}  
     \item the travel time is  lower bounded on $\hori'$ by $\tmin>0$, i.e.~$\trav(\g,t)\geq \tmin$ for all  times $t \in \hori'$. \label[thmpart]{ass: TravelGeneral: tmin} 
        \item the corresponding edge outflow rate $\g_\arc^-$ exists, i.e.~there exists $\g_\arc^- \in L_+(\hori)$ such that for all $t \in \hori$: 
    \begin{align}\label{eq: TravelGeneral: Inflow}
       \int_{\exit_\arc(\g_\arc,\cdot)^{-1}([0,t])} \g_\arc \di\sigma =  \int_{[0,t]}\g_\arc^- \di\sigma.
    \end{align} \label[thmpart]{ass: TravelGeneral: Inflow} 
\end{thmparts}
We additionally assume that there exists a contant $\tmax >0$ such that 
for all $\g \in \ell(\wir)$
\begin{thmparts}[resume]
     \item the travel time is  upper bounded on $\hori'$ by  $\tmax>0$, i.e.~$\trav(\g,t)\leq \tmax$ for all  times $t \in \hori'$. \label[thmpart]{ass: TravelGeneral: tmax} 

\end{thmparts}
\end{assumption}
Note that \labelcref{ass: TravelGeneral: ConFIFO} also implies that both $\exit_\arc(\g,\cdot)$ and $\arr_{\wa,j}(\g,\cdot)$ are absolutely continuous as well (\cite[Exercise 5.8.59]{Bogachev2007I}). 
Due to technical reasons, we also require that the travel times  lead  to arrival time functions whose range is contained in $\hori$:
\begin{assumption}\label{ass: AMapsHToH}
    For all $\g \in \ell(\wir)$, $\wa \in \Routes$, $j \in[\abs{\wa}+1]$   we assume that  $\arr_{\wa,j}(\g,\cdot)(\hori) \subseteq \hori $. 
\end{assumption}
Given arbitrary travel times satisfying  \Cref{ass: TravelGeneral}, we can achieve this additional property by suitably (absolutely continuously) decrease $\trav_\arc$ between $t_f'$ and ${t}_f$. Remark that this does not have an impact on the set of induced edge flows as these are required to be supported only on $\hori'=[0,t_f']$.

A well-studied dynamic flow model  that falls within  our model (and, in particular, satisfies \Cref{ass: TravelGeneral}) is the Vickrey queuing model  which we introduce next. 
\begin{definition}\label{ex:VickreyModel}
    In   the \emph{Vickrey queuing model}, each edge~$\arc \in \GA$ comes with a free flow travel time~$\tau_\arc > 0$ and a service rate~$\nu_\arc > 0$. The traversal time function~$\trav_\arc$ is then defined as the (unique) solution to a system of equations in terms of the corresponding edges flows $\g \in L_+(\hori)^\GA$:
         \begin{align}\label{eq: Vick}
              \trav_\arc(\g,t) = \tau_\arc + \frac{q_\arc(\g,t)}{\nu_e} \quad\text{ and }\quad q_\arc(\g, t) = \int_0^t \g_\arc\di\sigma-\int_{\exit_\arc(\g,\cdot)^{-1}([0,t+\tau_\arc])} g_\arc \di\sigma
         \end{align}
    together with the conditions that the queue is always non-negative and  that the derivative of $t \mapsto \int_{\exit_\arc(\g,\cdot)^{-1}([0,t])} g_\arc \di\sigma$ (i.e.\ the outflow rate of edge~$\arc$) is bounded by $\nu_\arc$ almost everywhere.\footnote{Note that, by \cite[Proposition~3.19e)]{GrafThesis}, this definition is equivalent to the more common definition of Vickrey flows in terms of in- and outflow rates.}
    Here, $\q_\arc(\g,t)$ denotes the flow volume in the queue of edge~$\arc$ at time~$t$ and, for~\eqref{eq: Vick}, we extend $g$ by~$0$ outside of~$\hori$.
\end{definition}

\myparagraph{Parameterized Network Loadings} 
In \cite{GHS24FD}, Graf~et al.\ introduced the concept of parametrized ($u$-based) network flows which describes for a given edge flow~$u$
how particles of a different walk-inflow~$h$  would hypothetically propagate throughout the network under the fixed travel times induced by~$u$. 
More precisely, given a travel time function $\trav$  and a vector $u\in L_+(\hori)^\GA$, we can define the corresponding parameterized travel time function $\trav^u:L_+(\hori)^\GA \times \hori \to \R_+$ via $\trav^{u}(\g,t):\equiv \trav(u,t)$ for all $\g \in  L_+(\hori)^\GA$ and $t\in \R_+$. 
We denote by $\ell^u$ the resulting network loading operator and by $\wir^u$  the set of walk inflows inducing an edge flow. 
Moreover, 
adopting the terminology from~\cite{GHS24FD}, we say 
that the walk inflow $h_\wa$ into a walk $\wa$ under the \emph{fixed} traversal time functions $\trav^u$ induces the (parameterized) flow $\f^\wa_j \in L(\hori)$ on the $j$-th edge of $\wa$ if  $\f^\wa_j$ fulfills:
\begin{align}\label{eq: DefUEdgeFlow}
    \int_0^t \f^\wa_j\di\sigma =  \int_{\arr_{\wa,j}(u,\cdot)^{-1}([0,t])}h_\wa\di\sigma \text{ for all } t \in \hori. 
\end{align} 
In the above situation \eqref{eq: DefUEdgeFlow}, we write $\ell^u_{\wa,j}(h_\wa):= \f^\wa_j$ for any $j \leq \abs{\wa}+1$ and say that the induced parameterized edge flow $\ell_{\wa,j}^u(h_\wa)$ exists. Here, for $j = \abs{\wa} +1$, the latter is to be interpreted as the network outflow rate of the induced flow. 
If, for some edge~$\arc \in \GA$, the induced (parameterized) edge flow $\ell_{\wa,j}(h_\wa)$ exists for all $j\leq \abs{\wa}$ with $\wa[j] = \arc$, we can define the total induced (parameterized) flow of
$h_\wa$ under $\trav^u$ on this edge via $\ell_{\wa,\arc}^u(h_\wa) := \sum_{j:\wa[j]  =\arc}\ell^u_{\wa,j}(h_\wa)$. 
Analogously, if $\ell_{\wa,\arc}^u(h_\wa)$ exists for all $\arc \in \wa$, we can define the induced (parameterized) flow on the whole  network via the vector   
$\ell^u_\wa(h_\wa):=(\ell^u_{\wa,\arc}(h_\wa))_{\arc \in \GA} \in L(\hori)^\GA$. 
We require the following assumption regarding the uniqueness of the network loading: 
\begin{assumption}\label{ass: LuImplL}
 For all $u \in L_+(\hori)^\GA$ supported on $\hori'$, the implication $\ell^u(h) = u \implies \ell(h) = u$ holds for all $h \in \wir^u$. 
\end{assumption}
The above assumption states that whenever a walk flow induces~$u$ under~$u$, then $u$ is also the actual network loading of~$h$. This is fulfilled whenever there is at most one way to define the network loading, i.e.\ if for any $h\in L_+(\hori)^\Routes$ with $\sum_{\wa \in \Routes_i}h_\wa = \inflow_i,i \in I$ there exist at most one $\g \in L_+(\hori)^\GA$ such that \eqref{eq: DefEdgeFlow}
is fulfilled. In this regard, note that Meunier and Wagner~\cite{MeunierW10} showed that this is the case 
under several assumptions regarding the travel time function $\trav$ which do not conflict with our assumptions. 
In particular, Meunier and Wagner~\cite{MeunierW10} show that  the Vickrey queuing model satisfies them.

Throughout this paper, we require several structural insights derived in \cite{GHS24FD,GHS24TollSODA} 
concerning the structure of network loadings and the implementability of flows. 
For the convenience of the reader, we collected the relevant definitions and statements in \Cref{sec:AppuBasedNetworkLoadings,sec:AppTolls}.

\subsection{Toll-Based Implementability and System Optima}\label{sec:CostBalancingTolls}

In addition to the physical aspects of dynamic flows described in the previous section, there is also a behavioural aspect: We think of a dynamic flow as being made up of individual flow particles that each individually and selfishly want to minimize their overall cost which is a weighted sum of travel time and tolls. A flow wherein every particle achieves this goal (with respect to the traversal times induced by that flow) is a dynamic equilibrium.

In order to formalize this, we denote by $\Psi_{(\hat{\wa},i)}(\g,t)$ the weighted total travel time a (hypothetical) particle belonging to commodity~$i$ and entering walk~$\wa = (\hat{\wa},i)$ at time~$t$ would experience under the edge flow $\g \in \ell(\wir)$, i.e.
\begin{align*}
    \Psi_{(\hat{\wa},i)}(\g,t) := \vot_i\cdot \sum_{j \leq|{\wa}|} \trav_{{\wa}[j]}(\g,\arr_{{\wa},j}(\g,t)) = \vot_i\cdot \big(\arr_{\wa,\abs{\wa}+1}(\g,t)-t\big).
\end{align*}

Furthermore, given  a vector-valued  function $\prices: \hori \to \R^\GA_+$ 
 associating with each edge $\arc \in \GA$ and every entry time $t\in \hori$ nonnegative and finite tolls $\prices_\arc(t)$, 
 we define the resulting total toll along walk $\wa=(\hat{\wa},i)$ for the entry time~$t$ under~$\g$ via $\Pf^{{\prices}}_{(\hat{\wa},i)}(\g,t) := \sum_{j \leq|{\wa}|}  \prices_{{\wa}[j]}(\arr_{{\wa},j}(\g,t))$. Note that the edge tolls~ $\prices_\arc(t)$ are anonymous in the sense that any particle entering edge~$\arc$ at time~$t$ must pay the same toll~$\prices_\arc(t)$, regardless of its identity (i.e.\ the \sink[,] VoT or chosen route).

\begin{definition}
    For any  toll function $\prices: \hori \to \R^\GA_+$,  a \emph{$\prices$-dynamic user equilibrium} ($\prices$-DUE) 
is a walk inflow rate function $h\in \wir$ with corresponding edge flow $\g = \ell(h)$ which satisfies Wardrop's first principle, that is for all  $\wa \in \Routes_i,i\in I$ and almost all $t \in \hori$ it satisfies:
\begin{align*}
    h_\wa(t)>0 \implies \Psi_\wa(\g,t) + \Pf^\prices_\wa(\g,t) \leq \Psi_{\wa'}(\g,t) + \Pf^\prices_{\wa'}(\g,t)\text{ for all } \wa'\in \Routes_i.
\end{align*}
\end{definition}

Note that for $\prices = 0$, the notion of $\prices$-DUE coincides with the classical notion of a DUE, cf.~\cite{Friesz93,ZhuM00}.
\begin{definition}[Implementability]
A vector $\g \in L(\hori)^\GA$ is \emph{implementable} via tolls $\prices : \hori \to \R^\GA_+$ and walk inflow rates $h \in \wir$, 
if  $h$ is a  $\prices$-DUE which induces $\g$ with  bounded total costs $\dup{\prices}{\g}< \infty$. 
We say that $\g$ is \emph{implementable} if there exist tolls $\prices$ and walk inflow rates $h\in \wir$ such that $\g$ is implementable via $(\prices,h)$.
 \end{definition}

Flows of particular interest regarding implementability are \emph{system optima}, that is, flows that minimize the total travel time among all induced flows.

\begin{definition}[System Optimum]
    We call a vector $\gop\in L_+(\hori)^\GA$ \sysop if it minimizes the total travel time among all induced flows (supported on $\hori'$), i.e.\ it is an optimal solution to the following optimization problem: 
 \begin{align}
    \inf_{\g} \; & \dup{\trav(\g,\cdot)}{\g }\tag{$\mathrm{P}^{\mathrm{sys}}$} \label{opt: System}  \\
    &\textstyle\g \in \ell(\wir),\nonumber
\end{align}
where $\ell(\wir) := \{ \ell (h ) \in L_+(\hori)^\GA \mid   h \in \wir \}$ is the set of induced (aggregated) edge flows. 
\end{definition}

\section{Non-Implementability for Multi-Source, Multi-\Sink Networks}\label{sec: NonImpl}

In this section, we consider the general case of multi-source, multi-\sink networks and show 
that not every \sysop edge flow~$\gop$ is implementable, even for well-behaved flow propagation models like the Vickrey model 
and a common VoT-parameter for all commodities. 

In order to show this, we first derive a necessary condition for implementability in multi-source, multi-\sink networks generalizing one direction of the characterization given in 
\Tollref{thm: main:ImplementableImpliesOptimal} for  single-source, single-\sink networks. 
To this end, we define the \emph{master problem} \wrt a vector $u \in \ell(\wir)$: 
\begin{align}
    \inf_{h} \; & \sum_{\wa \in \Routes}\dup{\Psi_\wa (u,\cdot)}{h_\wa }\tag{${\mathrm{P}}(u)$}\label{opt: Master}  \\
    \text{s.t.: } &  \ell^u(h) \leq u \label{ineq: Master}\\
    &(h_\wa)_{\wa \in \Routes} \in  \wir^u.\nonumber
\end{align} 
Note that \ref{opt: Master} differs from the definition of the master problem proposed in \cite{GHS24TollSODA} since we 
allow for multi-commodity walk inflow rates $h \in \wir$.  

\begin{theorem}\label{thm: ImplementableImpliesOptimal}
     If  $u$ is implementable via $(\prices,h)$ for some walk inflow rates $h$ and nonnegative edge costs $\prices:\hori\to \R^\GA_+$,  
     then~\ref{opt: Master} admits an optimal solution with tight inequality. 
\end{theorem}
\begin{proof}
     Let $u$ be implementable via tolls $\prices:\hori\to\R^\GA$  and $h \in \wir$, i.e.\ $h$ induces $u$ (under $\ell$) with $\dup{\prices}{u}<\infty$ and $h$ is a $\prices$-DUE. Then, $h$ also induces $u$ under $\ell^u$, i.e.~$h \in \wir^u$ and $\ell^u(h) = u$.  
    Optimality of~$h$ for~\ref{opt: Master} is then a direct consequence of the following claim: 
 \begin{claim}\label{claim:  ImplementableImpliesOptimal}
     For all $\tilde{h} \in \wir^u$ with corresponding aggregated edge flow $ \ell^u(\tilde{h})$, the following inequality holds (with the right hand side possibly being equal to $\infty$): 
     \begin{align}\label{eq: ineqOpth}
         \sum_{\wa \in \Routes} \dup{\Psi_\wa(u,\cdot)}{h_\wa} + \dup{\prices}{u}\leq \sum_{\wa \in \Routes} \dup{\Psi_\wa(u,\cdot)}{\tilde{h}_\wa} + \dup{\prices}{\ell^u(\tilde{h})}.
     \end{align}
 \end{claim}
  \begin{proofClaim} 
 We start by observing that \Cref{lem: aggCostsVSwalkCosts} lets us rewrite the desired inequality~\eqref{eq: ineqOpth} as:
    \begin{align*}
         \sum_{\wa \in \Routes} \dup{\Psi_\wa(u,\cdot) + \Pf^{\prices}_\wa(u,\cdot)}{h_\wa}  \leq \sum_{\wa \in \Routes} \dup{\Psi_\wa(u,\cdot)+\Pf^{\prices}_\wa(u,\cdot)}{\tilde{h}_\wa} .
     \end{align*}
 Now, assume for the sake of a contradiction that there 
 exists $\tilde{h} \in \wir^u$ for which this inequality does not hold. 
Then, there has to exist at least one commodity $i \in I$ with
    \begin{align*}
         \sum_{\wa \in \Routes_i} \dup{\Psi_\wa(u,\cdot) + \Pf^{\prices}_\wa(u,\cdot)}{h_\wa} > \sum_{\wa \in \Routes_i} \dup{\Psi_\wa(u,\cdot)+\Pf^{\prices}_\wa(u,\cdot)}{\tilde{h}_\wa}.
     \end{align*}
     Since $\sum_{\wa \in \Routes_i}h_\wa = \sum_{\wa \in \Routes_i}\tilde{h}_\wa$, there 
     have to exist two walks $\wa,\tilde{\wa} \in \Routes_i$ and a set $\mathfrak T \in \mathcal{B}(\hori)$ with $\sigma(\mathfrak T)>0$ such that $\Psi_\wa(u,t) + \Pf^{\prices}_\wa(u,t)>\Psi_{\tilde{\wa}}(u,t) + \Pf^{\prices}_{\tilde{\wa}}(u,t)$ and $h_\wa(t) > \tilde{h}_\wa(t)\geq 0$ for a.e.\ $t \in \mathfrak T$. 
This, however, contradicts the fact that $h$ is a $\prices$-DUE. 
 \end{proofClaim}
By
subtracting the constant term $\dup{\prices}{u}< \infty$ from both sides in \eqref{eq: ineqOpth} and taking the infimum over all $\tilde h \in \wir^u$, we get 
\begin{align*} 
         \sum_{\wa \in \Routes} \dup{\Psi_\wa(u,\cdot)}{h_\wa} 
            &\leq \inf_{\tilde{h} \in \wir^u} \sum_{\wa \in \Routes} \dup{\Psi_\wa(u,\cdot)}{\tilde{h}_\wa} + \dup{\prices}{\ell^u(\tilde{h})-u}\\ 
            &\leq   \inf_{\substack{\tilde{h} \in \wir^u:\\\ell^u(\tilde{h})\leq u}}  \sum_{\wa \in \Routes} \dup{\Psi_\wa(u,\cdot)}{\tilde{h}_\wa} + \dup{\prices}{\ell^u(\tilde{h})-u}\\ 
            &\leq   \inf_{\substack{\tilde{h} \in \wir^u:\\\ell^u(\tilde{h})\leq u}} \sum_{\wa \in \Routes} \dup{\Psi_\wa(u,\cdot)}{\tilde{h}_\wa},
     \end{align*}
     where we used for the last inequality that $\prices:\hori \to \R_+$ are nonnegative and hence $\dup{\prices}{\ell^u(\tilde{h})-u} \leq 0$. 
    Thus, $h$ is optimal for \ref{opt: Master}, which is exactly what we wanted to show. 
\end{proof}

In addition to the above insight, we also need some further notation regarding multi-commodity flows. 
To any $h \in \wir$, we can associate a \emph{multi-commodity edge flow} $(\g^i)_{i \in I} \in (L_+(\hori)^\GA)^I$ 
that partitions the induced edge flow $\ell(h) =: \g$ into a separate edge flow $\g^i$ for every commodity $\g^i$. 

With this notation and  \Cref{thm: ImplementableImpliesOptimal} at hand, we can show the promised statement of    not every \sysop $\gop$ being implementable. To do so, we will show that in the network depicted in \Cref{fig: CE_SystemOptimumImplementable} (for suitable parameters $M$ and $\varepsilon$), any implementable flow 
can not be system optimal. 

\begin{theorem}\label{thm: CounterSysopImpl}
    For the  5-commodity network  depicted in \Cref{fig: CE_SystemOptimumImplementable} with the Vickrey queuing model, there exist values $M \in \Rnn$ such that for every $\varepsilon\in (0,1]$ all system optimal flows are not implementable.
\end{theorem}

\begin{figure}
    \centering
    \BigPicture[1]{
        \begin{adjustbox}{max width=\textwidth}
        \def\CESOfigversion{network}
            \input{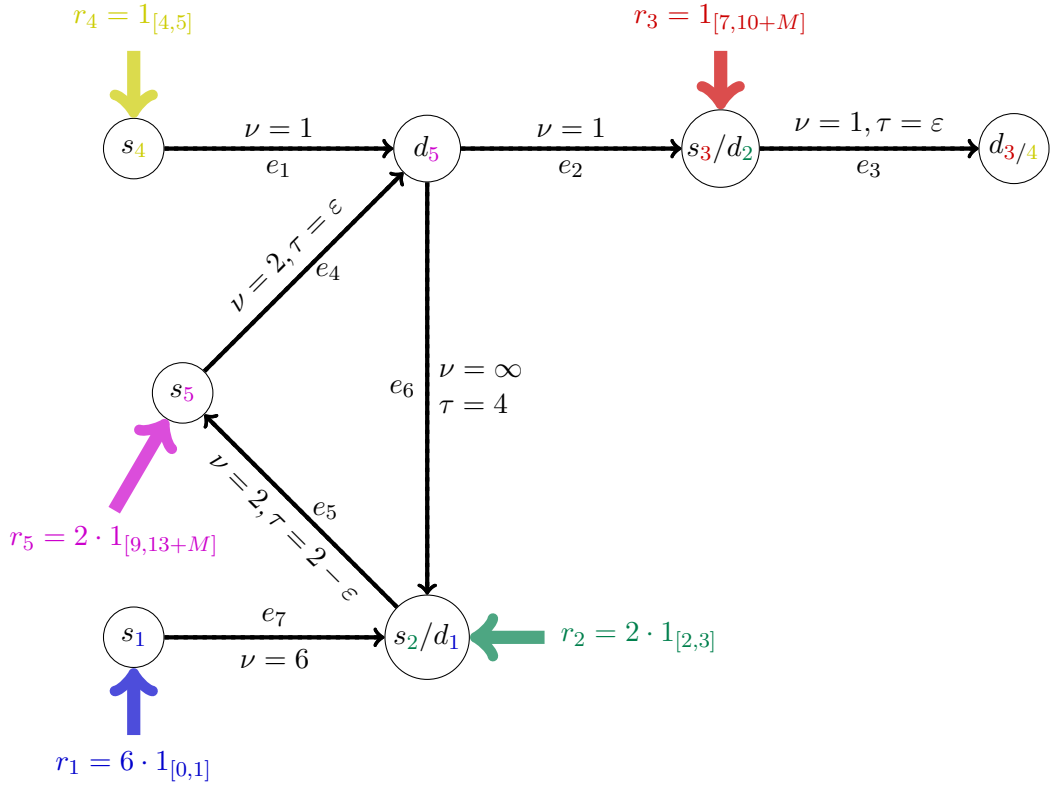}
        \end{adjustbox}
    }
    \caption{A 5-commodity network using the Vickrey model where the system optimum is not implementable (see \Cref{thm: CounterSysopImpl}). All free flow travel times not explicitly given (i.e.\ the ones on edges~$\arc_1$, $\arc_2$ and $\arc_7$) are~$1$. Capacities are as given on the edges.}
    \label{fig: CE_SystemOptimumImplementable}
\end{figure}

\begin{proof}
    We consider the network depicted in \Cref{fig: CE_SystemOptimumImplementable} with the Vickrey queuing model, the corresponding free flow travel times 
    $\tau_{\arc_1} = \tau_{\arc_2} = \tau_{\arc_7} = 1$, $\tau_{\arc_3} = \tau_{\arc_4} = \varepsilon$, $\tau_{\arc_5} = 2-\varepsilon$ and $\tau_{\arc_6} = 4$ as well the corresponding edge capacities $\nu_{\arc_1} = \nu_{\arc_2} = \nu_{\arc_3} = 1$, $\nu_{\arc_4}=\nu_{\arc_5} = 2$, $\nu_{\arc_6} = \infty$ and $\nu_{\arc_7} = 6$. There are $5$ commodities with identical VoT parameters $\gamma_i = 1,i \in I$, separate source, \sink[-]pairs $\source_i,\dest_i, i \in I$ and the respective network inflow rates $\inflow_\colNo{B} = 6\cdot 1_{[0,1]}$, $\inflow_\colNo{C} = 2\cdot 1_{[2,3]}$, $\inflow_\colNo{D} = 1\cdot 1_{[7,10+\M]}$, $\inflow_\colNo{E} = 1\cdot 1_{[4,5]}$ and $\inflow_\colNo{F} = 2\cdot1_{[9,13+M]}$. 
    Here, $\varepsilon \in (0,1]$ and $\M \in \R_{>0}$ are parameters for which we will show that (for any choice of $\varepsilon$) large enough values for $\M$ will lead to an example in which all \sysop $u$ are not implementable. Note that, for simplicity, we only describe the travel times on the flow carrying time horizon~$\hori'$ here as, during the following proof, we never have to consider travel times outside that time horizon anyway. In order to actually extend the travel time functions to the (suitably large chosen) full time horizon~$\hori$ one would apply the transformation described after \Cref{ass: AMapsHToH}.

    We will now show that in this network (for suitable choices of~$M$) no system optimal flow is implementable. We omit here most of the more technical computations and instead only highlight the main proof steps (in form of claims) and give high-level explanations for their correctness. The full formal proofs can be found in \Cref{app:ProofsOfCounterSysopImplClaims}.
    
    We start with two general properties of all flows in this network: 
    \begin{claim} \label{claim: travarc}
        For any $u$ induced by some $h \in \wir$ (under $\ell$), the following statements hold: 
        \begin{thmparts}
            \item \label[thmpart]{claim: travarc: 4>eps+infl} We have $\trav_{\arc_4}(u,t)\geq \varepsilon+\frac{1}{2}\int_9^t u^-_{\arc_5}\di\sigma$ for all $t \in [9,13+M]$. 
            \item \label[thmpart]{claim: travarc: 3>eps+infl} We have $\trav_{\arc_3}(u,t)\geq \varepsilon+\int_7^t u^{\colNo{E},-}_{\arc_2}\di\sigma$ for all $t \in [7,10+M]$. 
            \item The queue on $\arc_2$ never exceeds a size of~$3$. \label[thmpart]{claim: travarc: 2QueueUpperBound}
            \item $D_{\arc_4}(u,t) = \varepsilon$ for all $t \in [0,9]$. \label[thmpart]{claim: travarc: 4NoQueue}
            \item $\trav_{\arc_1}(u,t) = 1$, $\trav_{\arc_7}(u,t) = 1$ and $\trav_{\arc_6}(u,t)=4$ for all $t \in \hori'$. \label[thmpart]{claim: travarc: 1NoQueue}\label[thmpart]{claim: travarc: 7NoQueue}\label[thmpart]{claim: travarc: 6NoQueue}
        \end{thmparts}
    \end{claim}

    Here, \ref{claim: travarc: 4>eps+infl} holds since the \colName{F} flow (commodity~\colNo{F}) which enters the network during $[9,13+M]$ has to first use edge~$\arc_4$ and on its own already uses up all its capacity. Hence, any additional inflow (arriving via edge~$\arc_5$) creates a queue on edge~$\arc_4$, which then cannot deplete before the network inflow of \colName{F} flow stops. The reason for \ref{claim: travarc: 3>eps+infl} is completely analogous as here the \colName{D} flow (commodity~\colNo{D}) already uses up all of edge~$\arc_3$'s capacity during $[7,10+M]$. For \ref{claim: travarc: 2QueueUpperBound} one observes that only \colName{C} and \colName{E} flow (commodities~\colNo{C} and~\colNo{E}) may use edge~$\arc_2$ and their combined flow volume is~$3$. Finally, \ref{claim: travarc: 4NoQueue} and \ref{claim: travarc: 1NoQueue} hold since the inflow rates into these edges can never exceed their respective capacities on the given intervals and, thus, no queues can form.

    Next, we only consider implementable flows. Using \Cref{thm: ImplementableImpliesOptimal}, we first show that in all such flows the particles of commodity~\colNo{C} (\colName{C}) traverse edge~$\arc_5$ without experiencing any queuing delays. From this, we will then deduce that in any implementable flow at least one of the commodities~\colNo{D} and~\colNo{F} incurs an additional total cost of order~$M$.
    
    \begin{claim}\label{claim: ImplemenatbleImplies}
        If $u \in L_+(\hori)^\GA$ is implementable via $(h,\prices)$, then we have:
        \begin{thmparts}
            \item $u^{\colNo{B}}_{\arc_5}(t) = 0$ for almost all $t \in \hori$.\label[thmpart]{claim: ImplemenatbleImplies: <=3} 
            \item $\trav_{\arc_5}(u,t) = 2-\varepsilon$ for all $ t \in [0,4]$. \label[thmpart]{claim: ImplemenatbleImplies: Trav=1}
        \end{thmparts}
    \end{claim}

    Here, \ref{claim: ImplemenatbleImplies: <=3} states that no \colName{B} flow (commodity~\colNo{B}) ever enters edge $\arc_5$. This is true because $\arc_5$'s tail is that commodity's \sink and, hence, any flow leaving that node could instead also just directly leave the network there and, thereby, improve the flow's overall objective for \ref{opt: Master}  without increasing any edge loads -- which would be a contradiction to the implementability of the given flow by \Cref{thm: ImplementableImpliesOptimal}. 
        This, in turn, directly implies that no queue exists on edge~$\arc_5$ until time~$3$ (i.e.\ \ref{claim: ImplemenatbleImplies: Trav=1}) as the capacity of this edge is~$2$ and the only other commodity which can enter this edge during that time is commodity~\colNo{C} (\colName{C}) at a rate of~$2$.

    Next, we can deduce from \Cref{claim: ImplemenatbleImplies} that in any implementable flow the \colName{C} (commodity~\colNo{C}) flow arrives at node $\dest_\colNo{F}$ during $[4,5]$ where it has two choices: Either enter edge~$\arc_2$ to travel directly towards its \sink or take a detour via the cycle $(\arc_6,\arc_5,\arc_4)$ first.  If a significant part of \colName{C} flow takes the detour, it will, together with the \colName{F} flow entering the network at node~$\source_\colNo{F}$, create a queue on edge~$\arc_4$ which will delay all of the rest of the  \colName{F} commodity's flow and, hence, significantly increasing (by order $M$) that commodity's total travel time. Next, the \colName{E} flow (of commodity~\colNo{E}) arriving at node $\dest_\colNo{F}$ during~$[5,6]$ faces the same choice as the \colName{C} flow before it. If a significant part of it chooses the cycle, it again creates a queue on edge~$\arc_4$, delaying most of the \colName{F} commodity's flow. If, on the other hand, large parts of both \colName{C} \emph{and} \colName{E} flow choose their respective direct paths from~$\dest_{\colNo{F}}$ (both starting with edge~$\arc_2$) the \colName{C} flow creates a queue on edge~$\arc_2$, which in turn delays the later arriving \colName{E} flow. Consequently the \colName{E} flow arrives at node~$\source_\colNo{D}$ together with the \colName{D} flow entering the network there. Combined, this flow then exceeds edge~$\arc_3$'s capacity and builds up a queue, which delays all of the following \colName{D} flow. In the flow depicted in \Cref{fig: CE_SystemOptimumImplementable_ImplFlowVariantA} one can see both effects as it is right at the border between the two cases described before.

    \begin{figure}
        \centering
        \begin{adjustbox}{max width=\textwidth}
            \begin{adjustbox}{max width=\textwidth}
                \def\CESOfigversion{implementableB}
                \def\animationdatapath{tikz/}
                \input{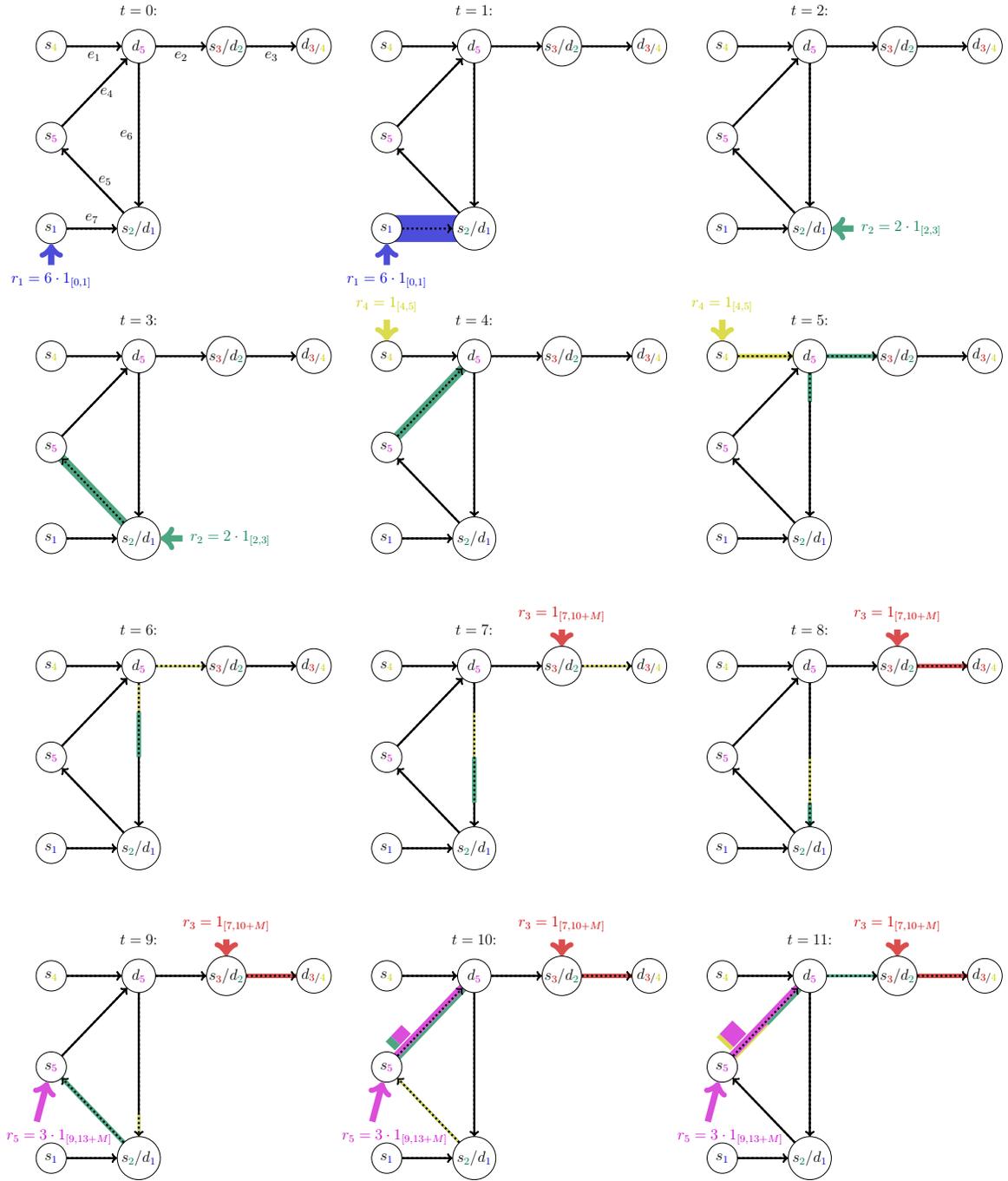}
            \end{adjustbox}
        \end{adjustbox}
        \caption{Another example for an implementable flow in the network considered in \Cref{thm: CounterSysopImpl} (with $\varepsilon=1$).  This is an example for a flow where \Cref{claim: ImplemenatbleImpliesHighTT:For6} holds, i.e.\ a significant amount of \protect\colName{C} and \protect\colName{E} flow travels around the cycle creating a queue on edge~$\arc_4$, which then delays all future particle of the \protect\colName{F} flow.}
        \label{fig: CE_SystemOptimumImplementable_ImplFlowVariantB}
    \end{figure}

    \begin{figure}
        \centering
        \begin{adjustbox}{max width=\textwidth}
            \def\CESOfigversion{implementableA}
            \def\animationdatapath{tikz/}
            \input{tikz/SystemOptimumNotImplementable}
        \end{adjustbox}
        \caption{An example for an implementable flow in the network considered in \Cref{thm: CounterSysopImpl} (with $\varepsilon=1$).  This is an example for a flow where both \Cref{claim: ImplemenatbleImpliesHighTT:For4,claim: ImplemenatbleImpliesHighTT:For6} hold. \Cref{claim: ImplemenatbleImpliesHighTT} holds as large amounts of both \protect\colName{C} and \protect\colName{E} flow take the direct path towards their \sink[.] The \protect\colName{C} flow then creates a queue on edge~$\arc_2$, delaying the later arriving \protect\colName{E} flow which then, in turn, creates a queue on edge~$\arc_3$ and delays all future particles of the \protect\colName{D} flow. \Cref{claim: ImplemenatbleImpliesHighTT:For6} holds as well since still a significant part of the flow uses the cycle~$(\arc_6,\arc_5,\arc_4)$ and, therefore, creates a queue on edge~$\arc_4$ delaying all future particles of the \protect\colName{F} flow.}
        \label{fig: CE_SystemOptimumImplementable_ImplFlowVariantA}
    \end{figure}

    \begin{claim}\label{claim: ImplemenatbleImpliesHighTT}
        If $u \in L_+(\hori)^\GA$ is implementable via $(h,\prices)$, then at least one of the following statements is true: 
        \begin{thmparts}
            \item $\dup{\trav(u,\cdot)}{u^\colNo{D}}\geq (3 + M)\varepsilon + M/4$ \label[thmpart]{claim: ImplemenatbleImpliesHighTT:For4}
            \item $\dup{\trav(u,\cdot)}{u^\colNo{F}}\geq 2\varepsilon(M+4)+M/4$.\label[thmpart]{claim: ImplemenatbleImpliesHighTT:For6}  
        \end{thmparts}
    \end{claim}

    With this, we can now conclude the proof of the \namecref{thm: CounterSysopImpl} by constructing a flow $h \in \wir$ with a corresponding edge flow $\g = \ell(h)$ that  is not implementable, yet, has strictly lower total travel time (for $\M$ sufficiently large)  than any  implementable flow $u$.  This then shows that any implementable flow is not \sysop and, conversely, no \sysop flow is implementable. 
    
    We define $h$ via 
        \begin{align}\label{eq: SysptFlowDef}
            h_\wa = \begin{cases}
                \inflow_{\colNo{B}}, &\text{if } \wa = ((\arc_7,\arc_5,\arc_4,\arc_6),\colNo{B}) \\
                \inflow_{\colNo{C}}, &\text{if } \wa = ((\arc_5,\arc_4,\arc_2),\colNo{C}) \\
                \inflow_{\colNo{D}}, &\text{if } \wa = ((\arc_3),\colNo{D}) \\
                \inflow_{\colNo{E}}, &\text{if } \wa = ((\arc_1,\arc_2,\arc_3),\colNo{E}) \\
                \inflow_{\colNo{F}}, &\text{if } \wa = ((\arc_4),\colNo{F}) \\
                0, &\text{else. } \\
            \end{cases}
        \end{align} 
    That is, $h$ is the flow wherein the \colName{B} particles (of commodity~\colNo{B}) travel around the cycle $(\arc_5,\arc_4,\arc_6)$ once before leaving the network at their \sink while all other commodities take their direct path. It is, then, an immediate consequence of \Cref{claim: ImplemenatbleImplies: <=3} that the induced edge flow~$\g$ is not implementable.
    A careful computation then shows the following structural properties of~$\g$ (in \Cref{fig: CE_SystemOptimumImplementable_OptFlow} we also provide a graphical depiction of this flow):

    \begin{figure}
        \centering
        \begin{adjustbox}{max width=\textwidth}
            \def\CESOfigversion{optimal}
            \def\animationdatapath{tikz/}
            \input{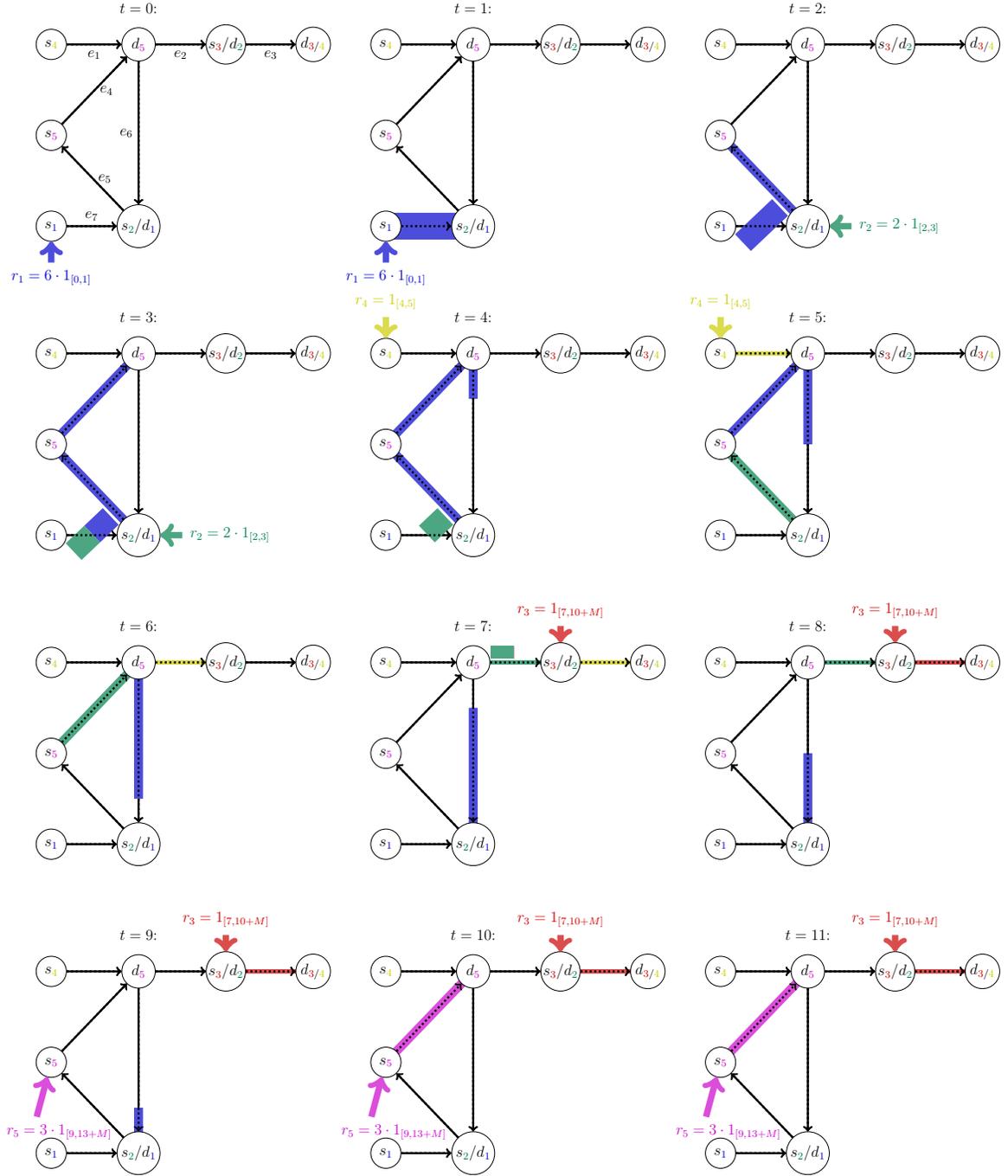}
        \end{adjustbox}
        \caption{The flow used in \Cref{claim: PropsSysop} to show that there are flows with strictly smaller total travel time than any implementable flow (depicted here for $\varepsilon=1$).}
        \label{fig: CE_SystemOptimumImplementable_OptFlow}
    \end{figure}

    \begin{claim}\label{claim: PropertiesOfSysOpt}
         Let $\g$ be the edge flow induced by $h$ defined in \eqref{eq: SysptFlowDef}. Then, we have 
         \begin{enumerate}[label=\alph*)]
             \item $\g^{\colNo{B}}_{\arc_7} = \inflow_{\colNo{B}} = 6\cdot 1_{[0,1]}$, $\g^{\colNo{B}}_{\arc_5} = 6\cdot 1_{[1,2]}$, $\g^{\colNo{B}}_{\arc_4} = 2\cdot 1_{[3-\varepsilon,6-\varepsilon]}$, $\g^{\colNo{B}}_{\arc_6} = 2\cdot 1_{[3,6]}$  and $\g^{\colNo{B}}_{\arc} = 0$ for all other edges $\arc\in \GA\setminus\{\arc_7,\arc_5,\arc_4,\arc_6\}$.  \label[thmpart]{claim: PropertiesOfSysOpt:comm2} 
             \item $\g^{\colNo{C}}_{\arc_5} = \inflow_\colNo{C} = 2\cdot 1_{[2,3]}$, 
             $\g^{\colNo{C}}_{\arc_4} = 2\cdot 1_{[6-\varepsilon,7-\varepsilon]}$, $\g^{\colNo{C}}_{\arc_2}= 2\cdot 1_{[6,7]}$ 
             and $\g^{\colNo{C}}_{\arc} = 0$ for all edges $\arc\in \GA \setminus\{\arc_5,\arc_4,\arc_2\}$.\label[thmpart]{claim: PropertiesOfSysOpt:comm3} 
             \item $\g^{\colNo{D}}_{\arc_3} = \inflow_{\colNo{D}} = 1_{[7,10+M]}$ and $\g^{\colNo{D}}_{\arc} = 0$ for all other edges $\arc\in \GA \setminus\{\arc_3\}$.  \label[thmpart]{claim: PropertiesOfSysOpt:comm4} 
             \item $\g^{\colNo{E}}_{\arc_1} = \inflow_{\colNo{E}} = 1_{[4,5]}$, $\g^{\colNo{E}}_{\arc_2} = 1_{[5,6]}$, $ g^{\colNo{E}}_{\arc_3} = 1_{[6,7]}$ and $\g^{\colNo{E}}_{\arc} = 0$ for all other edges $\arc\in \GA \setminus \{\arc_1,\arc_2,\arc_3\}$.  \label[thmpart]{claim: PropertiesOfSysOpt:comm5} 
             \item $\g^{\colNo{F}}_{\arc_4} = \inflow_{\colNo{F}} = 2\cdot 1_{[9,13+M]}$ and $\g^{\colNo{F}}_{\arc} = 0$ for all other edges $\arc\in \GA \setminus \{\arc_4\}$.  \label[thmpart]{claim: PropertiesOfSysOpt:comm6} 
         \end{enumerate} 
         Moreover, the induced travel times satisfy
         \begin{enumerate}[resume,label=\alph*)]
             \item $\trav_{\arc_5}(\g,t) \leq 5-\varepsilon$ for all $t \in [1,2]$ and $\trav_{\arc_5}(\g,t) =  4-\varepsilon$ for all $t \in [2,3]$. \label[thmpart]{claim: PropertiesOfSysOpt:trav5} 
             \item $\trav_{\arc_4}(\g,t) = \varepsilon$ for all $t \in \hori'$. \label[thmpart]{claim: PropertiesOfSysOpt:trav4} 
             \item $\trav_{\arc_2}(\g,t) = 1$ for all $t \in [5,6]$ and $\trav_{\arc_2}(\g,t) \leq 3$ for all $t \in [6,7]$. \label[thmpart]{claim: PropertiesOfSysOpt:trav2} 
             \item $\trav_{\arc_3}(\g,t) = \varepsilon$ for all $t \in \hori'$. \label[thmpart]{claim: PropertiesOfSysOpt:trav3} 
         \end{enumerate}
     \end{claim}

    With this, it is now straightforward to compute the individual commodities' total travel times and, in particular, conclude that the total travel times incurred by~$\g$ are optimal for commodities~\colNo{D} and~\colNo{F} and independent of~$M$ for all other commodities.
          
    \begin{claim}\label{claim: PropsSysop}
        Let $\g$ be the edge flow induced by $h$ defined in \eqref{eq: SysptFlowDef}. Then, $\dup{\trav(\g,\cdot)}{\g^i},i = \colNo{B},\colNo{C},\colNo{E}$ is bounded by some constant that is independent of~$\M$ and~$\varepsilon$. Moreover, we have $\dup{\trav(\g,\cdot)}{\g^{\colNo{D}}} = (3+M)\varepsilon$ and $\dup{\trav(\g,\cdot)}{\g^{\colNo{F}}} = 2(4+M)\varepsilon$ and these values are minimal among all edge flows ineducable by some walk inflow rate contained in~$\wir$.
    \end{claim}

    Finally, by combining these properties of~$\g$ with the properties of any implementable flow from \Cref{claim: ImplemenatbleImpliesHighTT}, we can deduce that (for sufficiently large~$M$) $\g$ has strictly lower total travel time than any implementable flow.
    
    \begin{claim}\label{claim: CostDifferenceSysOpImplFlow} 
        For large enough $\M \in \Rnn$ and any $\varepsilon \in (0,1]$
        any implementable flow $u$ has strictly higher total travel time than $\g$ induced by $h$ as defined in \eqref{eq: SysptFlowDef}.
    \end{claim}

    This claim then directly implies that there exists a network (namely the one depicted in \Cref{fig: CE_SystemOptimumImplementable} with a suitable choice of the parameter~$M$ and any choice of $\varepsilon \in (0,1]$) where no \sysop flow is implementable.
\end{proof}

\begin{remark}
    It is straightforward to verify that for any choice of $\varepsilon \in (0,1]$ the network constructed in the proof of \Cref{thm: CounterSysopImpl} does have a \sysop flow: 
    Since we have $\vot_i=1$ for all $i \in I$, we can use \Cref{lem: aggCostsVSwalkCosts} to rewrite \ref{opt: System} as 
    \begin{align}
        \inf_{h} \;  &\dup{\Psi(\ell(h),\cdot)}{h}\tag{$\tilde{\mathrm{P}}^{\mathrm{sys}}$} \label{opt: SystemRef}  \\
            &h \in \wir.\nonumber
    \end{align}
    Moreover, as the inflow rates are bounded, $\wir \subseteq L_+^p(\hori)^{\Routes}$ for any $p>1$. 
    Furthermore, as the travel times are lower bounded by~$\varepsilon>0$ and the induced aggregated edge flows are required to be supported on $\hori'$, there are only finitely many walks $\Routes'$ that can be utilized by a feasible $h$. 
    Combining these two insights, it is then straight forward to show that the set of feasible solutions to the above problem is weakly compact in $L^p(\hori)^{\Routes'}$ for $p>1$. 
    The existence of an optimal solution then follows by the fact that  the mapping sending walk-inflows to total travel time is weakly continuous as well, cf.~\cite[Lemma~7]{CominettiCL15}. 
    Hence, the above \Cref{thm: CounterSysopImpl} also shows that there exist \sysop flows which are not implementable.
\end{remark}

Moreover, comparing the total cost for the best flow constructed in the proof of \Cref{thm: CounterSysopImpl} with the lower bounds on the total cost of any implementable flow derived there, also gives us a lower bound on the price of stability, i.e.\ the ration of best equilibrium to system optimum (cf.\ \cite{Anshelevich08}), achievable by tolls. Formally, we define the \emph{price of stability achievable by tolls} as
    \[\sup_{\text{network } \mathcal{N}}\frac{\inf\set{\dup{\trav(\g,\cdot)}{\g} | \g \in \ell(\wir) \text{ implementable}}}{\inf\set{\dup{\trav(\g,\cdot)}{\g} | \g \in \ell(\wir)}},\]
where a network~$\mathcal{N}$ consists of a graph, network inflow rate, a time horizon, travel time functions and the induced network loading.

\begin{corollary}
    For multi-source, multi-\sink networks, the price of stability achievable by tolls is unbounded.
\end{corollary}
\begin{proof}
  We consider the instances from the proof of \Cref{thm: CounterSysopImpl} (i.e.\ the one depicted in \Cref{fig: CE_SystemOptimumImplementable}) parametrized by $M \in \Rnn$ and $\varepsilon \in (0,1]$.
From \Cref{claim: PropsSysop} we get the existence of a constant~$C \in \Rnn$ such that for any such instance there exists a flow~$\g$ (namely the one induced by~$h$ defined in \eqref{eq: SysptFlowDef}) with $\dup{\trav(\g,\cdot)}{\g} \leq C + (11+3M)\varepsilon$.
    On the other hand, \Cref{claim: ImplemenatbleImpliesHighTT} tells us that for any implementable flow~$u$ we have
        $\dup{\trav(u,\cdot)}{u} \geq M/4$. Together, this implies that the price of stability can be unbounded. 
\end{proof}

\section{Implementability for Single-Source, Single-\Sink Networks}  
As we demonstrated in \Cref{sec: NonImpl}, not every system-optimal flow is implementable in the multi-source, multi-\sink case. 
For the single-source, single-\sink case, Graf et al.\ characterized the implementability of a flow $\gop$  under a subset of our assumptions (\Cref{ass: TravelGeneral,ass: LuImplL}) through the property that $\gop$ does not send flow into edges leaving the \sink (\Tollref{thm: CombiChara}). 
From this, however, it does not follow directly that system optimal flows are implementable as 
 it remains unclear whether 
 sending flow on an outgoing edge from its \sink can reduce the total travel time (see the proof of \Cref{thm: CounterSysopImpl} for such an effect in the multi-source, multi-\sink case).
The following \Cref{exa: CE_SystemOptimumWithoutMonotonicity} illustrates that this is indeed possible under our current assumptions.

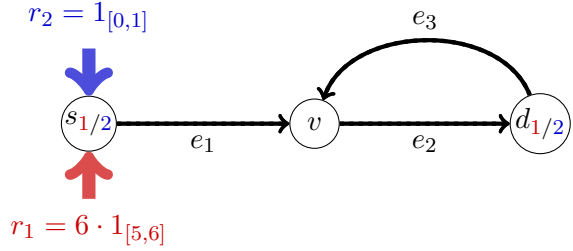
\begin{figure}
    \centering
    \BigPicture[0]{
    \def\CENMfigversion{network}
    \begin{adjustbox}{max width=\textwidth}
        \newcommand{\drawNetworkSkeleton}[1][]{
    \coordinate(s1)at(0,0);
    \coordinate(s2)at(3,0);
    \coordinate(t2)at(6,0);
    
    \node[namedVertexW,inner sep=.15](temp-s1)at(s1){$\source_{\color{colA}1/\color{colB}2}$};
    \node[namedVertexW,inner sep=4](temp-s2)at(s2){$v$};
    \node[namedVertexW,inner sep=.15](temp-t2)at(t2){$\dest_{{\color{colA}1}/\color{colB}2}$};
    
    \draw[edge](temp-s1) --node[below]{\ifthenelse{\equal{#1}{desc}}{$e_1$}{}}(temp-s2);
    \draw[edge](temp-s2) --node[below]{\ifthenelse{\equal{#1}{desc}}{$e_2$}{}}(temp-t2);
    \draw[edge](temp-t2) to[bend right=70]node[above]{\ifthenelse{\equal{#1}{desc}}{$e_3$}{}}(temp-s2);
}

\newcommand{\drawNetwork}[1][]{
    \ifthenelse{\equal{#1}{desc}}{}{
        \node at($(s2)!.5!(t2)+(0,1.8)$){$t=#1$:};
    }
    
    \node[namedVertexF,inner sep=.15](s1)at(s1){$\source_{{\color{colA}1}/\color{colB}2}$};
    \node[namedVertexF,inner sep=4](s2)at(s2){$v$};
    \node[namedVertexF,inner sep=.15](t2)at(t2){$\dest_{{\color{colA}1}/\color{colB}2}$};

    \ifthenelse{\equal{#1}{desc} \OR \(\lengthtest{#1 pt > 4.9pt} \AND \lengthtest{#1 pt < 6.1pt}\)}{
        \draw[line width=5,<-,fcolAr](s1) -- +(0,-1)node[anchor=north,color=colA]{$r_1=6\cdot\CharF[[5,6]]$};
    }{}
    \ifthenelse{\equal{#1}{desc} \OR \lengthtest{#1 pt < 1.1pt}}{
        \draw[line width=5,<-,fcolBr](s1) -- +(0,1)node[anchor=south,color=colB]{$r_2=\CharF[[0,1]]$};
    }{}

    \draw[edge,dotted,opacity=.3](s1) --(s2);
    \draw[edge,dotted,opacity=.3](s2) --(t2);
    \draw[edge,dotted,opacity=.3](t2) to[bend right=70](s2);
}


\ifthenelse{\equal{\CENMfigversion}{network}}{
\begin{tikzpicture}
    \drawNetworkSkeleton[desc]
    
    \drawNetwork[desc]
\end{tikzpicture}
}


\ifthenelse{\equal{\CENMfigversion}{optimal}}{
\begin{tikzpicture}
    \begin{scope}[xshift=0cm, yshift=0cm]
        \drawNetworkSkeleton
    
        \drawNetwork[0]
    \end{scope}

    \begin{scope}[xshift=9cm, yshift=0cm]
        \drawNetworkSkeleton

        \draw[line width=4pt,fcolBr](s1)--(s2);
    
        \drawNetwork[1]
    \end{scope}

    \begin{scope}[xshift=0cm, yshift=-4cm]
        \drawNetworkSkeleton

        \draw[line width=4pt,fcolBr](s2)--($(s2)!.5!(t2)$);
    
        \drawNetwork[2]
    \end{scope}

    \begin{scope}[xshift=9cm, yshift=-4cm]
        \drawNetworkSkeleton

        \draw[line width=4pt,fcolBr]($(s2)!.5!(t2)$)--(t2);
    
        \drawNetwork[3]
    \end{scope}

    \begin{scope}[xshift=0cm, yshift=-8cm]
        \drawNetworkSkeleton
    
        \drawNetwork[4]

        \draw[line width=4pt,fcolB](t2)to[bend right=70](s2);
    \end{scope}

    \begin{scope}[xshift=9cm, yshift=-8cm]
        \drawNetworkSkeleton

        \draw[line width=4pt,fcolBr](s2)--($(s2)!.67!(t2)$);
    
        \drawNetwork[5]
    \end{scope}

    \begin{scope}[xshift=0cm, yshift=-12cm]
        \drawNetworkSkeleton

        \draw[line width=4pt,fcolBr]($(s2)!.67!(t2)$)--(t2);
        \draw[line width=16pt,fcolAr](s1)--(s2);
    
        \drawNetwork[6]
    \end{scope}

    \begin{scope}[xshift=9cm, yshift=-12cm]
        \drawNetworkSkeleton

        \draw[line width=16pt,fcolAr](s2)--(t2);
    
        \drawNetwork[7]
    \end{scope}

    
\end{tikzpicture}
}{}


\ifthenelse{\equal{\CENMfigversion}{implementable}}{
\begin{tikzpicture}
    \begin{scope}[xshift=0cm, yshift=0cm]
        \drawNetworkSkeleton
    
        \drawNetwork[0]
    \end{scope}

    \begin{scope}[xshift=9cm, yshift=0cm]
        \drawNetworkSkeleton

        \draw[line width=4pt,fcolBr](s1)--(s2);
    
        \drawNetwork[1]
    \end{scope}

    \begin{scope}[xshift=0cm, yshift=-4cm]
        \drawNetworkSkeleton

        \draw[line width=4pt,fcolBr](s2)--($(s2)!.5!(t2)$);
    
        \drawNetwork[2]
    \end{scope}

    \begin{scope}[xshift=9cm, yshift=-4cm]
        \drawNetworkSkeleton

        \draw[line width=4pt,fcolBr]($(s2)!.5!(t2)$)--(t2);
    
        \drawNetwork[3]
    \end{scope}

    \begin{scope}[xshift=0cm, yshift=-8cm]
        \drawNetworkSkeleton
    
        \drawNetwork[4]
    \end{scope}

    \begin{scope}[xshift=9cm, yshift=-8cm]
        \drawNetworkSkeleton

        \drawNetwork[5]
    \end{scope}

    \begin{scope}[xshift=0cm, yshift=-12cm]
        \drawNetworkSkeleton

        \draw[line width=16pt,fcolAr](s1)--(s2);
    
        \drawNetwork[6]
    \end{scope}

    \begin{scope}[xshift=9cm, yshift=-12cm]
        \drawNetworkSkeleton

        \draw[line width=16pt,fcolAr](s2)--($(s2)!.67!(t2)$);
    
        \drawNetwork[7]
    \end{scope}

    \begin{scope}[xshift=0cm, yshift=-16cm]
        \drawNetworkSkeleton

        \draw[line width=16pt,fcolAr]($(s2)!.67!(t2)$)--(t2);
    
        \drawNetwork[8]
    \end{scope}

    \begin{scope}[xshift=9cm, yshift=-16cm]
        \drawNetworkSkeleton

        \drawNetwork[9]
    \end{scope}
\end{tikzpicture}
}{}
    \end{adjustbox}
    }
    \caption{The 2-commodity network considered in \Cref{exa: CE_SystemOptimumWithoutMonotonicity}. All travel times are constant (flow and time-independent) except on $\arc_2$ which has a travel time of $\trav_{\arc_2}(\g_{\arc_2},t) \coloneqq \max\Set{2-\frac{1}{2}\int_0^{t-1}\min\set{g_{\arc_2}(t'),2}\di\sigma(t'),1}$.}
    \label{fig:CE-NonMonotonicity-Network}
\end{figure}

\begin{example}\label{exa: CE_SystemOptimumWithoutMonotonicity}
    Consider the 2-commodity network depicted in \Cref{fig:CE-NonMonotonicity-Network}. All edges except for edge~$\arc_2$ have a constant (flow and time independent) travel time of~$1$ whereas the travel time of edge~$\arc_2$ is defined by $\trav_{\arc_2}(\g_{\arc_2},t) \coloneqq \max\Set{2-\frac{1}{2}\int_0^{t-1}\min\set{g_{\arc_2}(t'),2}\di\sigma(t'),1}$, i.e.\ the travel time starts at~$2$ and then decreases with increasing cumulative inflow until it reaches its minimal travel time of~$1$. The network inflow rates are $r_1 = 6\cdot\CharF[[5,6]]$ and $r_2=\CharF[[0,1]]$. 
    The only (and, therefore, best possible) implementable flow in this network is the flow in which both commodities send all flow along the direct walk $(\arc_1,\arc_2)$, resulting in a total travel time of $6\cdot (1+1.5) + 1\cdot (2+1) = 18$. To see why this is the only implementable flow, observe that any other flow has to send flow along the backwards edge~$\arc_3$. As this creates a flow-carrying $\dest$-cycle, any such flow is not implementable by \Cref{thm: CombiChara}.
    A (non-implementable) flow with strictly lower total travel time is given by the flow in which the particles of the second commodity take a detour along the cycle $(\arc_2,\arc_3)$ while the particles of the first commodity take the direct walk $(\arc_1,\arc_2)$,  resulting in a total travel time of $6\cdot (1+1) + 1\cdot (1+2+1+1.5) = 17.5$   
    (see \Cref{fig:CE-NonMonotonicity-Implementable,fig:CE-NonMonotonicity-Optimal} for graphical depictions of the two flows). 
 
    \begin{figure}
        \centering
        \BigPicture[0]{
        \def\CENMfigversion{implementable}
        \begin{adjustbox}{max width=\textwidth}
            \newcommand{\drawNetworkSkeleton}[1][]{
    \coordinate(s1)at(0,0);
    \coordinate(s2)at(3,0);
    \coordinate(t2)at(6,0);
    
    \node[namedVertexW,inner sep=.15](temp-s1)at(s1){$\source_{\color{colA}1/\color{colB}2}$};
    \node[namedVertexW,inner sep=4](temp-s2)at(s2){$v$};
    \node[namedVertexW,inner sep=.15](temp-t2)at(t2){$\dest_{{\color{colA}1}/\color{colB}2}$};
    
    \draw[edge](temp-s1) --node[below]{\ifthenelse{\equal{#1}{desc}}{$e_1$}{}}(temp-s2);
    \draw[edge](temp-s2) --node[below]{\ifthenelse{\equal{#1}{desc}}{$e_2$}{}}(temp-t2);
    \draw[edge](temp-t2) to[bend right=70]node[above]{\ifthenelse{\equal{#1}{desc}}{$e_3$}{}}(temp-s2);
}

\newcommand{\drawNetwork}[1][]{
    \ifthenelse{\equal{#1}{desc}}{}{
        \node at($(s2)!.5!(t2)+(0,1.8)$){$t=#1$:};
    }
    
    \node[namedVertexF,inner sep=.15](s1)at(s1){$\source_{{\color{colA}1}/\color{colB}2}$};
    \node[namedVertexF,inner sep=4](s2)at(s2){$v$};
    \node[namedVertexF,inner sep=.15](t2)at(t2){$\dest_{{\color{colA}1}/\color{colB}2}$};

    \ifthenelse{\equal{#1}{desc} \OR \(\lengthtest{#1 pt > 4.9pt} \AND \lengthtest{#1 pt < 6.1pt}\)}{
        \draw[line width=5,<-,fcolAr](s1) -- +(0,-1)node[anchor=north,color=colA]{$r_1=6\cdot\CharF[[5,6]]$};
    }{}
    \ifthenelse{\equal{#1}{desc} \OR \lengthtest{#1 pt < 1.1pt}}{
        \draw[line width=5,<-,fcolBr](s1) -- +(0,1)node[anchor=south,color=colB]{$r_2=\CharF[[0,1]]$};
    }{}

    \draw[edge,dotted,opacity=.3](s1) --(s2);
    \draw[edge,dotted,opacity=.3](s2) --(t2);
    \draw[edge,dotted,opacity=.3](t2) to[bend right=70](s2);
}


\ifthenelse{\equal{\CENMfigversion}{network}}{
\begin{tikzpicture}
    \drawNetworkSkeleton[desc]
    
    \drawNetwork[desc]
\end{tikzpicture}
}


\ifthenelse{\equal{\CENMfigversion}{optimal}}{
\begin{tikzpicture}
    \begin{scope}[xshift=0cm, yshift=0cm]
        \drawNetworkSkeleton
    
        \drawNetwork[0]
    \end{scope}

    \begin{scope}[xshift=9cm, yshift=0cm]
        \drawNetworkSkeleton

        \draw[line width=4pt,fcolBr](s1)--(s2);
    
        \drawNetwork[1]
    \end{scope}

    \begin{scope}[xshift=0cm, yshift=-4cm]
        \drawNetworkSkeleton

        \draw[line width=4pt,fcolBr](s2)--($(s2)!.5!(t2)$);
    
        \drawNetwork[2]
    \end{scope}

    \begin{scope}[xshift=9cm, yshift=-4cm]
        \drawNetworkSkeleton

        \draw[line width=4pt,fcolBr]($(s2)!.5!(t2)$)--(t2);
    
        \drawNetwork[3]
    \end{scope}

    \begin{scope}[xshift=0cm, yshift=-8cm]
        \drawNetworkSkeleton
    
        \drawNetwork[4]

        \draw[line width=4pt,fcolB](t2)to[bend right=70](s2);
    \end{scope}

    \begin{scope}[xshift=9cm, yshift=-8cm]
        \drawNetworkSkeleton

        \draw[line width=4pt,fcolBr](s2)--($(s2)!.67!(t2)$);
    
        \drawNetwork[5]
    \end{scope}

    \begin{scope}[xshift=0cm, yshift=-12cm]
        \drawNetworkSkeleton

        \draw[line width=4pt,fcolBr]($(s2)!.67!(t2)$)--(t2);
        \draw[line width=16pt,fcolAr](s1)--(s2);
    
        \drawNetwork[6]
    \end{scope}

    \begin{scope}[xshift=9cm, yshift=-12cm]
        \drawNetworkSkeleton

        \draw[line width=16pt,fcolAr](s2)--(t2);
    
        \drawNetwork[7]
    \end{scope}

    
\end{tikzpicture}
}{}


\ifthenelse{\equal{\CENMfigversion}{implementable}}{
\begin{tikzpicture}
    \begin{scope}[xshift=0cm, yshift=0cm]
        \drawNetworkSkeleton
    
        \drawNetwork[0]
    \end{scope}

    \begin{scope}[xshift=9cm, yshift=0cm]
        \drawNetworkSkeleton

        \draw[line width=4pt,fcolBr](s1)--(s2);
    
        \drawNetwork[1]
    \end{scope}

    \begin{scope}[xshift=0cm, yshift=-4cm]
        \drawNetworkSkeleton

        \draw[line width=4pt,fcolBr](s2)--($(s2)!.5!(t2)$);
    
        \drawNetwork[2]
    \end{scope}

    \begin{scope}[xshift=9cm, yshift=-4cm]
        \drawNetworkSkeleton

        \draw[line width=4pt,fcolBr]($(s2)!.5!(t2)$)--(t2);
    
        \drawNetwork[3]
    \end{scope}

    \begin{scope}[xshift=0cm, yshift=-8cm]
        \drawNetworkSkeleton
    
        \drawNetwork[4]
    \end{scope}

    \begin{scope}[xshift=9cm, yshift=-8cm]
        \drawNetworkSkeleton

        \drawNetwork[5]
    \end{scope}

    \begin{scope}[xshift=0cm, yshift=-12cm]
        \drawNetworkSkeleton

        \draw[line width=16pt,fcolAr](s1)--(s2);
    
        \drawNetwork[6]
    \end{scope}

    \begin{scope}[xshift=9cm, yshift=-12cm]
        \drawNetworkSkeleton

        \draw[line width=16pt,fcolAr](s2)--($(s2)!.67!(t2)$);
    
        \drawNetwork[7]
    \end{scope}

    \begin{scope}[xshift=0cm, yshift=-16cm]
        \drawNetworkSkeleton

        \draw[line width=16pt,fcolAr]($(s2)!.67!(t2)$)--(t2);
    
        \drawNetwork[8]
    \end{scope}

    \begin{scope}[xshift=9cm, yshift=-16cm]
        \drawNetworkSkeleton

        \drawNetwork[9]
    \end{scope}
\end{tikzpicture}
}{}
        \end{adjustbox}
        }
        \caption{The 2-commodity network considered in \Cref{exa: CE_SystemOptimumWithoutMonotonicity}. The flow depicted here is the only implementable flow and has a total travel time of~$18$.}
        \label{fig:CE-NonMonotonicity-Implementable}
    \end{figure}
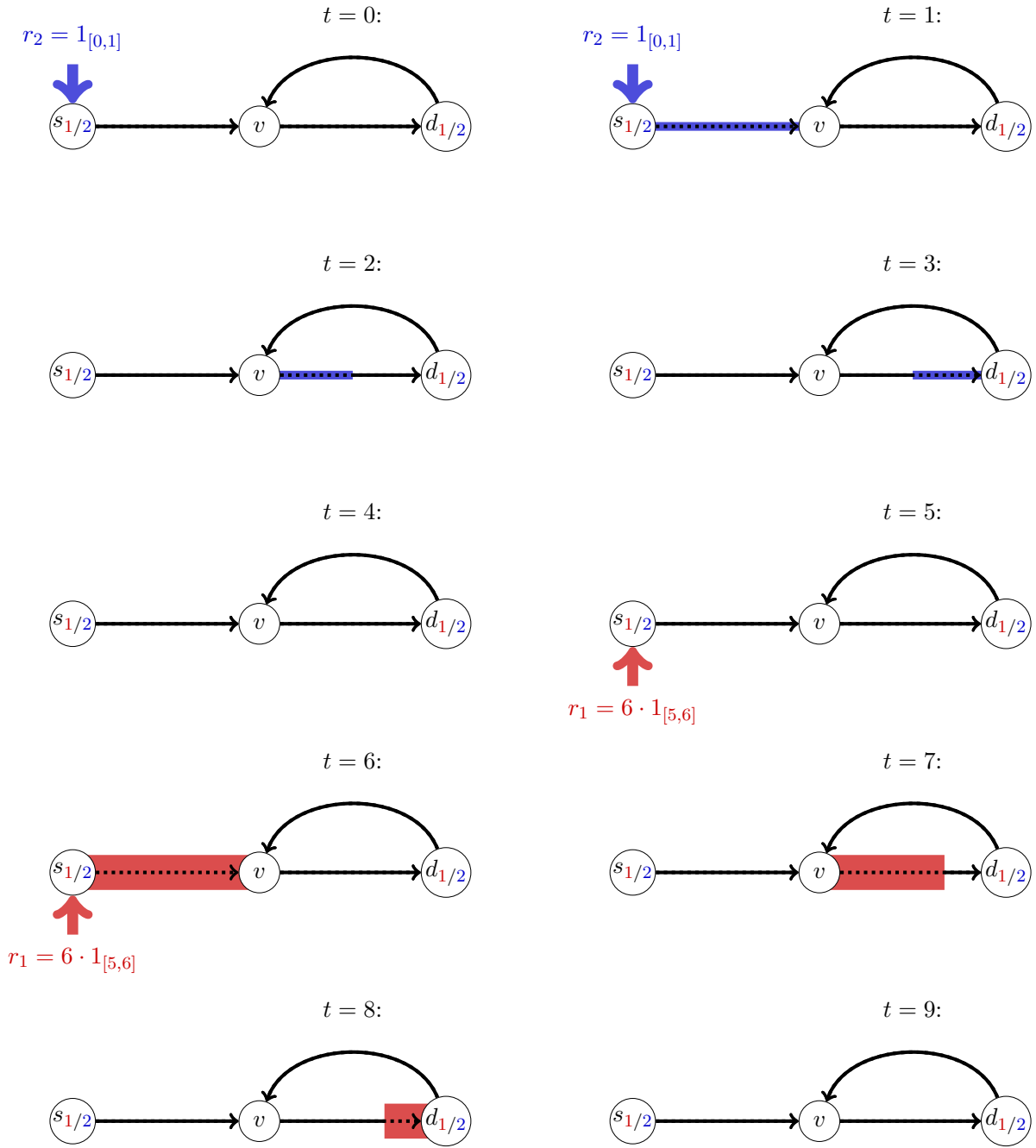
 
    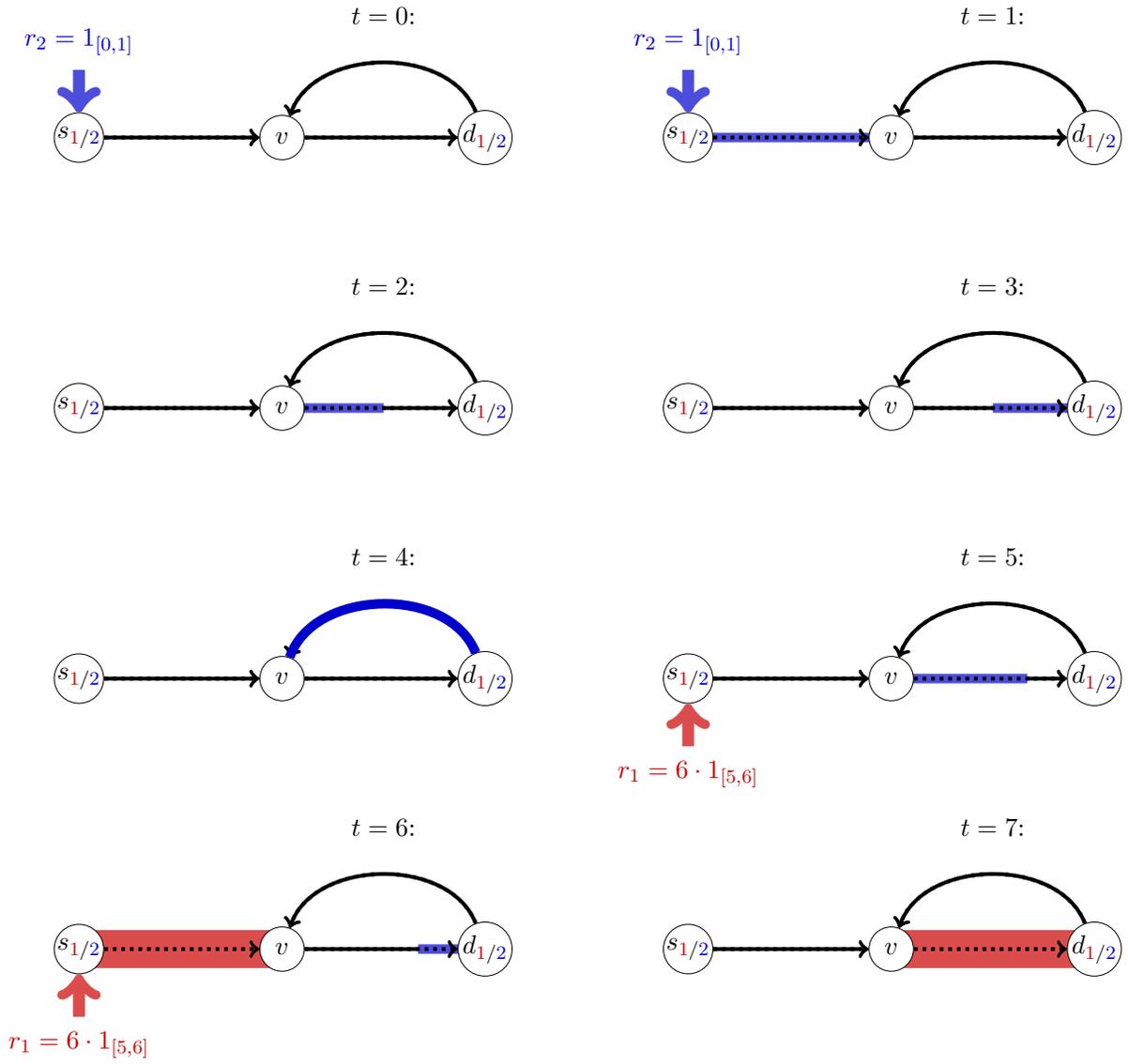
\begin{figure}
        \centering
        \BigPicture[0]{
        \def\CENMfigversion{optimal}
        \begin{adjustbox}{max width=\textwidth}
            \newcommand{\drawNetworkSkeleton}[1][]{
    \coordinate(s1)at(0,0);
    \coordinate(s2)at(3,0);
    \coordinate(t2)at(6,0);
    
    \node[namedVertexW,inner sep=.15](temp-s1)at(s1){$\source_{\color{colA}1/\color{colB}2}$};
    \node[namedVertexW,inner sep=4](temp-s2)at(s2){$v$};
    \node[namedVertexW,inner sep=.15](temp-t2)at(t2){$\dest_{{\color{colA}1}/\color{colB}2}$};
    
    \draw[edge](temp-s1) --node[below]{\ifthenelse{\equal{#1}{desc}}{$e_1$}{}}(temp-s2);
    \draw[edge](temp-s2) --node[below]{\ifthenelse{\equal{#1}{desc}}{$e_2$}{}}(temp-t2);
    \draw[edge](temp-t2) to[bend right=70]node[above]{\ifthenelse{\equal{#1}{desc}}{$e_3$}{}}(temp-s2);
}

\newcommand{\drawNetwork}[1][]{
    \ifthenelse{\equal{#1}{desc}}{}{
        \node at($(s2)!.5!(t2)+(0,1.8)$){$t=#1$:};
    }
    
    \node[namedVertexF,inner sep=.15](s1)at(s1){$\source_{{\color{colA}1}/\color{colB}2}$};
    \node[namedVertexF,inner sep=4](s2)at(s2){$v$};
    \node[namedVertexF,inner sep=.15](t2)at(t2){$\dest_{{\color{colA}1}/\color{colB}2}$};

    \ifthenelse{\equal{#1}{desc} \OR \(\lengthtest{#1 pt > 4.9pt} \AND \lengthtest{#1 pt < 6.1pt}\)}{
        \draw[line width=5,<-,fcolAr](s1) -- +(0,-1)node[anchor=north,color=colA]{$r_1=6\cdot\CharF[[5,6]]$};
    }{}
    \ifthenelse{\equal{#1}{desc} \OR \lengthtest{#1 pt < 1.1pt}}{
        \draw[line width=5,<-,fcolBr](s1) -- +(0,1)node[anchor=south,color=colB]{$r_2=\CharF[[0,1]]$};
    }{}

    \draw[edge,dotted,opacity=.3](s1) --(s2);
    \draw[edge,dotted,opacity=.3](s2) --(t2);
    \draw[edge,dotted,opacity=.3](t2) to[bend right=70](s2);
}


\ifthenelse{\equal{\CENMfigversion}{network}}{
\begin{tikzpicture}
    \drawNetworkSkeleton[desc]
    
    \drawNetwork[desc]
\end{tikzpicture}
}


\ifthenelse{\equal{\CENMfigversion}{optimal}}{
\begin{tikzpicture}
    \begin{scope}[xshift=0cm, yshift=0cm]
        \drawNetworkSkeleton
    
        \drawNetwork[0]
    \end{scope}

    \begin{scope}[xshift=9cm, yshift=0cm]
        \drawNetworkSkeleton

        \draw[line width=4pt,fcolBr](s1)--(s2);
    
        \drawNetwork[1]
    \end{scope}

    \begin{scope}[xshift=0cm, yshift=-4cm]
        \drawNetworkSkeleton

        \draw[line width=4pt,fcolBr](s2)--($(s2)!.5!(t2)$);
    
        \drawNetwork[2]
    \end{scope}

    \begin{scope}[xshift=9cm, yshift=-4cm]
        \drawNetworkSkeleton

        \draw[line width=4pt,fcolBr]($(s2)!.5!(t2)$)--(t2);
    
        \drawNetwork[3]
    \end{scope}

    \begin{scope}[xshift=0cm, yshift=-8cm]
        \drawNetworkSkeleton
    
        \drawNetwork[4]

        \draw[line width=4pt,fcolB](t2)to[bend right=70](s2);
    \end{scope}

    \begin{scope}[xshift=9cm, yshift=-8cm]
        \drawNetworkSkeleton

        \draw[line width=4pt,fcolBr](s2)--($(s2)!.67!(t2)$);
    
        \drawNetwork[5]
    \end{scope}

    \begin{scope}[xshift=0cm, yshift=-12cm]
        \drawNetworkSkeleton

        \draw[line width=4pt,fcolBr]($(s2)!.67!(t2)$)--(t2);
        \draw[line width=16pt,fcolAr](s1)--(s2);
    
        \drawNetwork[6]
    \end{scope}

    \begin{scope}[xshift=9cm, yshift=-12cm]
        \drawNetworkSkeleton

        \draw[line width=16pt,fcolAr](s2)--(t2);
    
        \drawNetwork[7]
    \end{scope}

    
\end{tikzpicture}
}{}


\ifthenelse{\equal{\CENMfigversion}{implementable}}{
\begin{tikzpicture}
    \begin{scope}[xshift=0cm, yshift=0cm]
        \drawNetworkSkeleton
    
        \drawNetwork[0]
    \end{scope}

    \begin{scope}[xshift=9cm, yshift=0cm]
        \drawNetworkSkeleton

        \draw[line width=4pt,fcolBr](s1)--(s2);
    
        \drawNetwork[1]
    \end{scope}

    \begin{scope}[xshift=0cm, yshift=-4cm]
        \drawNetworkSkeleton

        \draw[line width=4pt,fcolBr](s2)--($(s2)!.5!(t2)$);
    
        \drawNetwork[2]
    \end{scope}

    \begin{scope}[xshift=9cm, yshift=-4cm]
        \drawNetworkSkeleton

        \draw[line width=4pt,fcolBr]($(s2)!.5!(t2)$)--(t2);
    
        \drawNetwork[3]
    \end{scope}

    \begin{scope}[xshift=0cm, yshift=-8cm]
        \drawNetworkSkeleton
    
        \drawNetwork[4]
    \end{scope}

    \begin{scope}[xshift=9cm, yshift=-8cm]
        \drawNetworkSkeleton

        \drawNetwork[5]
    \end{scope}

    \begin{scope}[xshift=0cm, yshift=-12cm]
        \drawNetworkSkeleton

        \draw[line width=16pt,fcolAr](s1)--(s2);
    
        \drawNetwork[6]
    \end{scope}

    \begin{scope}[xshift=9cm, yshift=-12cm]
        \drawNetworkSkeleton

        \draw[line width=16pt,fcolAr](s2)--($(s2)!.67!(t2)$);
    
        \drawNetwork[7]
    \end{scope}

    \begin{scope}[xshift=0cm, yshift=-16cm]
        \drawNetworkSkeleton

        \draw[line width=16pt,fcolAr]($(s2)!.67!(t2)$)--(t2);
    
        \drawNetwork[8]
    \end{scope}

    \begin{scope}[xshift=9cm, yshift=-16cm]
        \drawNetworkSkeleton

        \drawNetwork[9]
    \end{scope}
\end{tikzpicture}
}{}
        \end{adjustbox}
        }
        \caption{The 2-commodity network considered in \Cref{exa: CE_SystemOptimumWithoutMonotonicity}. The flow depicted here has a total travel time of $17.5$ which is strictly less than the best possible implementable flow depicted in \Cref{fig:CE-NonMonotonicity-Implementable}.}
        \label{fig:CE-NonMonotonicity-Optimal}
    \end{figure}
 
\end{example}

As the above example illustrates, making sure that system optima are implementable requires additional assumptions about traversal times. We will now outline these assumptions and give a brief intuition of what they entail. 

\begin{assumption}\label{ass: TravelMono}
    Assume that the traversal times satisfy the following properties: 
  \begin{enumerate}[label = \roman*)]
    \item for any $\arc \in \GA$ and corresponding edge inflows $\g_\arc,\tilde{\g}_\arc \in L_+(\hori)$ with 
    $\tilde\G_\arc(t) \leq {\G}_\arc(t)$ for all $t \in \hori$, we have 
    $\tilde\G^-_\arc(t) \leq {\G}^-_\arc(t)$ for all $t \in \hori$.  
    \label[thmpart]{ass: TravelMono: Monotonicity} 
    \item for any $\arc \in \GA$ and corresponding edge inflows $\g_\arc,\tilde{\g}_\arc \in L_+(\hori)$ with $\tilde{\g}_\arc\leq  {\g}_\arc$, we have $\trav_{\arc}(\tilde{\g}_\arc,\cdot) \leq\trav_{\arc}({\g}_\arc,\cdot)$.  
    \label[thmpart]{ass: TravelMono: TravelMonotonicity}
\end{enumerate}  
\end{assumption}
Here, \labelcref{ass: TravelMono: Monotonicity}  states that larger cumulative inflow (at all times) has to lead to larger cumulative outflow (at all times) while \labelcref{ass: TravelMono: TravelMonotonicity} requires that larger inflow rates (at almost all times) lead to larger travel times (at all times).
Note that both conditions are fulfilled by the Vickrey queuing model (\cite[Corollary~3.23 and Corollary~3.19b)]{GrafThesis}) while neither is satisfied by the travel time on edge~$\arc_2$ in the network from \Cref{exa: CE_SystemOptimumWithoutMonotonicity}.  

Under these additional monotonicity assumptions, we are now able to derive our positive result that system optima in single-source, single-\sink networks are implementable. To prove this, we will first show that (under the above assumptions) a system optimum never sends flow over outgoing edges from the \sink[.] We actually show this property for the more general case of \emph{multi-source}, single-\sink networks. 

\begin{theorem}\label{thm: SsSysOptNoDOutflow}
    Consider a multi-source, single-\sink network with  travel times fulfilling \Cref{ass: TravelMono}.
    Then,  any  \sysop flow $\gop$ does not send flow along outgoing edges from $\dest$, i.e.~$\gop_\arc = 0,\arc \in \edgesFrom{\dest}$.
\end{theorem}

From this, our main result then follows by \Tollref{thm: CombiChara}:

\begin{theorem}\label{cor: SsSysOptImpl}
    In the single-source, single-\sink case with  travel times fulfilling \Cref{ass: TravelMono}, 
   any  \sysop flow is implementable. 
\end{theorem} 

For the proof of the above \namecref{thm: SsSysOptNoDOutflow}s, we require the additional concept of flows with and without waiting.
More precisely, we say that  $\g \in L_+(\hori)^\GA$ is an \emph{$(\source_i,\dest)$-flow without waiting} if $\opp_\dest \G$ is non-increasing and  for all $v \neq \dest$ and $t \in \hori$ the following equality holds
\begin{align}\label{eq: DefOutflow}
    \opp_v\G(t) = \int_0^t \sum_{i:s_i = v}\inflow_i \di \sigma.  
\end{align}
In particular, for $v \notin \{s_i\}_{i \in I}$, the right hand side is always $0$ and $\g$ is  said to fulfill flow conservation at $v$.
In contrast, we say that $\g$ is an \emph{$(\source_i,\dest)$-flow with (potential) waiting}, if for all $v \neq\dest$ and $t \in \hori$ the following \emph{in}equality holds
\begin{align*}
    \opp_v\G(t) \leq \int_0^t \sum_{i:s_i = v}\inflow_i \di \sigma  
\end{align*}
as well as equality for $t = t_f$.
For such flows, the value $\int_0^t \sum_{i:s_i = v}\inflow_i \di\sigma - \opp_v\G(t)$ can be interpreted as the amount of flow particles waiting at node $v$ at time $t$. 
Note that for flows with potential waiting, we also do not impose the property  of $\opp_\dest \G$ being non-increasing anymore. 
Remark also that it is straight forward to verify that induced edge flows $\g \in \ell(\wir)$ are $(\source_i,\dest)$-flows without waiting.

With these preliminaries at hand, we  start by showing  the following key ingredient for the proof of \Cref{thm: SsSysOptNoDOutflow}:  Under \Cref{ass: TravelMono}, any edge flow with waiting can be turned into an edge flow without waiting  and with at most the same total travel time. 
In order to prove this, we divide the time horizon~$\hori'$ into $K$ intervals of length $\tmin$ and adjust $\g$ inductively on them. More precisely, in the $j$-th induction step, 
we adjust $\g$ on $[(j-1)\cdot \tmin,j\cdot \tmin]$
by pushing waiting particles at nodes 
into suitable outgoing edges. 
Due to the monotonicity in \ref{ass: TravelMono: Monotonicity} and travel times being lower bounded by $\tmin$, this results in a flow $\tilde{\g}^{j}$ with potential waiting for which particles are not waiting during $[0,j\cdot\tmin]$. Additionally, the resulting flow has at most the same total travel time as $\g$ and hence, 
 $\tilde{\g}^K$ represents the desired flow.

\begin{lemma}\label{lem: WlogNoWaiting} 
    Consider a multi-source, single-\sink network with  travel times fulfilling \Cref{ass: TravelMono} and an arbitrary $(\source_i,\dest)$-flow $\g\in L_+(\hori)^\GA$ supported on $\hori'$ with potential waiting at nodes. 
    Then, there exists 
    an $(\source_i,\dest)$-flow  $\tilde{\g}\in L_+(\hori)^\GA$ supported on $\hori'$ without waiting and fulfilling 
    $\opp_\dest \tilde{\G}(t) \leq \opp_\dest {\G}(t)$ for all $t \in \hori$. 
\end{lemma} 
\begin{proof}
We introduce a super source~$\ssource$ and \sink $\sdest$ and extend   $\GV,\GA$ via   $\GAS:= \GA \cup \{(\ssource,\source_i)\}_{i \in I}\cup\{(\dest,\sdest)\}$ as well as $\GVS:= \GV \cup \{\ssource,\sdest\}$. The new edges have constant travel time of zero, i.e.\ $\trav_\arc\equiv 0$ for $\arc \in \GAS\setminus\GA$, and, hence, also fulfill the properties stated in \Cref{ass: TravelMono}. We extend $\g$ to a vector $\g \in L_+(\hori)^\GAS$  via ${\g}_{(\ssource,\source_i)}  =   \inflow_i, i\in I$ as well as $\g_{(\dest,\sdest)} = -\inflow_\dest$. 
Here, $\inflow_\dest$ denotes the net outflow rate of~$\g$ at~$\dest$, i.e.~the unique function fulfilling $\int_0^t\inflow_\dest \di\sigma = \opp_\dest \G (t), t\in \hori$. Note that the latter exists by \Cref{ass: TravelGeneral: Inflow}. 

    We assume \wlg that  the end point of $\hori'$ is a multiple of $\tmin$, i.e.~there exists a $K \in \N$ with $t_f' = K\cdot \tmin$ (otherwise, reduce $\tmin$ sufficiently). As described before, the following claim can be proven via an induction on~$j$ (cf.~\Cref{sec: WlogNoWaiting}): 
    \begin{claim}\label{claim: WlogNoWaiting}
        For all $j \in \{0,1,\ldots,K\}$, there exists a flow $\tilde{\g}^j \in L_+(\hori)^\GAS$ with
    \begin{enumerate}[label=(\arabic*)] 
        \item $\opp_v\tilde{\G}^{j}(t)\leq 0$ for all $t \in \hori$ and all $v \in \GV$, \label{indu2: Feasible}
        \item $\opp_v\tilde{\G}^{j}(t) = 0 $ for all $t\in [0,j\cdot \tmin] \cup \{t_f\}$ and all $v \in \GV$,\label{indu2: NoWai}
         \item $\opp_\sdest \tilde{\G}^{j}(t) \leq \opp_\sdest {\G}(t)$ for all $ t \in \hori$,  \label{indu2: SmallerObj} 
         \item $ \tilde{\g}^{j}_{(\ssource,\source_i)} =   \inflow_i$ for all $i\in I$, and\label{indu2: NetworkInflow} 
         \item $\tilde{\g}^{j}$ is supported on $\hori'=[0,t_f']$.\label{ind2u: Support} 
    \end{enumerate}
    \end{claim}
With this claim at hand,  we now observe that $\tilde\g^K$ is an $(\ssource,\sdest)$-flow in the extended network without waiting at nodes in $\GV$. This is, because due to \ref{indu2: NoWai} we have $\opp_v\tilde{\G}^K(t) = 0$ for both $t=K\cdot \tmin = t'_f$ and $t=t_f$ and, since $\tilde\g^K$ is supported on $[0,t'_f]$, the same holds for all times inbetween. Together with \ref{indu2: NoWai} we, therefore, have equality (i.e.\ no waiting) on all of~$\hori$ for all $v \in \GV$. Then, choosing $\tilde{\g}$ as the flow corresponding to $\tilde{\g}^K \in L_+(\hori)^\GAS$ in the original network, gives us an $(\source_i,\dest)$-flow without waiting which, additionally, satisfies $\opp_\dest\tilde{\G} = \opp_\sdest\tilde{\G}^K \leq \opp_\sdest {\G}$.
\end{proof}

With the above \namecref{lem: WlogNoWaiting} at hand we can now prove \Cref{thm: SsSysOptNoDOutflow} as follows: Let $\g^*$ be a \sysop flow and assume, by way of contradiction, that there exists an edge $\arc^* \in \edgesFrom{\dest}$ with $\g^*_{\arc^*} \neq 0$. Due to flow conservation, this implies that there must be a whole flow-carrying $\dest$-cycle~$c$ under~$\g^*$. We would now like to remove the flow on~$c$ from~$\g^*$ in order to construct a flow in $\ell(\wir)$ with strictly lower travel time. However, just reducing the edge inflow rates this way results in a flow $\g^\Delta$ that, in general, does not fulfill flow conservation aynmore as the travel times on the edges contained in~$c$ may change. 
 Yet, by the monotonicity stated in \Cref{ass: TravelMono: TravelMonotonicity}, we can show that the cumulative flow balance of $\g^\Delta$ is always nonpositive, that is, $\tilde{g}$ can be interpreted as flow in which particles may wait at nodes.  
Moreover, we can show that $\g^\Delta$ has a strictly lower travel time than $\g^*$. 
Thus, applying \Cref{lem: WlogNoWaiting} to $\g^\Delta$ yields   an edge flow $\tilde{\g}$
fulfilling flow conservation  and having strictly lower travel time than $\g^*$. 
 The proof is then finished by applying the flow decomposition theorem \FDref{thm: FlowDecomp} to $\tilde{\g}$, yielding a corresponding walk inflow rate in $\wir$, demonstrating that $\tilde{\g} \in \ell(\wir)$. 
 Since $\tilde{\g}$ has a strictly lower travel time than $\gop$, the existence of $\tilde{\g}$ contradicts the  system optimality of $\g^*$.

\begin{proof}[Proof of \Cref{thm: SsSysOptNoDOutflow}] 
    Let $\gop$ be  supported on $\hori'$  and \sysop with $h^* \in \wir$ inducing $\gop$ (i.e.\ $\gop$ is an optimal solution to~\eqref{opt: System}). 
    Assume for the sake of a contradiction that $\gop$ sends flow into at least one edge leaving the common \sink $\dest$, i.e.~there exists $\arc^*\in \edgesFrom{\dest}$ with $\gop_{\arc^*} >0$.

In the same way as in the proof of \Cref{lem: WlogNoWaiting} we extend the given network by introducing a super source~$\ssource$ and \sink $\sdest$ and extend   $\GV,\GA$ via   $\GAS:= \GA \cup \{(\ssource,\source_i)\}_{i \in I}\cup\{(\dest,\sdest)\}$ as well as $\GVS:= \GV \cup \{\ssource,\sdest\}$. The new edges have constant travel time of zero $\trav_\arc\equiv 0,\arc \in \GAS\setminus\GA$. Hence, they also fulfill the conditions stated in \Cref{ass: TravelMono}. 
Clearly, there is a bijection between the set of original \stwalk[\source_i]s $\Routes$ and the set of $\ssource$,$\sdest$-walks $\sRoutes$ in the extended network by sending $\wa \in \Routes_i$ to $\swa:=((\ssource,\source_i),\wa,(\dest,\sdest))$.  
Similarly, we have a bijection $\wir \to \swir, h \mapsto \tilde{h}$ between the possible walk inflow rates in the original network and the extended network via $\tilde{h}_{\swa} = h_\wa,\wa \in \Routes$,
which preserves the total travel time as the travel time along the new edges is zero, i.e.~$\dup{\trav(\ell(h),\cdot)}{\ell(h)} = \dup{\trav(\ell(\tilde{h}),\cdot)}{\ell(\tilde{h})}$
Corresponding to the bijection between walk flows, we also have a bijection between $(\source_i,\dest)$-flows with (and without) waiting in the original network and $(\ssource,\sdest)$-flows with  (and without) waiting at nodes contained in $\GV$ via setting ${\g}_{(\ssource,\source_i)}  =   \inflow_i, i\in I$ as well as $\g_{(\dest,\sdest)} = -\inflow_\dest$. 
Here, $\inflow_\dest$ denotes the net outflow rate at $\dest$ of $\g$ which exists by the required  condition in \Cref{ass: TravelGeneral: Inflow} while $(\ssource,\sdest)$-flows are analogously defined to $(\source_i,\dest)$-flows with only one source $\ssource$ and corresponding inflow rate $\inflow_{\ssource}:= \sum_{i \in I} \inflow_i$. 
From now on, we will see $h^*$ and $\gop$ as flows in the extended network via the above transformation, i.e.~as vectors in  $L_+(\hori)^{\sRoutes}$ and $L_+(\hori)^\GAS$, respectively.  

We will now construct a walk flow $\tilde{h} \in \swir$ with  corresponding edge flow $\tilde{\g}$ with a  strictly  lower total travel time than $\gop$ and  which  is still supported on $[0,t_f']$. As the 
  above bijection preserves the total travel time, this then yields a contradiction to $\gop$ being \sysop[.] 

As our first step we aim to show that
particles traveling along a $\dest$-cycle before entering the edge towards the super \sink can be rerouted by directly taking the latter edge. Because of our monotonicity assumptions, the resulting flow is then an  $(\ssource,\sdest)$-flow $\g^\Delta \in L_+(\hori)^\GAS$   in which particles may wait at nodes $v \in \GV$, i.e.~for which $\opp_v\G^\Delta(t)\leq 0$ for all $t \in \hori$ and $v \in \GV$.

We make this mathematically precise in the following: There exists, by \Cref{lem: Relations:h>0u>0} and $\g^*_{\arc^*}>0$, a walk $\wa \in \sRoutes$, $j^* \leq \abs{\wa}-1$ with $\wa[j^*] = \arc^*$ and $h^*_\wa(t)>0$ for almost all $t$ in a non-null set $\mathfrak T$. 
The walk $c:=(\wa[j^*],\ldots,\wa[\abs{\wa}-1])$ is then a $\dest$-cycle and we define  
the flow $\g^c \in L_+(\hori)^\GAS$ representing the flow on the cycle $c$ induced by $h^*_\wa$ (under the fixed travel times induced by~$u \coloneqq \gop$) via  
    \begin{align*}
        \g^c_\arc \coloneqq \begin{cases}
            \sum_{j\geq j^*:\wa[j] = \arc}\ell^u_{\wa,j}(h^*_\wa), &\text{for } \arc \in \GA \\
            0,                                                      &\text{else.}
        \end{cases}
    \end{align*}
Next, we construct a flow with waiting $\g^\Delta \in L_+(\hori)^\GAS$ from $\gop$ by sending the flow entering the $\dest$-cycle~$c$ directly towards the super-\sink instead and reducing the flow on all edges of~$c$ accordingly:
    \begin{align*}
        \g^\Delta_\arc \coloneqq \gop_\arc - \g^c_\arc +1_{(\dest,\sdest)}(\arc) \cdot(-\ell^u_{\wa,\abs{\wa}}(h^*_\wa)+\ell^u_{\wa,j^*}(h^*_\wa)) = \begin{cases}
            \gop_\arc-\ell^u_{\wa,\abs{\wa}}(h^*_\wa)+\ell^u_{\wa,j^*}(h^*_\wa), &\text{if } \arc = (\dest,\sdest) \\
            \gop_\arc-\g^c_\arc,                                                   &\text{else.}
        \end{cases}
    \end{align*}
The flow $\g^\Delta$ is a flow with waiting and strictly less time spend in the network in total. A key ingredient for the proof (see \Cref{sec: ProofsSsSysOptNoDOutflow}) is the monotonicity of the travel times stated in \Cref{ass: TravelMono}. 
\begin{claim} \label{claim: FlowCarryingDestCycle}
    The flow $\g^\Delta$ is an $(\ssource,\sdest)$-flow with (potential) waiting at nodes $v \in \GV$. Moreover, it satisfies $\opp_\sdest \G^\Delta(t) \leq \opp_\sdest \G^*(t)$ for all $t \in\hori$ with strict inequality for a nonempty subset of $\hori$. 
\end{claim}

Now, $\g^\Delta$ is a flow with waiting and with support contained in that of $\gop$ and, hence, contained in $[0,t_f']$. Thus, 
we can apply \Cref{lem: WlogNoWaiting} to $\g^\Delta$ in order to transform the latter into  an $(\ssource$,$\sdest)$-flow without waiting.
More precisely, we arrive at a flow~$\tilde{\g}$ supported on $[0,t_f']$ that is an $(\ssource$,$\sdest)$-flow with net outflow rate $\sum_{i\in I} \inflow_i$ at $\ssource$, fulfilling flow conservation at all nodes $v \in \GV$ and having a non-positive flow balance $\op_{\sdest} \tilde{\g}\leq 0$ (for $u=\tilde{\g}$).  
For the last statement, observe that for all $ t\in \hori$ (cf.~\Cref{sec:uBasedNetworkLoadings:FlowBala}):
\begin{align*}
    \op_{\sdest} \tilde{\g}([0,t]) &:= \sum_{\arc \in \edgesFrom{\sdest}}\int_0^t \tilde{\g}_\arc \di\sigma - 
    \sum_{\arc \in \edgesTo{\sdest}}\int_{\exit_\arc(u,\cdot)^{-1}([0,t])} \tilde{\g}_\arc \di\sigma \\
    &=  \opp_\sdest \tilde{\G}(t) \leq \opp_\sdest {\G}^\Delta(t) \leq \opp_\sdest \Gop(t) \leq 0. 
\end{align*}
 Hence, by the flow decomposition theorem (\Cref{thm: FlowDecomp}), we can find a walk flow $\tilde{h} \in L_+(\hori)^{\sRoutes}$ with $\sum_{\wa \in \sRoutes} h_\wa= \sum_{i \in I}\inflow_i$ inducing $\tilde{\g}$ under the parametrized network loading operator $\ell^u$ for $u = \tilde{\g}$. Note that there are no zero cycle inflow rates due to the travel time on all $\arc \in \GA$ being bounded from below by $\tmin >0$ on the support of $\tilde{\g}$. Moreover, by \Cref{ass: LuImplL}, we also get that $\ell(\tilde{h}) = \tilde{\g}$ from which $\tilde{h} \in \swir$ follows immediately as $\tilde{\g}$ is supported on $\hori'$. 
 It now remains to show that 
 $\dup{\trav(\tilde{\g},\cdot)}{{\tilde\g}} < \dup{\trav(\gop,\cdot)}{{\gop}} $ in the extended network. Note that this is enough obtain the desired contradiction to the optimality of $\gop$ in the original network as translating both flows back to the original network preserves their respective total travel times (cf.~the discussion at the start of the proof).  
 To this end, we make use of the following claim which can be shown by a direct computation using several properties of parametrized flows (cf.\ \Cref{sec: ProofsSsSysOptNoDOutflow}): 
 
\begin{claim}\label{claim: TotalTravelRewrite}
For any   $\g \in \ell(\swir) \subseteq L_+(\hori)^\GAS$, we can rewrite the total travel time as follows  
\begin{align*}
    \dup{\trav(\g,\cdot)}{{\g}} =  \int_\hori (t_f - \id) \cdot \sum_{i \in I}\inflow_i \di\sigma  + \int_{\hori}\opp_\sdest\G \di\sigma.
\end{align*}
\end{claim}

Now, by \Cref{claim: FlowCarryingDestCycle} and \Cref{lem: WlogNoWaiting}, we have 
\begin{align*}
      \opp_\sdest \tilde{\G}(t) \leq \opp_\sdest {\G}^\Delta(t) \leq \opp_\sdest \Gop(t)  
\end{align*}
for all $t \in \hori$ with the second inequality being strict for a non-empty interval. Hence, by using the representation of the total travel time in \Cref{claim: TotalTravelRewrite} for both~$\tilde{\g}$ and~$\gop$ we get
    \begin{align*}
        \dup{\trav(\tilde{\g},\cdot)}{\tilde{\g}} 
            &= \int_\hori (t_f - \id) \cdot \sum_{i \in I}\inflow_i \di\sigma  + \int_{\hori}\opp_\sdest\tilde\G \di\sigma\\
            &< \int_\hori (t_f - \id) \cdot \sum_{i \in I}\inflow_i \di\sigma  + \int_{\hori}\opp_\sdest\Gop \di\sigma 
                = \dup{\trav(\gop,\cdot)}{{\gop}},
    \end{align*}
    which is a contradiction to $\gop$ being a system optimum.
 \end{proof}

With \Cref{thm: SsSysOptNoDOutflow} at hand, the implementability of system optima (i.e.\ \Cref{cor: SsSysOptImpl}) is now a direct consequence of \Tollref{thm: CombiChara}:

\begin{proof}[Proof of \Cref{cor: SsSysOptImpl}]
    Let $\gop \in \ell(\wir)$ be a \sysop flow. Since, due to \Cref{ass: TravelGeneral}, the travel times are strictly positive on the support of~$\gop$, \Cref{thm: CombiChara} characterizes the implementability of $\gop$ via the property that $\gop$ never sends any flow into edges leaving the common \sink $\dest$, i.e.~$\gop_\arc=0$ for all $\arc \in \delta^+(\dest)$. 
    By \Cref{thm: SsSysOptNoDOutflow}, any system optimal flow fulfills this property and is, thus,   implementable. 
\end{proof}

\section{Conclusion}
We considered the question whether system optimal flows are implementable by tolls and provided a positive and perhaps surprising negative result: In multi-source, multi-\sink networks, system optimal flows are in general not implementable, even for well-behaved flow propagation models such as the Vickrey queuing  model  and common value-of-time parameter. In particular, we showed that the price of stability achievable by tolls is in fact unbounded, demonstrating a substantial gap between the dynamic and static case. 
In contrast, for single-source, single-\sink networks, we provided sufficient monotonicity assumptions regarding the edge travel times such that system optimal flows are guaranteed to be implementable. These monotonicity assumptions are, in particular, fulfilled by  the Vickrey queuing and linear edge delay model. 

Several questions remain open. For \emph{multi-source}, single-\sink networks we could show in~\Cref{thm: SsSysOptNoDOutflow}
that any system optimal flow does not send positive flow 
out of the \sink[.] This, however, does not directly imply that
those system optima are implementable. One problem is that key properties  of the master problem~\eqref{opt: Master}, like strong duality,  are not clear a priori but would be needed to obtain implementing prices. Another open topic is to consider computational aspects of obtaining implementing tolls.
\clearpage

\bibliographystyle{plain}
\bibliography{master-bib}

\begin{thebibliography}{10}

\bibitem{Anshelevich08}
E.~Anshelevich, A.~Dasgupta, J.~Kleinberg, {\'E}.~Tardos, T.~Wexler, and
  T.~Roughgarden.
\newblock The price of stability for network design with fair cost allocation.
\newblock {\em SIAM J. Comput.}, 38(4):1602--1623, 2008.

\bibitem{BhaskarFA15}
Umang Bhaskar, Lisa Fleischer, and Elliot Anshelevich.
\newblock A {S}tackelberg strategy for routing flow over time.
\newblock {\em Games and Economic Behavior}, 92:232--247, 2015.

\bibitem{Bogachev2007I}
Vladimir~I. Bogachev.
\newblock {\em Measure Theory}, volume~I.
\newblock Springer Science \& Business Media, Berlin Heidelberg, 2007.

\bibitem{ColeDR03}
Richard Cole, Yevgeniy Dodis, and Tim Roughgarden.
\newblock Pricing network edges for heterogeneous selfish users.
\newblock In Lawrence~L. Larmore and Michel~X. Goemans, editors, {\em
  Proceedings of the 35th Annual {ACM} Symposium on Theory of Computing, June
  9-11, 2003, San Diego, CA, {USA}}, pages 521--530. {ACM}, 2003.

\bibitem{CominettiCL15}
Roberto Cominetti, Jos{\'{e}}~R. Correa, and Omar Larr{\'{e}}.
\newblock Dynamic equilibria in fluid queueing networks.
\newblock {\em Oper. Res.}, 63(1):21--34, 2015.

\bibitem{CCOPoAforNashFlows}
Jos\'{e} Correa, Andr\'{e}s Cristi, and Tim Oosterwijk.
\newblock On the price of anarchy for flows over time.
\newblock {\em Mathematics of Operations Research}, 47(2):1394--1411, 2022.

\bibitem{guide2006infinite}
Charalambos~Aliprantis D. and Kim~C. Border.
\newblock {\em Infinite dimensional analysis: A hitchhiker's guide}.
\newblock Springer-Verlag Berlin and Heidelberg GmbH \& Co. KG, 2006.

\bibitem{Fleischer04}
L.~Fleischer, K.~Jain, and M.~Mahdian.
\newblock Tolls for heterogeneous selfish users in multicommodity networks and
  generalized congestion games.
\newblock In {\em Proc. 45th Annual IEEE Sympos. Foundations Comput. Sci.},
  pages 277--285, 2004.

\bibitem{FrascariaO22}
Dario Frascaria and Neil Olver.
\newblock Algorithms for flows over time with scheduling costs.
\newblock {\em Math. Program.}, 192(1):177--206, 2022.

\bibitem{Friesz93}
Terry~L. Friesz, David Bernstein, Tony~E. Smith, Roger~L. Tobin, and B.~W. Wie.
\newblock A variational inequality formulation of the dynamic network user
  equilibrium problem.
\newblock {\em Oper. Res.}, 41(1):179--191, January 1993.

\bibitem{GrafThesis}
Lukas Graf.
\newblock {\em Dynamic Network Flows with Adaptive Route Choice based on
  Current Information}.
\newblock Mathematische Optimierung und Wirtschaftsmathematik | Mathematical
  Optimization and Economathematics. Springer Spektrum Wiesbaden, 2024.

\bibitem{GHS24FD}
Lukas Graf, Tobias Harks, and Julian Schwarz.
\newblock A decomposition theorem for dynamic flows, 2024.
\newblock Preprint on arXiv: \url{https://arxiv.org/abs/2407.04761}.

\bibitem{GHS24TollSODA}
Lukas Graf, Tobias Harks, and Julian Schwarz.
\newblock Tolls for dynamic equilibrium flows.
\newblock In {\em Proceedings of the 2025 Annual ACM-SIAM Symposium on Discrete
  Algorithms (SODA)}, pages 2560--2606. {SIAM}, 2025.

\bibitem{GHKKInformationDesign}
Svenja~M. Griesbach, Martin Hoefer, Max Klimm, and Tim Koglin.
\newblock Information design for congestion games with unknown demand.
\newblock {\em Proceedings of the AAAI Conference on Artificial Intelligence},
  38(9):9722--9730, March 2024.

\bibitem{HarksS23}
Tobias Harks and Julian Schwarz.
\newblock A unified framework for pricing in nonconvex resource allocation
  games.
\newblock {\em SIAM J. Optim.}, 33(2):1223--1249, 2023.

\bibitem{Karakostas04}
G.~Karakostas and S.~Kolliopoulos.
\newblock Edge pricing of multicommodity networks for heterogeneous selfish
  users.
\newblock In {\em Proc. 45th Annual IEEE Sympos. Foundations Comput. Sci.},
  pages 268--276, 2004.

\bibitem{Knight1924}
F.~Knight.
\newblock Some fallacies in the interpretation of social cost.
\newblock {\em Quart. J. Econ.}, 38(4):582--606, 1924.

\bibitem{Koch11}
Ronald Koch and Martin Skutella.
\newblock Nash equilibria and the price of anarchy for flows over time.
\newblock {\em Theory Comput. Syst.}, 49(1):71--97, 2011.

\bibitem{Lombardi21}
Claudio Lombardi, Luis Picado-Santos, and Anuradha~M. Annaswamy.
\newblock Model-based dynamic toll pricing: An overview.
\newblock {\em Applied Sciences}, 11(11), 2021.

\bibitem{Ma2017}
Rui Ma, Xuegang~(Jeff) Ban, and W.Y. Szeto.
\newblock Emission modeling and pricing on single-destination dynamic traffic
  networks.
\newblock {\em Transportation Research Part B: Methodological}, 100:255--283,
  2017.

\bibitem{MeunierW10}
Fr{\'{e}}d{\'{e}}ric Meunier and Nicolas Wagner.
\newblock Equilibrium results for dynamic congestion games.
\newblock {\em Transportation Science}, 44(4):524--536, 2010.
\newblock An updated version (2014) is available on Arxiv.

\bibitem{Pigou20}
A.~Pigou.
\newblock {\em The Economics of Welfare}.
\newblock Macmillan, London, UK, 1920.

\bibitem{Vickrey69}
William~S Vickrey.
\newblock Congestion theory and transport investment.
\newblock {\em American Economic Review}, 59(2):251--60, May 1969.

\bibitem{WieTobin1998}
Byung-Wook Wie and Roger~L. Tobin.
\newblock Dynamic congestion pricing models for general traffic networks.
\newblock {\em Transportation Research Part B: Methodological}, 32(5):313--327,
  1998.

\bibitem{Yang04}
H.~Yang and H.-J. Huang.
\newblock The multi-class, multi-criteria traffic network equilibrium and
  systems optimum problem.
\newblock {\em Transportation Res.}, 38(B):1--15, 2004.

\bibitem{Yang1998}
Hai Yang and Qiang Meng.
\newblock Departure time, route choice and congestion toll in a queuing network
  with elastic demand.
\newblock {\em Transportation Research Part B: Methodological}, 32(4):247--260,
  1998.

\bibitem{Yang2005}
Hailiang Yang and Haijun Huang.
\newblock {\em Mathematical and Economic Theory of Road Pricing}.
\newblock Emerald Group Publishing Limited, 2005.

\bibitem{ZhuM00}
Daoli Zhu and Patrice Marcotte.
\newblock On the existence of solutions to the dynamic user equilibrium
  problem.
\newblock {\em Transportation Sci.}, 34(4):402--414, 2000.

\end{thebibliography}

\appendix 

\section{Omitted Formal Proofs} \label{sec: Proofs}
In this section, we provide several omitted formal proofs. 
For some of these proofs,  we require some additional concepts from measure zero -- namely two types of standard (Borel-)measures  which we introduce in the following: 
Firstly, for any measurable function $\g:\hori\to \R$, we denote by $\g \cdot \sigma$ the measure on $\mathcal{B}(\hori)$
given by $\g\cdot\sigma(\mathfrak T) := \int_{\mathfrak T} \g\di\sigma,\mathfrak T \in \mathcal{B}(\hori)$. Secondly, for any measurable function $A: \hori\to\hori$ and any measure $\mu$ on $\mathcal{B}(\hori)$, we denote by $\mu\circ A^{-1}$ the image measure of $\mu$ under $A$ which is defined by $\mu\circ A^{-1}(\mathfrak T) := \mu(A^{-1}(\mathfrak T))$. The latter is again a measure on $\mathcal{B}(\hori)$. 
Moreover, we will use the notation of $\mu\leq \mu'$ for two measures, meaning that $\mu(\mathfrak T) \leq \mu'(\mathfrak T)$ for all $\mathfrak T \in \mathcal{B}(\hori)$.
We refer to~\cite{Bogachev2007I} for a comprehensive overview of measure theory. 

\subsection{Proofs Omitted in \Cref{thm: CounterSysopImpl}}\label{app:ProofsOfCounterSysopImplClaims}
    \begin{proofClaim}[Proof of \Cref{claim: travarc}]
  \begin{structuredproof}
        \proofitem{\ref{claim: travarc: 4>eps+infl}} Flow conservation at node~$\source_\colNo{F}$ (which is not a \sink node of any commodity) together with the fact that Vickrey queues are always non-negative gives us for any $t \in [9,13+M]$:
        \begin{align*}
            \q_{\arc_4}(u,t) 
                &= \int_0^t u_{\arc_4}\di\sigma - \int_0^{t+\tau_{\arc_4}}u^-_{\arc_4}\di\sigma 
                    = \q_{\arc_4}(u,9) + \int_9^t u_{\arc_4}\di\sigma - \int_{9+\tau_{\arc_4}}^{t+\tau_{\arc_4}}u^-_{\arc_4}\di\sigma \\
                &\geq \int_9^t u_{\arc_4}\di\sigma - \int_{9+\tau_{\arc_4}}^{t+\tau_{\arc_4}}u^-_{\arc_4}\di\sigma 
                    \symoverset{1}{=} \int_9^t u^-_{\arc_5}\di\sigma + \int_9^t r_\colNo{F} \di\sigma - \int_{9+\tau_{\arc_4}}^{t+\tau_{\arc_4}}u^-_{\arc_4}\di\sigma \\
                &\symoverset{2}{\geq} \int_9^t u^-_{\arc_5}\di\sigma + (t-9)\cdot 2 - (t-9)\cdot 2
                    = \int_9^t u^-_{\arc_5}\di\sigma.
        \end{align*}
        Here, we used flow conservation at~\refsym{1} and the fact that in a Vickrey flow the outflow rate of an edge is bounded by its capacity (which is $2$ for $\arc_5$) at~\refsym{2}. This now directly implies
            \[\trav_{\arc_4}(u,t) = \tau_{\arc_4} + \frac{\q_{\arc_4}(u,t)}{\nu_{\arc_5}} \geq \varepsilon + \frac{1}{2}\int_9^t u^-_{\arc_5}\di\sigma\]
        as claimed.

        \proofitem{\ref{claim: travarc: 3>eps+infl}} This can be shown completely analogous to \ref{claim: travarc: 4>eps+infl} using flow conservation at node~$\source_\colNo{D}$ (which is not a \sink node for commodities~\colNo{D} or~\colNo{E}).      
    
        \proofitem{\ref{claim: travarc: 2QueueUpperBound}} 
        The only walks containing edge~$\arc_2$ are those of commodities~\colNo{C} and~\colNo{E}. Moreover, each of these walks contains this edge exactly once. As the combined cumulative network inflow of these two commodities is~$3$, this is also an upper bound on cumulative inflow into and, hence, the size of the queue on edge~$\arc_2$ (at any point in time). 
        
        \proofitem{\ref{claim: travarc: 4NoQueue}} Since the only flow arriving at node~$\source_\colNo{F}$ before time~$9$ is via edge~$\arc_5$, flow conservation at this node ensures $u_{\arc_4}(t) \leq u^-_{\arc_5}(t) \leq \nu_{\arc_5} = \nu_{\arc_4}$ for all $t \in [0,9]$. Hence, no queue can build up during $[0,9]$, which implies $\trav_{\arc_4}(u,t) = \tau_{\arc_4} = \varepsilon$ for all $t \in [0,9]$.

        \proofitem{\ref{claim: travarc: 1NoQueue}} As node $\source_\colNo{E}$ has no incoming edges, flow conservation at this node ensures $u_{\arc_1} \leq r_\colNo{E} \leq 1 = \nu_{\arc_1}$ for all times. Hence, no queue can ever build up on edge~$\arc_1$ implying $\trav_{\arc_1}(u,\cdot) = \tau_{\arc_1} = 1$. Completely analogous one shows $\trav_{\arc_7}(u,\cdot) = \tau_{\arc_7} = 1$. Finally, $\trav_{\arc_6}(u,\cdot) = \tau_{\arc_6} = 4$ holds since the capacity of edge~$\arc_6$ is unbounded and, hence, no queue ever builds up on this edge.
        \qedhere
    \end{structuredproof}
    \end{proofClaim}

 \begin{proofClaim}[Proof of \Cref{claim: ImplemenatbleImplies}]
               \begin{structuredproof}
            \proofitem{\ref{claim: ImplemenatbleImplies: <=3}} Assume, for the sake of a contradiction, that there exists some measurable non-null set $\mathfrak T \subseteq \hori$ such that we have $u^{\colNo{B}}_{\arc_5}(t) > 0$ for almost all $t \in \mathfrak T$. Then, by \Cref{lem: Relations:h>0u>0}, there must be some subset $\mathfrak D \in \mathcal{B}(\hori)$ with positive measure $\sigma(\mathfrak D) >0$ and a walk $\wa^\colNo{B} \in \Routes_\colNo{B}$  with $\wa^\colNo{B}[2] = \arc_5$ such that $h_{\wa^\colNo{B}}>0$ on~$\mathfrak D$. By choosing a suitable subset of $\mathfrak D$ (still with positive measure) we may even assume  $h_{\wa^\colNo{B}}(t) > 0$ for almost all $t \in \mathfrak D$. Then, we can define a new walk inflow rate $\hat{h}$ via
            \begin{align*}
                \hat{h}_\wa \coloneqq \begin{cases}
                    0,                             &\text{for } \wa = \wa^{\colNo{B}}, \\
                    h_\wa + h_{\wa^\colNo{B}},     &\text{for } \wa = (\arc_7,\colNo{B}), \\
                    h_\wa,                         &\text{else,}
                \end{cases} \quad \quad \text{ for all } \wa \in \Routes.
            \end{align*}
            It is then easily verified that $\hat{h} \in  \wir^u$ and  $\ell^u(\hat{h})< u$. Since $\trav_\arc(u,t) > 0$ for all $t \in \hori'$ and $\arc \in \GA$, it follows that 
            $\dup{\trav(u,\cdot)}{\ell^u(\hat{h})} < \dup{\trav(u,\cdot)}{u} $. This shows that $\hat{h}$ is feasible for the master problem \ref{opt: Master} and has a better objective value than any $\tilde{h}$ inducing $u$ since we can rewrite the objective value for any feasible $h'$  via  \Cref{lem: aggCostsVSwalkCosts} by 
            $\dup{\Psi(u,\cdot)}{h'} = \dup{\trav(u,\cdot)}{\ell^u(h')}$.  
            Thus, we arrive at the desired contradiction as this shows, by \Cref{thm: ImplementableImpliesOptimal}, that $u$ is not implementable. 
            
            \proofitem{\ref{claim: ImplemenatbleImplies: Trav=1}} Since we have $\trav_{\arc_5}(u,t)=\tau_{\arc_5}+\frac{\q_{\arc_5}(u,t)}{\nu_{\arc_5}} = 2-\varepsilon + \frac{\q_{\arc_5}(u,t)}{\nu_{\arc_5}}$, it suffices to show that the queue on edge~$\arc_5$ remains empty during $[0,4]$. This, in turn, follows by observing that the inflow rate into this edge is at most~$2$ during any point in this interval: No flow can arrive at the tail of edge~$\arc_5$ via edge~$\arc_6$ during that time since $\trav_{\arc_6}(u,t) \geq \tau_{\arc_6}=4$. Therefore, we only have to consider flow of commodity~\colNo{C} entering the network at node $s_\colNo{C}$ at a rate of at most~$2$ and flow of commodity~\colNo{B}. The latter, however, cannot enter edge~$\arc_5$ by \ref{claim: ImplemenatbleImplies: <=3}.  \qedhere
        \end{structuredproof}
    \end{proofClaim}

   \begin{proofClaim}[Proof of \Cref{claim: ImplemenatbleImpliesHighTT}]
        The flow particles of commodities~\colNo{C} and~\colNo{E} arriving at $\dest_\colNo{F}$ either directly enter $\arc_2$ or go along the cycle $(\arc_6,\arc_5,\arc_4)$ at least once before entering $\arc_2$. In particular, the following is a complete case distinction: Either at least $7/8$ of commodity~\colNo{C}'s and $3/4$ of commodity~\colNo{E}'s  particles enter $\arc_2$ directly or at least $1/8$ of commodity~\colNo{C}'s or $1/4$ of commodity~\colNo{E}'s particles take at least once the cycle $(\arc_6,\arc_5,\arc_4)$. 

        \begin{proofbycases}
            \caseitem{At least $7/8$ of commodity~\colNo{C} and $3/4$ of commodity~\colNo{E} enter $\arc_2$ directly}
            In this case, 
                \[\int_{[2,3]}h_{\wa^\colNo{C}} \di\sigma \geq 7/8\cdot \int_{[2,3]}r_\colNo{C}\di\sigma = 7/4\]
            for $\wa^\colNo{C}=((\arc_5,\arc_4,\arc_2),\colNo{C})$ and 
                \begin{align}\label{eq:DirectFlowC5}
                    \int_{[4,5]}h_{\wa^\colNo{E}} \di\sigma \geq 3/4\cdot \int_{[4,5]}r_\colNo{E}\di\sigma = 3/4
                \end{align}
            for $\wa^\colNo{E}=((\arc_1,\arc_2,\arc_3),\colNo{E})$. 
            Using the former we can derive the following lower bound on the total flow of commodity~\colNo{C} entering $\arc_2$ during $[0,5]$: 
            \begin{align}
                \int_{[0,5]}u^{\colNo{C}}_{\arc_2}\di\sigma &= \sum_{\wa \in \Routes_\colNo{C}} \sum_{j:\wa[j]=\arc_2} \int_{[0,5]} \ell^u_{\wa,j}(h_\wa)\di\sigma 
                &\geq \int_{[0,5]} \ell^u_{\wa^\colNo{C},3}(h_{\wa^\colNo{C}})\di\sigma =  \int_{\arr_{\wa^\colNo{C},3}(u,\cdot)^{-1}([0,5])}  h_{\wa^\colNo{C}} \di\sigma \nonumber \\
                &= \int_{[0,3]} h_{\wa^\colNo{C}}\di\sigma  \geq 7/4,\label{eq: comm3e2}
            \end{align} 
            where the last equality holds by \Cref{claim: travarc: 4NoQueue}, \Cref{claim: ImplemenatbleImplies: Trav=1} and \Cref{lem: elluinj}. 
            Moreover, we have $u_{\arc_2}(t) = 0$ on $[0,4]$ since no flow has the possibility to reach $\arc_2$ during this interval (only flow of commodities~\colNo{C} and~\colNo{E} can enter this edge and neither can reach it before time~$4$). Together with the non.negativity of Vickrey queues this ensures $u^-_{\arc_2}=0$ on~$[0,5]$. Therefore, we get
            \begin{align}\label{eq:LowerBoundQueueE2At5}
                 \q_{\arc_2}(u,5) 
                    = \int_{[0,5]} u_{\arc_2} \di\sigma - \int_{[0,5+\tau_{\arc_2}]}u_{\arc_2}^-\di\sigma 
                            \symoverset{1}{\geq}  7/4 -  1 = 3/4,
            \end{align}
            where we used for \refsym{1} the inequality in \eqref{eq: comm3e2}  and that in Vickrey flows the outflow rate is bounded by the edge's capacity (which is $1$ for~$\arc_2$). 

            Now, since the travel time on edge~$\arc_1$ is always~$1$ (\Cref{claim: travarc: 1NoQueue}) and no flow of commodity~\colNo{E} can arrive at node~$\dest_\colNo{F}$ via any other edge during $[5,6]$, we have $u^{\colNo{E}}_{\arc_2} \leq u^{\colNo{E},-}_{\arc_1} = 1$ during that interval. This, together with~\eqref{eq:DirectFlowC5} allows us to derive the following lower bound on the cumulative inflow of commodity~\colNo{E} into edge~$\arc_2$ up to any time $t \in [5,6]$:
            \begin{align*}
                \int_{[5,t]}u^{\colNo{E}}_{\arc_2}\di\sigma 
                    &= \int_{[5,6]}u^{\colNo{E}}_{\arc_2}\di\sigma - \int_{[t,6]}u^{\colNo{E}}_{\arc_2}\di\sigma 
                        \geq \int_{[5,6]}\ell^u_{\wa^\colNo{E},2}(h_{\wa^\colNo{E}})\di\sigma - \int_{[t,6]}1\di\sigma \\
                    &= \int_{\arr_{\wa^\colNo{E},2}(u,\cdot)^{-1}([5,6])}h_{\wa^\colNo{E}}\di\sigma - (6-t) 
                        \symoverset{2}{=} \int_{[4,5]}h_{\wa^\colNo{E}}\di\sigma - (6-t) \\
                    &\toverset{\eqref{eq:DirectFlowC5}}{\geq} 3/4 - (6-t) = t-6+3/4,
            \end{align*}
            where \refsym{2} holds due to \Cref{claim: travarc: 1NoQueue} and \Cref{lem: elluinj}. Combining this with~\eqref{eq:LowerBoundQueueE2At5} gives us the following lower bound on the the queue at $\arc_2$ during $[5,6]$:
            \begin{align*}
                \q_{\arc_2}(u,t) 
                    &= \int_{[0,t]} u_{\arc_2} \di\sigma - \int_{[0,t+1]}u_{\arc_2}^-\di\sigma \\ 
                    &= \q_{\arc_2}(u,5)  + \int_{[5,t]} u_{\arc_2} \di\sigma - \int_{[6,t+1]}u_{\arc_2}^-\di\sigma \\
                    &\toverset{\eqref{eq:LowerBoundQueueE2At5}}{\geq} 3/4 + (t-6+3/4) - (t-5) = 1/2,
            \end{align*}
        where we also used $u_{\arc_2}^- \leq \nu_2 = 1$ for the last inequality. 
        Hence, we have $\trav_{\arc_2}(u,t)\geq 3/2$ for all $t \in [5,6]$.  On the other hand, we know by \Cref{claim: travarc: 2QueueUpperBound} that $\trav_{\arc_2}(u,t)\leq 4$ for all $t \in [5,6]$. Together, this shows $\exit_{\arc_2}(u,\cdot)^{-1}([7,10]) \supseteq [7-3/2,10-4] = [5.5,6]$, which allows us to conclude that 
        \begin{align*}
            \int_{[7,10]}u_{\arc_3}^{\colNo{E}}  \di\sigma 
                &= \int_{\exit_{\arc_2}(u,\cdot)^{-1}([7,10])} u_{\arc_2}^{\colNo{E}}\di\sigma 
                    \geq  \int_{[5.5,6]} u_{\arc_2}^{\colNo{E}}\di\sigma\\ 
                &=\int_{[5,6]} u_{\arc_2}^{\colNo{E}}\di\sigma - \int_{[5,5.5]} u_{\arc_2}^{\colNo{E}}\di\sigma 
                    \geq   3/4  - \int_{[5,5.5]} 1 \di\sigma 
                    = 1/4 .
        \end{align*}
        where we used in the second inequality that $u_{\arc_2}^{\colNo{E}}\leq \nu_{\arc_2} = 1$.  
        With this, \Cref{claim: travarc: 3>eps+infl} directly implies $\trav_{\arc_3}(u,t) \geq \varepsilon+1/4$ for all $t \in [10,10+M]$, which finally allows us to derive the desired lower bound on the total travel time incurred by commodity~\colNo{D} as follows: 
        \begin{align*}
            \dup{\trav(u,\cdot)}{u^{\colNo{D}}} 
                &= \int_{[7,10]} \trav_{\arc_3}(u,\cdot)\cdot u^{\colNo{D}}_{\arc_3}\di\sigma + \int_{[10,10+M]} \trav_{\arc_3}(u,\cdot)\cdot u^{\colNo{D}}_{\arc_3}\di\sigma \\
                &\geq \int_{[7,10]}\varepsilon\cdot 1 + \int_{[10,10+M]} (\varepsilon+1/4)\cdot 1\di\sigma 
                    = 3\varepsilon + M\varepsilon + M/4.
        \end{align*}

        \caseitem{At least $1/8$ of commodity~\colNo{C} or $1/4$ of commodity~\colNo{E} take the cycle at least once} In this case we have
            \[\int_{[2,3]}\sum_{\wa \in \Routes_{\colNo{C}}:\wa \neq \wa^{\colNo{C}}} h_{\wa} \di\sigma + \int_{[4,5]}\sum_{\wa \in \Routes_{\colNo{E}}:\wa \neq \wa^{\colNo{E}}} h_{\wa} \di\sigma \geq 1/4\]
        for $\wa^{\colNo{C}}=((\arc_5,\arc_4,\arc_2),\colNo{C})$ and $\wa^{\colNo{E}}=((\arc_1,\arc_2,\arc_3),\colNo{E})$. 
        For the total (combined) flow of commodities~\colNo{C} and~\colNo{E} entering $\arc_5$ during $[8,10]$ we then get: 
        \begin{align*}
            \int_{[8,10]}u^{\colNo{C}}_{\arc_5}+u^{\colNo{E}}_{\arc_5}\di\sigma 
                &= \sum_{\wa \in \Routes_{\colNo{C}}} \sum_{j:\wa[j]=\arc_5} \int_{[8,10]} \ell^u_{\wa,j}(h_\wa)\di\sigma + \sum_{\wa \in \Routes_{\colNo{E}}} \sum_{j:\wa[j]=\arc_5} \int_{[8,10]} \ell^u_{\wa,j}(h_\wa)\di\sigma \\
                &\geq \sum_{\wa \in \Routes_{\colNo{C}}:\wa \neq \wa^{\colNo{C}}}   \int_{[8,10]} \ell^u_{\wa,4}(h_{\wa})\di\sigma + \sum_{\wa \in \Routes_{\colNo{E}}:\wa \neq \wa^{\colNo{E}}}   \int_{[8,10]} \ell^u_{\wa,3}(h_{\wa})\di\sigma \\
                &= \sum_{\wa \in \Routes_{\colNo{C}}:\wa \neq \wa^{\colNo{C}}}  \int_{\arr_{\wa,4}(u,\cdot)^{-1}([8,10])} h_{\wa} \di\sigma + \sum_{\wa \in \Routes_{\colNo{E}}:\wa \neq \wa^{\colNo{E}}}  \int_{\arr_{\wa,3}(u,\cdot)^{-1}([8,10])} h_{\wa} \di\sigma\\
                &\symoverset{1}{\geq} \sum_{\wa \in \Routes_{\colNo{C}}:\wa \neq \wa^{\colNo{C}}}  \int_{[2,4]} h_{\wa}\di\sigma + \sum_{\wa \in \Routes_{\colNo{E}}:\wa \neq \wa^{\colNo{E}}}  \int_{[3,5]} h_{\wa}\di\sigma
                    \geq 1/4
        \end{align*}   
        For \refsym{1}, we used that any flow entering the cycle $(\arc_5,\arc_4,\arc_6)$ during $[2,4]$ or the path $(\arc_1,\arc_6)$ during $[3,5]$ arrives at the start of edge~$\arc_5$ during the interval $[8,10]$, i.e.~$\arr_{\wa,4}(u,\cdot)^{-1}([8,10]) \supseteq [2,4]$ for all $\wa \in \Routes_{\colNo{C}}\setminus\{\wa^{\colNo{C}}\}$ and $\arr_{\wa,3}(u,\cdot)^{-1}([8,10]) \supseteq [3,5]$ for all $\wa \in \Routes_{\colNo{E}}\setminus\{\wa^{\colNo{E}}\}$. To see the former, note that \Cref{claim: travarc: 4NoQueue} and \Cref{claim: ImplemenatbleImplies: Trav=1} together imply that the travel time along cycle $(\arc_5,\arc_4,\arc_6)$ is exactly $2-\varepsilon+\varepsilon+4=6$ during the interval $[0,4]$. Similarly, for the latter, \Cref{claim: travarc: 1NoQueue} ensures that the travel time along path $(\arc_1,\arc_6)$ is always exactly $1+4=5$.
        
        This now allows us to lower bound the volume of flow on edge~$\arc_5$ at time~$10$ as follows: 
        \begin{align*}
            \int_0^{10} u_{\arc_5}\di\sigma - \int_0^{10} u^-_{\arc_5}\di\sigma 
                &= \int_0^{7+\varepsilon} u_{\arc_5}\di\sigma - \int_0^{7+\varepsilon+\tau_{\arc_5}} u^-_{\arc_5}\di\sigma + \int_{7+\varepsilon}^{10} u_{\arc_5}\di\sigma - \int_{7+\varepsilon+\tau_{\arc_5}}^{10} u^-_{\arc_5}\di\sigma \\
                &= \q_{\arc_5}(u,7+\varepsilon) +  \int_{7+\varepsilon}^{10} u_{\arc_5}\di\sigma - \int_{7+\varepsilon+\tau_{\arc_5}}^{10} u^-_{\arc_5}\di\sigma \\
                &\symoverset{1}{\geq} \int_8^{10} u_{\arc_5}^{\colNo{D}} + u_{\arc_5}^{\colNo{F}}\di\sigma - \int_{9}^{10} u^-_{\arc_5}\di\sigma   
                    \geq 1/4 - \int_{9}^{10} u^-_{\arc_5}\di\sigma. 
        \end{align*}
        Here, we used non-negativity of Vickrey queues and $\varepsilon \leq 1$ as well as $\varepsilon+\tau_{\arc_5}=2$ for \refsym{1}. Hence, at time~$10$ flow of volume at least $x \coloneqq \max\Set{1/4-\int_9^{10}u^-_{\arc_5}\di\sigma,0}$ is currently on edge~$\arc_5$. Since, in a Vickrey flow, flow cannot stay on an edge for too long (\cite[Corollary 3.24]{GrafThesis}), we know that at least that amount of flow has to leave edge~$\arc_5$ during the interval $[10,10+\frac{x}{\nu_{\arc_5}}+\tau_{\arc_5}]$ and, hence,
        \begin{align*}
            \int_{10}^{13} u^-_{\arc_5}\di\sigma \geq \int_{10}^{10+\frac{x}{\nu_{\arc_5}}+\tau_{\arc_5}} u^-_{\arc_5}\di\sigma \geq x.
        \end{align*}
        Altogether this implies
        \begin{align}
            \int_9^{13}u_{\arc_5}^-\di\sigma 
                    = \int_9^{10}u_{\arc_5}^-\di\sigma + \int_{10}^{13}u_{\arc_5}^-\di\sigma 
                &\geq \int_9^{10}u_{\arc_5}^-\di\sigma + x \notag\\
                &= \int_9^{10}u_{\arc_5}^-\di\sigma + \max\Set{1/4-\int_9^{10}u^-_{\arc_5}\di\sigma,0} \geq 1/4. \label{eq:InflowIntoE5}
        \end{align}
        Hence, the claim follows as 
        \begin{align*}
            \dup{\trav(u,\cdot)}{u^{\colNo{F}}} 
                &\geq  \int_9^{13+M}  \trav_{\arc_4}(u,\cdot)\cdot  u^{\colNo{F}}_{\arc_4}\di\sigma 
                    = \int_9^{13}  \trav_{\arc_4}(u,\cdot)\cdot  u^{\colNo{F}}_{\arc_4}\di\sigma + \int_{13}^{13+M}  \trav_{\arc_4}(u,\cdot)\cdot  u^{\colNo{F}}_{\arc_4}\di\sigma \\
                &\symoverset{2}{\geq}  \int_9^{13}\varepsilon\cdot r_{\colNo{F}} \di\sigma + \int_{13}^{13+M}\Bigl(\varepsilon+\tfrac{1}{2}\int_9^{13}u^-_{\arc_5}\di\sigma\Bigr)\cdot r_{\colNo{F}} \di\sigma \\
                &\toverset{\eqref{eq:InflowIntoE5}}{\geq} (13-9)\cdot\varepsilon\cdot 2 + M\cdot \bigl(\varepsilon + \tfrac{1}{2\cdot4}\bigr)\cdot 2 
                    = 2\varepsilon(M+4)+\tfrac{M}{4},
        \end{align*}
        where \refsym{2} holds due to \Cref{claim: travarc: 4>eps+infl} together with flow conservation at node~$\source_{\colNo{F}}$. \qedhere                    
        \end{proofbycases}
    \end{proofClaim}

    \begin{proofClaim}[Proof of \Cref{claim: PropertiesOfSysOpt}]
        The statements about the inflow rates into the first edges of all walks hold directly by the definition of~$h$. In particular, \labelcref{claim: PropertiesOfSysOpt:comm4,claim: PropertiesOfSysOpt:comm6} are trivial.
    \begin{structuredproof}
        \proofitem{\ref{claim: PropertiesOfSysOpt:comm2}} $\g^{\colNo{B}}_{\arc_5}=6\cdot1_{[1,2]}$ follows directly from $\g^{\colNo{B}}_{\arc_7}=\inflow_{\colNo{B}} = 6\cdot1_{[0,1]}$ due to \Cref{claim: travarc: 7NoQueue}. This, in turn, together with the fact that no other commodities' flow can enter edge~$\arc_5$ before time~$2$ and $\nu_{\arc_5}=2$, implies that the queue on that edge grows linearly from~$0$ to~$4$ during that time interval. Consequently, the travel time grows linearly from~$2-\varepsilon$ to~$2-\varepsilon+\frac{4}{2}$, which implies $\g^{\colNo{B}}_{\arc_4} = 2\cdot 1_{[3-\varepsilon,6-\varepsilon]}$. Since we have $\trav_{\arc_4}(g,t)=\varepsilon$ for all $t \in [0,9]$ by \Cref{claim: travarc: 4NoQueue}, this gives us $\g^{\colNo{B}}_{\arc_6} = 2\cdot 1_{[3,6]}$.
        
        \proofitem{\ref{claim: PropertiesOfSysOpt:comm3}} As argued for~\ref{claim: PropertiesOfSysOpt:comm2}, we have $\q_{\arc_5}(\g,2)  = 4$. Since $\g_{\arc_5} = \g^{\colNo{C}}_{\arc_5} = 2\cdot 1_{[2,3]}$ on $[2,3]$ and $\nu_{\arc_5} = 2$, the queue remains at $4$ for that interval, i.e.~$\q_{\arc_5}(\g,t)  = 4$ for $t \in [2,3]$. Thus, $\trav_{\arc_5}(\g,t) = 4- \varepsilon$ for $t \in [2,3]$, resulting in $\g^{\colNo{C}}_{\arc_4} = 2\cdot 1_{[6-\varepsilon,7-\varepsilon]}$. 
        Moreover, we know that  $\trav_{\arc_4}(g,t) = \varepsilon$ for all $t \in [6-\varepsilon,7-\varepsilon]$ (\Cref{claim: travarc: 4NoQueue}) and, hence, $\g^{\colNo{C}}_{\arc_2}= 2\cdot 1_{[6,7]}$. 
            
        \proofitem{\ref{claim: PropertiesOfSysOpt:comm5}} By \Cref{claim: travarc: 1NoQueue}, we have $\trav_{\arc_1}(\g,t)=1$ for all~$t$ and, therefore, $\g^{\colNo{E}}_{\arc_2} = 1_{[5,6]}$. Moreover, by combining the previously shown parts of \Cref{claim: PropertiesOfSysOpt} we deduce that $\g_{\arc_2} = \g^{\colNo{E}}_{\arc_2} =1_{[5,6]}$ on $[0,6]$. Hence, no queue forms on this edge during $[0,6]$ and, consequently, $\trav_{\arc_2}(\g,t)= 1$ for $t \in[0,6]$ and $\g^{\colNo{E}}_{\arc_3} = 1_{[6,7]}$ follow. 
        
        \proofitem{\ref{claim: PropertiesOfSysOpt:trav5}} As argued in the proof of \ref{claim: PropertiesOfSysOpt:comm3}, we have $\trav_{\arc_5}(\g,2) = 4-\varepsilon$. Since $\trav$ fulfills FIFO (i.e.\ $\exit$ is non-decreasing), this directly implies $\trav_{\arc_5}(\g,t) \leq 5-\varepsilon$ for all $t \in [1,2]$. The second part was already shown in the proof of \ref{claim: PropertiesOfSysOpt:comm3}. 
        
        \proofitem{\ref{claim: PropertiesOfSysOpt:trav4}} For $t \in [0,9]$ we have $\trav_{\arc_4}(u,t) = \varepsilon$ for any flow by \Cref{claim: travarc: 4NoQueue} and, hence, in particular for~$g$. Thus, we have $\q_{\arc_4}(g,9)=0$. Since the only flow entering edge~$\arc_4$ after time~$9$ is that of commodity~\colNo{F} (at a rate of $2 = \nu_{\arc_4}$), the queue remains empty after that. Therefore, we have $\trav_{\arc_4}(g,t) = \tau_{\arc_4}=\varepsilon$ for all $t \geq 9$ as well.
        
        \proofitem{\ref{claim: PropertiesOfSysOpt:trav2}} In the proof of \ref{claim: PropertiesOfSysOpt:comm5} we already showed that $\q_{\arc_2}(\g,t)=0$ (and, hence, $\trav_{\arc_2}(g,t)=1$) for all $t \in [5,6]$. Together with
        \begin{align*}
            \int_{[6,7]}\g_{\arc_2}\di\sigma = \int_{[6,7]}\g_{\arc_2}^{\colNo{C}}\di\sigma =  \int_{\hori}\g_{\arc_2}^{\colNo{C}}\di\sigma = 2,
        \end{align*}
        this implies $\q_{\arc_2}(\g,t) \leq 2$ and, therefore, $\trav_{\arc_2}(\g,t) \leq 3$ for all $t \in [6,7]$.

        \proofitem{\ref{claim: PropertiesOfSysOpt:trav3}} This is an immediate consequence of $\g_{\arc_3} = \g^{\colNo{D}}_{\arc_3} + \g^{\colNo{E}}_{\arc_3} = 1_{[7,10+M]} + 1_{[6,7]} \leq 1 = \nu_{\arc_3}$.
        \qedhere  
    \end{structuredproof}
    \end{proofClaim}

   \begin{proofClaim}[Proof of \Cref{claim: PropsSysop}]
        Using \Cref{claim: PropertiesOfSysOpt} we calculate: 
    \begin{align*}
        \dup{\trav(\g,\cdot)}{\g^{\colNo{B}}} &= 1\cdot6 +\int_{[1,2]} \trav_{\arc_5}(\g,\cdot) \cdot \g^{{\colNo{B}}}_{\arc_5} \di\sigma  + \int_{[3-\varepsilon,6-\varepsilon]} \trav_{\arc_4}(\g,\cdot) \cdot \g^{\colNo{B}}_{\arc_4} \di\sigma + 4\cdot 6 \\
                &\Croverset{claim: PropertiesOfSysOpt:trav5,claim: PropertiesOfSysOpt:trav4}{\leq} 6 + (5-\varepsilon)\cdot 6 + \varepsilon\cdot 6 + 24= 60 \\
        \dup{\trav(\g,\cdot)}{\g^{\colNo{C}}} &=\int_{[2,3]} \trav_{\arc_5}(\g,\cdot) \cdot \g^{\colNo{C}}_{\arc_5} \di\sigma  + \int_{[6-\varepsilon,7-\varepsilon]} \trav_{\arc_4}(\g,\cdot) \cdot \g^{\colNo{C}}_{\arc_4} \di\sigma + \int_{[6,7]} \trav_{\arc_2}(\g,\cdot) \cdot \g^{\colNo{C}}_{\arc_2} \di\sigma \\
                &\Croverset{claim: PropertiesOfSysOpt:trav5,claim: PropertiesOfSysOpt:trav2,claim: PropertiesOfSysOpt:trav4}{\leq} (4-\varepsilon)\cdot 2 + \varepsilon\cdot 2 + 3 \cdot 2 = 14 \\
        \dup{\trav(\g,\cdot)}{\g^{\colNo{D}}} &\Croverset{claim: PropertiesOfSysOpt:trav3}{=} (3+M)\varepsilon \\
        \dup{\trav(\g,\cdot)}{\g^{\colNo{E}}} &\toverset{\Crefshort{claim: travarc: 1NoQueue},\Crefshort{claim: PropertiesOfSysOpt:trav2,claim: PropertiesOfSysOpt:trav3}}{=} 1 + 1+ \varepsilon = 2+\varepsilon \\
        \dup{\trav(\g,\cdot)}{\g^{\colNo{F}}} &\Croverset{claim: PropertiesOfSysOpt:trav4}{=} 2(4+M)\varepsilon.
    \end{align*}    
    For the minimality of $\dup{\trav(\g,\cdot)}{\g^{\colNo{D}}}$ and $\dup{\trav(\g,\cdot)}{\g^{\colNo{F}}}$ take any $\tilde\g$ induced by some walk inflow rate contained in $\wir$. Then, we have
        \[\dup{\trav(\tilde\g,\cdot)}{\tilde\g^{\colNo{D}}} \geq \int_{\hori'} \trav_{\arc_3}(\tilde\g,\cdot)\cdot \tilde\g^{\colNo{D}}_{\arc_3}\di\sigma \geq \int_{\hori'} \varepsilon\cdot \tilde\g^{\colNo{D}}_{\arc_3}\di\sigma = \varepsilon(M+3) = \dup{\trav(\g,\cdot)}{\g^{\colNo{D}}}\]
    as well as
        \[\dup{\trav(\tilde\g,\cdot)}{\tilde\g^{\colNo{F}}} \geq \int_{\hori'} \trav_{\arc_4}(\tilde\g,\cdot)\cdot \tilde\g^{\colNo{F}}_{\arc_4}\di\sigma \geq \int_{\hori'} \varepsilon\cdot \tilde\g^{\colNo{F}}_{\arc_4}\di\sigma = \varepsilon\cdot2(M+4) = \dup{\trav(\g,\cdot)}{\g^{\colNo{F}}}. \qedhere\]
    \end{proofClaim}

     \begin{proofClaim}[Proof of \Cref{claim: CostDifferenceSysOpImplFlow}]
        By \Cref{claim: PropsSysop} there exists a constant $C \in \R_{\geq 0}$ independent of~$M$ such that
            \[\sum_{i\in\{\colNo{B},\colNo{C},\colNo{E}\}} \dup{\trav(\g,\cdot)}{\g^i} \leq C. \]
        Now, let $u$ be any implementable flow. We then distinguish between which of the two cases of \Cref{claim: ImplemenatbleImpliesHighTT} apply to~$u$:
        \begin{proofbycases}
            \caseitem{\Cref{claim: ImplemenatbleImpliesHighTT:For4}} In this case we estimate
                \begin{align*}
                    \dup{\trav(u,\cdot)}{u}-\dup{\trav(\g,\cdot)}{\g } 
                        &\Croverset{claim: PropsSysop}{\geq} \dup{\trav(u,\cdot)}{u^{\colNo{D}}}+\dup{\trav(u,\cdot)}{u^{\colNo{F}}} - C - (3+M)\varepsilon-\dup{\trav(\g,\cdot)}{\g^{\colNo{F}}} \\
                        &\symoverset{1}{\geq} \dup{\trav(u,\cdot)}{u^{\colNo{D}}} - C - (3+M)\varepsilon \\
                        &\Croverset{claim: ImplemenatbleImpliesHighTT:For4}{\geq} M/4-C,
                \end{align*}
            where in~\refsym{1} we used the minimality of $\dup{\trav(\g,\cdot)}{\g^{\colNo{F}}}$ from \Cref{claim: PropsSysop}.

            \caseitem{\Cref{claim: ImplemenatbleImpliesHighTT:For6}} In this case analogously estimate
                \begin{align*}
                    \dup{\trav(u,\cdot)}{u}-\dup{\trav(\g,\cdot)}{\g } 
                        &\Croverset{claim: PropsSysop}{\geq} \dup{\trav(u,\cdot)}{u^{\colNo{D}}}+\dup{\trav(u,\cdot)}{u^{\colNo{F}}} - C - \dup{\trav(\g,\cdot)}{\g^{\colNo{D}}}-2(4+M)\varepsilon \\
                        &\symoverset{2}{\geq} \dup{\trav(u,\cdot)}{u^{\colNo{F}}} - C - 2(4+M)\varepsilon \\
                        &\Croverset{claim: ImplemenatbleImpliesHighTT:For6}{\geq} M/4-C,
                \end{align*}
            where in~\refsym{2} we used the minimality of $\dup{\trav(\g,\cdot)}{\g^{\colNo{D}}}$ from \Cref{claim: PropsSysop}.
        \end{proofbycases}
    Hence, for $M>4C$, the flow $u$ has strictly higher total travel time than~$\g$ in both cases.   
    \end{proofClaim}

    \subsection{Proofs Omitted in \Cref{lem: WlogNoWaiting}}\label{sec: WlogNoWaiting}
    \begin{proofClaim}[Proof of \Cref{claim: WlogNoWaiting}] 
    We prove the claim via an induction on $j$ with the following induction claim:
\begin{proofbyinduction}
    \inductionclaim For all $j \in \{0,1,\ldots,K\}$, there exists a flow $\tilde{\g}^j \in L_+(\hori)^\GAS$ with
    \begin{enumerate} 
        \item $\opp_v\tilde{\G}^{j}(t)\leq 0$ for all $t \in \hori$ and all $v \in \GV$, \label{indu: Feasible}
        \item $\opp_v\tilde{\G}^{j}(t) = 0 $ for all $t\in [0,j\cdot \tmin] \cup \{t_f\}$ and all $v \in \GV$,\label{indu: NoWai}
         \item $\opp_\sdest \tilde{\G}^{j}(t) \leq \opp_\sdest {\G}(t)$ for all $ t \in \hori$,  \label{indu: SmallerObj} 
         \item $ \tilde{\g}^{j}_{(\ssource,\source_i)} =   \inflow_i$ for all $i\in I$, and\label{indu: NetworkInflow} 
         \item $\tilde{\g}^{j}$ is supported on $\hori'=[0,t_f']$.\label{indu: Support} 
    \end{enumerate}
 \basecase{$j = 0$}This is clear as we can simply set $\tilde{\g}^0 = \g$.
    \inductionstep{$j \to j+1$} 
  Consider $\tilde{\g}^{j}$ and an arbitrary $v \in \GV$. 
    Denote by $\inflow_v^{j+}$ and $\inflow_v^{j-}$ the gross outflow and inflow rate at $v$, i.e.~the functions fulfilling 
   \begin{align*}
        \underset{{\arc \in \delta^-(v)}}{\sum}\int_{\exit_\arc(\tilde{\g}^j_\arc,\cdot)^{-1}([0,t])} \tilde{\g}^j_\arc \di\sigma =  \int_{[0,t]} \inflow_v^{j-} \di\sigma \text{ and }   \underset{{\arc \in \delta^+(v)}}{\sum}\int_{ [0,t] } \tilde{\g}^j_\arc \di\sigma =  \int_{[0,t]} \inflow_v^{j+} \di\sigma 
    \end{align*}
    where $\edgesFrom{v},\edgesTo{v}\subseteq \GAS$ are \wrt the extended network. 
 Note that the gross inflow rates do exist by  the property of the travel time function required in \Cref{ass: TravelGeneral: Inflow}.

  By induction hypothesis, $\tilde{g}^j$ has a nonpositive cumulative flow balance  $\opp_v\tilde{\G}^{j}(t)\leq 0$ for all $t \in \hori$ with equality at $t  =t_f$, i.e.~$\int_{[0,t]}  \inflow_v^{j+}\di\sigma \leq   \int_{[0,t]} \inflow_v^{j-}$ for any $t \in \hori$ with equality at $t = t_f$. Moreover,   
    both functions $\hori \to \R, t \mapsto \int_{[0,t]}  \inflow_v^{j+}\di\sigma$ and $\hori \to \R, t \mapsto \int_{[0,t]} \inflow_v^{j-}\di\sigma$ are continuous and have the same value for $t= 0$ as well as at $t= t_f$. Thus, the function 
    \begin{align*}
        \hfu_v : \hori \to \hori, t \mapsto \min \Big\{\hat{t} \in \hori \mid \int_{[0,\hat{t}]}  \inflow_v^{j-} \di\sigma = \int_{[0,t]}  \inflow_v^{j+}\di\sigma \Big\}
    \end{align*}
     is well-defined, monotone increasing and fulfills $\hfu_v(t) \leq t, t\in\hori$. Furthermore, $\hfu_v$ is measurable by \cite[Corollary 18.8 + Theorem 18.19]{guide2006infinite}.  Thus, $({\inflow_v^{j+}}\cdot\sigma) \circ \hfu_v^{-1}$ is a well-defined Borel-measure on $\hori$. 
    Moreover, we have
    \begin{align}
        \int_{\hfu_v^{-1}([0,t])} \inflow_v^{j+}\di\sigma =   \int_{[0, {t}]}  \inflow_v^{j-} \di\sigma \text{ for all }t \in \hori, v \in \GV. \label{eq: ShiftedInteEquaV} 
    \end{align}
    Note that  the claimed equality in \eqref{eq: ShiftedInteEquaV}   for $t \in \hfu_v(\hori)$ follows by definition of $\hfu_v$ whereas the  equality  for $t \in \hori\setminus \hfu_v(\hori)$ then  follows by  $\opp_v\tilde{\G}^j(t_f)=0$. 

     Hence, the corresponding measures $\inflow_v^{j-}\cdot\sigma = ({\inflow_v^{j+}}\cdot\sigma) \circ \hfu_v^{-1}$ coincide and, consequently, latter measure is absolutely continuous. Using $\sum_{\arc \in \edgesFrom{v}}\tilde{\g}_\arc^j  = \inflow_v^{j+}$ and $\tilde{\g}_\arc^j  \geq 0,\arc \in\edgesFrom{v}$, this now implies that 
    $(\tilde{\g}_\arc^j \cdot\sigma) \circ \hfu_v^{-1},\arc \in \edgesFrom{v} $ are absolutely continuous as well. 
    Hence, their  Radon-Nikodym derivatives exist and we can define $\tilde{\g}^{j+1}_\arc,\arc \in \edgesFrom{v}$ as those derivatives.
    
    Applying the above not all nodes $v \in \GV$ and setting $\tilde{\g}^{j+1}_{(\ssource,\source_i)} \coloneqq \tilde{\g}^j_{(\ssource,\source_i)},i \in I$ we obtain a vector $\tilde{\g}^{j+1}\in L_+(\hori)^\GAS$. Note that $\tilde{\g}^{j+1}$ is  nonnegative since $\tilde{\g}^j\in L_+(\hori)^\GAS$ is likewise.     
    We now show that $\tilde{\g}^{j+1}$ has at least the same cumulative inflow as $\tilde{\g}^{j}$ and is the same on $[0,j\cdot\tmin]$:  
    \begin{subclaim}\label{subclaim: ShiftedFuncs}
        For all $\arc\in \GAS$, the edge inflow rate $\tilde{\g}^{j+1}_\arc$ fulfills 
        \begin{enumerate}[label=\alph*)]
        
             \item $\int_0^t\tilde{\g}^{j+1}_\arc\di\sigma \geq \int_0^t\tilde{\g}^{j}_\arc\di\sigma $ for all $t \in \hori$ with equality at $t = t_f$. \label[thmpart]{subclaim: ShiftedFuncs: Cumul}

            \item $\tilde{\g}^{j+1}_\arc (t)= \tilde{\g}^{j}_\arc(t)$ for almost all  $t\in [0,j\cdot\tmin]$.  \label[thmpart]{subclaim: ShiftedFuncs: EqualOnInterval}  

         \end{enumerate}
    \end{subclaim}
 \begin{proofClaim}
      For $\arc = (\ssource,\source_i),i \in I$ the claims are trivial since we chose $\tilde{\g}^{j+1}_\arc = \tilde{\g}^j_\arc$.
        \begin{structuredproof}
            \proofitem{\ref{subclaim: ShiftedFuncs: Cumul}} 
                       
            We calculate for arbitrary $\arc \in \edgesFrom{v},v \in \GV $ and  $t \in \hori$: 
                \begin{align*}
        \int_0^t\tilde{\g}^{j+1}_\arc\di\sigma = \tilde{\g}^{j+1}_\arc\cdot\sigma([0,t])=  (\tilde{\g}_\arc^j \cdot\sigma) \circ \hfu_v^{-1}    ([0,t]) =  \int_{\hfu_v^{-1}([0,t])}\tilde{\g}_\arc^j\di\sigma \symoverset{1}{\geq}  \int_{[0,t]}\tilde{\g}_\arc^j\di\sigma. 
    \end{align*}
    Here we used for the inequality indicated with \refsym{1} that $\hfu$ is  monotone increasing and fulfills $\hfu_v(t) \leq t$, implying $\hfu_v^{-1}([0,t]) \supseteq [0,t]$. Furthermore, we note that this inequality is tight for $t =t_f$ since $\hfu_v^{-1}(\hori) = \hori$.  

    \proofitem{\ref{subclaim: ShiftedFuncs: EqualOnInterval}} 
    Consider  $\arc \in \edgesFrom{v}$ for arbitrary $v \in \GV$. 
        On the interval $[0,j\cdot \tmin]$, we can derive the following equalities: 
         \begin{align*}
            \sum_{\arc \in \edgesFrom{v}}  \tilde{\g}^{j } \cdot \sigma  & =   \inflow_v^{j+}\cdot\sigma \symoverset{1}{=} \inflow_v^{j-}\cdot\sigma = ({\inflow_v^{j+}}\cdot\sigma) \circ \hfu_v^{-1}   
              = (\sum_{\arc \in \edgesFrom{v}} \tilde{\g}^j \cdot \sigma ) \circ \hfu_v^{-1}\\ &=  
              \sum_{\arc \in \edgesFrom{v}} (\tilde{\g}^j \cdot \sigma ) \circ \hfu_v^{-1} =    \sum_{\arc \in \edgesFrom{v}}  \tilde{\g}^{j+1} \cdot \sigma  
         \end{align*}
         where we used for \refsym{1} that  $\opp_v\tilde{\G}^j(t)= 0$  for  $t \in [0,j\cdot \tmin]$ holds by our induction hypothesis. 
         With this and \ref{subclaim: ShiftedFuncs: Cumul}, we can deduce that  $\tilde{\g}_\arc^j \cdot\sigma= \tilde{\g}_\arc^{j+1} \cdot\sigma,\arc \in \edgesFrom{v} $   on $[0,j\cdot \tmin]$. \qedhere    
        \end{structuredproof}
    \end{proofClaim}

   With this, we can now verify that $\tilde{\g}^{j+1}$ satisfies all of the required properties: 

\begin{structuredproof} 
\proofitem{\ref{indu: Feasible}}
For arbitrary $v \in \GV$ and $t \in \hori$ we calculate:
    \begin{align}
        \opp_v\tilde{\G}^{j+1} (t) &= \sum_{\arc\in \edgesFrom{v}} \int_{[0,t]}  \tilde{\g}^{j+1}_\arc 
  \di\sigma - \sum_{\arc\in \edgesTo{v}} (\tilde{\G}^{j+1})^-_\arc(t) \nonumber\\
  &\leq \sum_{\arc\in \edgesFrom{v}} \int_{[0,t]}  \tilde{\g}^{j+1}_\arc 
  \di\sigma - \sum_{\arc\in \edgesTo{v}} (\tilde{\G}^{j})^-_\arc(t) \label{eq: IneqInsteadOfEq}\\ 
  &=  \sum_{\arc\in \edgesFrom{v}}   \int_{\hfu_v^{-1}([0,t])}  \tilde{\g}^{j}_\arc 
  \di\sigma - \sum_{\arc\in \edgesTo{v}} (\tilde{\G}^{j})^-_\arc(t) \nonumber \\
  &=      \int_{\hfu_v^{-1}([0,t])} \sum_{\arc\in \edgesFrom{v}} \tilde{\g}^{j}_\arc 
  \di\sigma - \sum_{\arc\in \edgesTo{v}} (\tilde{\G}^{j})^-_\arc(t) \nonumber \\
  &=      \int_{\hfu_v^{-1}([0,t])} \inflow_v^{j+}
  \di\sigma - \sum_{\arc\in \edgesTo{v}} (\tilde{\G}^{j})^-_\arc(t)\nonumber\\
  \overset{\eqref{eq: ShiftedInteEquaV}}&{=}\int_{[0,t]} \inflow_v^{j-} 
  \di\sigma - \sum_{\arc\in \edgesTo{v}} (\tilde{\G}^{j})^-_\arc(t)   =  0\nonumber
     \end{align}
where we used for \eqref{eq: IneqInsteadOfEq} that  $(\tilde{\G}^{j+1})^-_\arc(t) \geq (\tilde{\G}^{j})^-_\arc(t),t \in \hori$ due to 
\Cref{subclaim: ShiftedFuncs: Cumul} and \Cref{ass: TravelMono: Monotonicity}. 

\proofitem{\ref{indu: NoWai}}
    Using the same calculation from \ref{indu: Feasible} it only remains to show that we have equality at \eqref{eq: IneqInsteadOfEq} for all $t\in [0,j\cdot \tmin]$ as well as for $t=t_f$. The latter follows immediately from the equality in \Cref{subclaim: ShiftedFuncs: Cumul}.  

    For the former we first observe that by \Cref{ass: TravelGeneral: tmin} we have for any $\arc \in \GA$ that $\exit_\arc(\tilde{g}^{j+1}_\arc,t) \geq t + \trav_\arc(\tilde{g}^{j+1}_\arc,t) \geq t + \tmin > (j+1)\cdot\tmin$ for all $t \in (j\cdot\tmin,t'_f]$ and, therefore, $\exit_\arc(\tilde{g}^{j+1}_\arc,\cdot)^{-1}([0,t]) \subseteq [0,j\cdot\tmin]$ for all $t  \in [0,(j+1)\cdot\tmin]$.
    Hence, \Cref{subclaim: ShiftedFuncs: EqualOnInterval} implies that $(\tilde{\G}^{j+1})^-_\arc(t)  = (\tilde{\G}^{j})^-_\arc(t)$ for all $t \in [0,(j+1)\cdot\tmin]$ and all $\arc \in \GA$. 
    Since also $(\tilde{\G}^{j+1})^-_{(\ssource,\source_i)} =(\tilde{\G}^{j})^-_{(\ssource,\source_i)}$ by  
    $\tilde{\g}^{j+1}_{(\ssource,\source_i)} =\tilde{\g}^{j}_{(\ssource,\source_i)}$ and $\trav_{(\ssource,\source_i)}\equiv 0$ for all $i \in I$ holds, we get an equality in \eqref{eq: IneqInsteadOfEq} for all $t \in [0,(j+1)\cdot\tmin]$ and all $v \in \GV$, showing the claim.
    
    \proofitem{\ref{indu: SmallerObj}} This is a direct consequence of the induction hypthesis and \Cref{subclaim: ShiftedFuncs: Cumul} together with the following identity for any $\hat{\g} \in L_+(\hori)^\GAS$:
    \begin{align*}
          \opp_{\sdest}\hat\G(t) = -\int_{[0,t]} \hat\g_{(\dest,\sdest)} \di\sigma.
    \end{align*}
    The above holds as  $\sdest$ is only incident to $(\dest,\sdest)$ and the corresponding travel time is always equal to zero.
    
    \proofitem{\ref{indu: NetworkInflow}} This is clear by induction hypothesis and the definition $\tilde{\g}^{j+1}_{(\ssource,\source_i)} := \tilde{\g}^j_{(\ssource,\source_i)}, i\in I$.  
    
    \proofitem{\ref{indu: Support}} This is, again, a direct consequence of \Cref{subclaim: ShiftedFuncs: Cumul} and the induction hypothesis: Since we have $(\tilde \G^{j+1})_\arc^+(t'_f) \geq (\tilde G^{j})^+_\arc(t'_f)$ as well as $(G^{j+1})_\arc^+(t_f) = (\tilde G^{j})^+_\arc(t_f)$ and $\tilde \g^j_\arc$ is supported on $[0,t_f']$, $\tilde \g^{j+1}_\arc$ can only be supported on $[0,t_f']$ as well. 
\end{structuredproof}
\end{proofbyinduction}
\end{proofClaim}

    \subsection{Proofs Omitted in \Cref{thm: SsSysOptNoDOutflow}}\label{sec: ProofsSsSysOptNoDOutflow}

 \begin{proofClaim}[Proof of \Cref{claim: FlowCarryingDestCycle}]
 Let $v \in \GV$ and $t \in \hori$ be arbitrary. We calculate (explanations follow): 
\begin{align*}
    \opp_{v}\G^\Delta(t) &=   \sum_{\arc \in \edgesFrom{v}} \int_{[0,t]} \g^\Delta_\arc \di\sigma -    \sum_{\arc \in \edgesTo{v}} \int_{\exit_\arc(\g^\Delta_\arc,\cdot)^{-1}([0,t])} \g^\Delta_\arc \di\sigma \\
 &=  \sum_{\arc \in \edgesFrom{v}} \int_{[0,t]}  \g^\Delta_\arc \di\sigma -    \sum_{\arc \in \edgesTo{v}} \int_{\exit_\arc(\g^\Delta_\arc,\cdot)^{-1}([0,t])} \gop_\arc - \g^c_\arc  \di\sigma   \\  
&\symoverset{1}{\leq} \sum_{\arc \in \edgesFrom{v}} \int_{[0,t]}  \g^\Delta_\arc \di\sigma  -  \sum_{\arc \in \edgesTo{v}} \int_{\exit_\arc(\gop_\arc,\cdot)^{-1}([0,t])} \gop_\arc - \g^c_\arc  \di\sigma \\
&= \sum_{\arc \in \edgesFrom{v}} \int_{[0,t]} \gop_\arc - \g^c_\arc +1_{(\dest,\sdest)}(\arc) \cdot(-\ell^u_{\wa,\abs{\wa}}(h^*_\wa)+\ell^u_{\wa,j^*}(h^*_\wa)) \di\sigma \\
&\quad \quad -\sum_{\arc \in \edgesTo{v}}  (\Gop)^{-}_\arc(t) - \int_{\exit_\arc(\gop_\arc,\cdot)^{-1}([0,t])}  \g^c_\arc  \di\sigma\\
&= \sum_{\arc \in \edgesFrom{v}}\Bigl((\Gop)^{+}_\arc(t)  + \int_{[0,t]}  -\g^c_\arc\di\sigma +1_{(\dest,\sdest)}(\arc) \cdot\int_{[0,t]}(-\ell^u_{\wa,\abs{\wa}}(h^*_\wa)+\ell^u_{\wa,j^*}(h^*_\wa)) \di\sigma\Bigr) \\
&\quad \quad -\sum_{\arc \in \edgesTo{v}}  \Bigl((\Gop)^{-}_\arc(t) - \int_{\exit_\arc(\gop_\arc,\cdot)^{-1}([0,t])}  \g^c_\arc  \di\sigma\Bigr)\\
&\symoverset{3}{=} \opp_v\Gop(t) - \op_v\g^c([0,t]) + 1_{\dest}(v) \cdot \int_{[0,t]} -\ell^u_{\wa,\abs{\wa}}(h^*_\wa)+\ell^u_{\wa,j^*}(h^*_\wa) \di\sigma \\
&\symoverset{2}{=}  \opp_v\Gop(t) \symoverset{4}{=} 0
\end{align*}

For the inequality indicated by \refsym{1}, we used  $\g^*_{\arc}-\g^c_\arc  \geq 0$ and  the monotonicity of the travel times functions stated in \Cref{ass: TravelMono: TravelMonotonicity} which implies that $\exit_\arc(\g^{\Delta}_\arc ,\cdot)^{-1}([0,t]) \supseteq \exit_\arc(\g^*_\arc ,\cdot)^{-1}([0,t])$ for all $\arc \in \GAS$. Moreover, note that the inequality \refsym{1} is tight for $t = t_f$ as the previous inclusion is an equality in that case.  
The equality \refsym{3}  follows by the definition of the node balances and parameterized node balance for $u = \gop$,  cf.~\Cref{sec:uBasedNetworkLoadings:FlowBala}.  
For the equality \refsym{2}, remark that by \Cref{lem: flowcon}, $\g^c$ fulfills parameterized (\wrt $u= \gop$) flow conservation at all $v \in \GV\setminus\{\dest\}$ (i.e.~$ \op_v\g^c([0,t]) = 0,t\in \hori$) while at $\dest$ it is exactly equal to $\op_{\dest}\g^c([0,t]) =  \int_{[0,t]} -\ell^u_{\wa,\abs{\wa}}(h^*_\wa)+\ell^u_{\wa,j^*}(h^*_\wa) \di\sigma$ for all $t \in \hori$. Finally, \refsym{4} holds as $\gop$ fulfills flow conservation at all nodes $v \in \GV$.

In order to show the second part of the claim, we first observe that for any $\g \in L_+(\hori)^\GAS$ we have
\begin{align*} 
      \opp_{\sdest}\G(t) = -\int_{[0,t]} \g_{(\dest,\sdest)} \di\sigma
\end{align*}
since $\sdest$ is only incident to $(\dest,\sdest)$ and the corresponding travel time is always equal to zero.
Now we calculate
\begin{align}\label{eq: repres}\begin{split}
      \opp_{\sdest}\G^\Delta(t)= - \int_{[0,t]} \g^\Delta_{(\dest,\sdest)} \di\sigma &=  -\int_{[0,t]} \gop_{(\dest,\sdest)}  -\ell^u_{\wa,\abs{\wa}}(h^*_\wa)+\ell^u_{\wa,j^*}(h^*_\wa)   \di\sigma\\
      &=    \opp_{\sdest}\Gop(t) - \int_{[0,t]} -\ell^u_{\wa,\abs{\wa}}(h^*_\wa)  + \ell^u_{\wa,j^*}(h^*_\wa)   \di\sigma. 
\end{split}\end{align}
By \Cref{lem: elluPropagation}, we have for all $t \in \hori$ the equality
\begin{align*}
    \int_{\arr_{\wa_{\geq j^*},\abs{\wa_{\geq j^*}}}(u,\cdot)^{-1}([0,t])} \ell^u_{\wa,j^*}(h^*_\wa)  \di\sigma  = \int_{[0,t]} \ell^u_{\wa,\abs{\wa}}(h^*_\wa) \di\sigma. 
\end{align*}
Since the travel times are non-negative, we have $\arr_{\wa_{\geq j^*},\abs{\wa_{\geq j^*}}}(u,\cdot)^{-1}([0,t]) \subseteq [0,t]$ and can, therefore, further estimate
\begin{align}\label{eq: InflowVsOutflow}
     \int_{[0,t]} \ell^u_{\wa,j^*}(h^*_\wa)   \di\sigma \geq   \int_{\arr_{\wa_{\geq j^*},\abs{\wa_{\geq j^*}}+1}(u,\cdot)^{-1}([0,t])} \ell^u_{\wa,j^*}(h^*_\wa)  \di\sigma  = \int_{[0,t]} \ell^u_{\wa,\abs{\wa}}(h^*_\wa) \di\sigma. 
\end{align}
Using this in~\eqref{eq: repres} gives us $\opp_\sdest \G^\Delta(t) \leq \opp_\sdest \G^*(t)$ for all $t \in\hori$. 
Moreover, for  $\bar t = \sup\{t \in \hori \mid  \int_{[0,t]}\ell^u_{\wa,j^*}(h^*_\wa) \di\sigma =0\}$, we know that $\bar t < t'_f$ since $\ell^u_{\wa,j^*}(h^*_\wa) \in L_+(\hori)\setminus\{0\}$ and the flow induced by $h^*$ is supported in $\hori'=[0,t'_f]$. For $t \in [\bar t,t'_f]$, we know from \Cref{ass: TravelGeneral: tmin} that the travel times are lower bounded by $\tmin$ for edges in $\GA$. Since $\wa_{\geq j^*}$ contains at least one of them ($\wa[j^*] = \arc^*$), we get $\arr_{\wa_{\geq j^*},\abs{\wa_{\geq j^*}}}(u,\cdot)^{-1}[0,t] \subseteq [0,t-\tmin]$ for such $t$ and can, therefore, strengthen \eqref{eq: InflowVsOutflow} to
\begin{align}\label{eq: InflowVsOutflowS}
     \int_{[0,t]} \ell^u_{\wa,j^*}(h^*_\wa)   \di\sigma &\geq \int_{[0,t-\tmin]} \ell^u_{\wa,j^*}(h^*_\wa)   \di\sigma \\ &\geq   \int_{\arr_{\wa_{\geq j^*},\abs{\wa_{\geq j^*}}}(u,\cdot)^{-1}([0,t])} \ell^u_{\wa,j^*}(h^*_\wa)  \di\sigma = \int_{[0,t]} \ell^u_{\wa,\abs{\wa}}(h^*_\wa) \di\sigma.  \nonumber
\end{align}
Here then, for $t \in (\bar t,\min\{\bar t+\tmin,t'_f\}]$ the leftmost term is strictly larger than zero, while the second term is zero. Hence, we get strict inequality in~\eqref{eq: InflowVsOutflowS} and putting this back into~\eqref{eq: repres} gives us $\opp_\sdest \G^\Delta(t) < \opp_\sdest \G^*(t)$ for all such $t$ as well.
%
 \end{proofClaim}

\begin{proofClaim}[Proof of \Cref{claim: TotalTravelRewrite}]
    For any   $\g \in \ell(\swir) \subseteq L_+(\hori)^\GAS$, we can rewrite the total travel time as follows (for $u = \g$): 
\begin{align}
   \dup{\trav(\g,\cdot)}{{\g}}  &=   \sum_{v \in \GVS} \int_\hori-\id \di \big(\op_v {\g}\big) \label{eq: l1}\\
    &=\int_\hori-\id \di \big(\op_\ssource {\g}\big) + \int_\hori-\id \di \big(\op_\sdest  {\g}\big) \label{eq: l2} \\
    &= \int_\hori-\id \di \big(\op_\ssource {\g}\big)-t_f \cdot \op_\sdest\g(\hori )-\int_{\hori}-\opp_\sdest\G \di\sigma \label{eq: l3}\\
    &= \int_\hori-\id \di \big(\sum_{i \in I} \inflow_i \cdot \sigma\big) + t_f \cdot  \int_{\hori}\sum_{i \in I}\inflow_i \di\sigma+ \int_{\hori}\opp_\sdest\G \di\sigma \label{eq: l4} \\
    &=  \int_\hori (t_f - \id) \cdot \sum_{i \in I}\inflow_i \di\sigma  + \int_{\hori}\opp_\sdest\G \di\sigma\label{eq: l5}
\end{align}
The equality in the first line \eqref{eq: l1} is due to \Cref{lem: flowcontrav}. 
The equality \eqref{eq: l2} is implied by $\g$ fulfilling flow conservation at all nodes $v\in \GV$, i.e.~$\op_v\g=0,v \in \GV$. 
By \cite[Exercise 5.8.112]{Bogachev2007I}, we get \eqref{eq: l3} while \eqref{eq: l4} then follows by choosing a $h \in \ell(\swir)$ with $\ell(h)=g$ and applying \Cref{lem: flowcon} to all components~$h_w$ of~$h$.
\end{proofClaim}
 
\clearpage 
 \section{Properties of Parameterized Network Loadings Derived in \cite{GHS24FD}}\label{sec:AppuBasedNetworkLoadings}

\subsection{Existence of Parametrized Network Loadings}

\begin{lemma}[\FDref{lem: elluinj}]\label{lem: elluinj}
       Consider an arbitrary walk $\wa$, $j\in [|\wa|+1]$ and $h_\wa \in L(\hori)$. If $\ell^u_{\wa,j}(h_\wa)$ exists, then $h_\wa = 0$ almost everywhere on $\arr_{\wa,j}(u,\cdot)^{-1}(\arr_{\wa,j}(u,\cdot)(\mathfrak T))\setminus \mathfrak T$ for any $\mathfrak T\in \mathcal{B}(\hori)$.
\end{lemma}

\subsection{Parameterized Node Balances and  \boldmath{$\source$,$\dest$}-Flows} \label{sec:uBasedNetworkLoadings:FlowBala}

The parameterized node balance at node $v \in \GV$ for an arbitrary vector $\eflow \in L(\hori)^\GA$ is given by the measure  $\op_v\eflow \coloneqq \sum_{\arc \in \delta^+(v)}   \eflow_\arc \cdot \sigma -  \sum_{\arc \in \delta^-(v)}  ( \eflow_\arc \cdot\sigma) \circ \exit_\arc(u,\cdot)^{-1}$ which describes for an arbitrary $\mathfrak T \in \mathcal{B}(\hori)$ the difference between the cumulative inflow into $v$   and the cumulative outflow from $v$ during $\mathfrak T$, i.e.
\begin{align}\label{eq: FlowBalance}
    \op_v \eflow(\mathfrak T) =  
\sum_{\arc \in \delta^+(v)} \int_{\mathfrak T} \eflow_\arc \di \sigma -  \sum_{\arc \in \delta^-(v)} \int_{\exit_\arc(u,\cdot)^{-1}(\mathfrak T)} \eflow_\arc \di \sigma. 
\end{align}
A vector $\eflow \in L_+(\hori)^\GA$ that has a net outflow rate $\inflow_s \in L_+(\hori)$ at $s$, fulfills flow conservation at 
all $v \neq \source,\dest$ and has a nonpositive node balance at $\dest$ is called parametrized $\source$,$\dest$-flow (\wrt $u$). 
Here, we say that the node balance $\op_v\eflow$ is nonpositive if \eqref{eq: FlowBalance} is nonpositive for any $\mathfrak  T \in \mathcal{B}(\hori)$.

  \begin{lemma}[\FDref{lem: flowcon}]\label{lem: flowcon}
Consider an arbitrary $v_1,v_2$-walk $\wa$, a corresponding walk inflow rate $h_\wa \in L(\hori)$ with $\g^\wa:=\ell^u_\wa(h_\wa)$ existing and a node $v \in \GV$.
Then we have 
\begin{align*} 
    \op_v \g^\wa = 1_{v_1}(v)\cdot  h_\wa \cdot \sigma   -1_{v_2}(v) \cdot (h_\wa \cdot\sigma)\circ\arr_{\wa,|\wa|+1}(u,\cdot)^{-1}.
\end{align*}
If furthermore $\ell^u_{\wa,\abs{\wa}+1}(h_\wa)$ exists, then $\op_{v_2} \g^\wa =  1_{v_1}(v_2)\cdot  h_\wa  -\ell^u_{\wa,\abs{\wa}+1}(h_\wa) \cdot \sigma$, i.e.~$\g^\wa$ has the outflow rate $ 1_{v_1}(v_2)\cdot  h_\wa  -\ell^u_{\wa,\abs{\wa}+1}(h_\wa)$ at the end node $v_2$. 
       \end{lemma}

\subsection{Properties of Parametrized \boldmath{$\source$,$\dest$}-Flows}

\begin{lemma}[\FDref{lem: Relations:h>0u>0}]\label{lem: Relations:h>0u>0} 
Let $\Routes'$ be an arbitrary countable collection of walks and $h \in \edom{\Routes'}\cap L_+(\hori)^{\Routes'}$ with $\g:=\ell^u(h)$.  The following statement is 
    true for all 
    $\arc \in \GA$: 
    \begin{itemize}
        \item[d)] For all $\mathfrak T \in \mathcal{B}(\hori), \sigma(\mathfrak T)> 0$ with ${\g}_\arc(t)>0$ for a.e.~$t \in \mathfrak T$,
    we can find a countable set $M$ and walks $\wa^m,m\in M$ together with indices $j_m\leq |\wa^m|$ and measurable sets $\mathfrak D_m, \sigma(\mathfrak D_m)>0$ for all $m \in M$ such that $\wa^m[j_m] = \arc$, ${h}_{\wa^m}(t) > 0$ for a.e.~$t \in \mathfrak D_m$ and $\arr_{\wa^m,j_m}(u,\cdot)(\mathfrak D_m)$ are disjoint with $\bigcup_{m \in M} \arr_{\wa^m,j_m}(u,\cdot)(\mathfrak D_m)$ equalling $\mathfrak T$ up to a null set. Furthermore, for any walk $\wa \in {\Routes'}$, there are only finitely many $m \in M$ with $\wa^m=\wa$.     \label[thmpart]{lem: Relations:h>0u>0:u>0ExistsCountableM}
    \end{itemize}
\end{lemma}

\begin{lemma}[\FDref{lem: elluPropagation}]\label{lem: elluPropagation}
     Consider an arbitrary walk $\wa$, two edge indices $j_1\leq j_2\leq|\wa|+1$ and $h_\wa \in L(\hori)$ with $\ell^u_{\wa,j'}(h_\wa),j'\in\{j_1,j_2\}$ existing. Then, we have the equality 
     $\ell^u_{\wa,j_2}(h_\wa) = \ell^u_{\wa_{\geq j_1},j_2-j_1+1}(\ell^u_{\wa,j_1}(h_\wa))$. 
\end{lemma}

 \begin{lemma}[\FDref{lem: flowcontrav}]\label{lem: flowcontrav}    For any  $\eflow \in L(\hori)^\GA$, we have  
 \begin{align*}
      \dup{\trav(u,\cdot)}{\eflow} =  \sum_{v \in \GV}\int_{\hori} -\id \di (\op_v   \eflow)
 \end{align*}
  with $-\id$ denoting the identity function in $L(\hori)$.  
 \end{lemma} 

\subsection{Existence of Flow Decomposition} 

\begin{theorem}[\FDref{thm: FlowDecomp}]\label{thm: FlowDecomp}
    Every parameterized $\source$,$\dest$-flow has a parameterized flow decomposition, that is, a vector of walk inflow rates $h \in L_+(\hori)^{\hat{\Routes}}$ together with zero-cycle inflow rates $h \in L_+(\hori)^{\mathcal{C}}$ such that $\eflow = \sum_{\wa \in \hat{\Routes}}\ell^u_\wa(h_\wa) + \sum_{c \in \mathcal{C}}\ell^u_c(h_c)$.
\end{theorem}

 \section{Insights Regarding Tolls for Dynamic Flows Derived in \cite{GHS24TollSODA}}\label{sec:AppTolls}

 \subsection{Adjoint Operator of \texorpdfstring{\boldmath$\ell^u$}{lu}}
 
\begin{lemma}[\Tollref{lem: aggCostsVSwalkCosts}] \label{lem: aggCostsVSwalkCosts} 
Consider  $h \in \edom{\Routes}$ with $\g:= \ell^u(h)$   as well as a measurable function $\prices:\hori\to \R^\GA_+$. 

Then  
the equality  $\dup{\prices}{\g} = \sum_{\wa \in \Routes}\dup{\Pf^{\prices}_\wa(u,\cdot)}{{h}_\wa}$ is valid (with possibly both expressions being equal to $\infty$). 

In particular, for $h\in \wir$ we have $\gamma_i\cdot \dup{\trav(u,\cdot)}{\g^i} = \sum_{\wa \in \Routes_i}\dup{\Psi_\wa(u,\cdot)}{{h}_\wa}$ for any $i\in I$.

\end{lemma}

\subsection{Implementability}

We consider  a single-source, singe-\sink network in this section.

\begin{theorem}[\Tollref{thm: CombiChara}]\label{thm: CombiChara}
    For networks where all commodities share the same source and \sink and the travel times are strictly positive, the following statements are equivalent: 
    \begin{thmparts}
        \item $u$ is implementable. \label[thmpart]{thm: CombiChara:Optimal}
    \item $u$ does not send flow along outgoing edges from the destination, i.e.~$u_\arc = 0$ for all $\arc \in \delta^+(\dest)$. 
\end{thmparts}
\end{theorem}

\section{List of Symbols}

{
	\newcommand{\losEntry}[2]{#1 & #2 \\}
	\renewcommand{\arraystretch}{1.2}
	
	\begin{longtable}{p{4cm}p{10cm}}
		Symbol				& Description \\\hline
		
		\hline\multicolumn{2}{l}{\textbf{General}}\\\hline
		
		\losEntry{$L(\hori)$}{space of integrable functions on $\hori$}
		\losEntry{$L_+(\hori)$}{non-negative functions in $L(\hori)$}
        \losEntry{$\seql[1][M][L(\hori)]$}{vectors  $(h_m)_{m\in M}\in (L(\hori))^M$ for an arbitrary countable set $M$ whose corresponding series $\sum_{m\in M}h_m$ converges absolutely in $L(\hori)$. }
        \losEntry{$\seql[\infty][M][L^\infty(\hori)]$}{vectors  $(h_m)_{m\in M}\in (L^\infty(\hori))^M$  whose entries are uniformly bounded, i.e.~$\sup_{m \in M}\norm{h_m}_{\infty}< \infty$. }
		\losEntry{$L^\infty(\hori)$}{space of measurable essentially bounded functions on $\hori$}
		\losEntry{$L^\infty_+(\hori)$}{non-negative functions in $L^\infty(\hori)$}

		\losEntry{$\sigma$}{the Lebesgue measure on $\hori$}
		\losEntry{$\mathfrak{T}, \mathfrak{D}$}{measurable subsets of $\hori$}

        \losEntry{$1_{\mathfrak T}$}{characteristic function of a  set~$\mathfrak T$, i.e. $1_{\mathfrak T}(t)=1$ if $t \in \mathfrak T$ and $0$, otherwise}
        
        \losEntry{$1_{\arc}:=1_{\{\arc\}}$}{characteristic function of singleton set $\{\arc\}$}        

        \losEntry{$\dup{f}{g}$}{the bilinear form between the dual pair $(\seql[1][M][L(\hori)],\seql[\infty][M][L^\infty(\hori)])$, i.e.\ $\dup{f}{g} \coloneqq \sum_{m \in M}\int_\hori f_m\cdot g_m \sigma$}

		\hline\multicolumn{2}{l}{\textbf{Network}}\\\hline
		
		\losEntry{$G=(\GV,\GA)$}{directed graph with nodes $\GV$ and edges $\GA$}
		\losEntry{$\edgesFrom{v}$}{edge starting from node $v$}
		\losEntry{$\edgesTo{v}$}{edge ending at node $v$}
		\losEntry{$\source_i \in V$}{source node of commodity $i$}
		\losEntry{$\dest_i \in V$}{destination node of commodity $i$}
        \losEntry{$\hori'=[0,t_f']$}{planning horizon on which all flow has to be supported}
        \losEntry{$\hori:=[0,t_f]$}{enlarged (due to technical reasons) planning horizon}
		\losEntry{$t \in \hori$}{time}
		\losEntry{$\hat\Routes$}{set of (finite) \stwalk s}
		\losEntry{$\Routes \coloneqq \bigcup_i\Routes_i$}{set of walk-commodity pairs where $\Routes_i \coloneqq \hat\Routes \times \set{i}$}
  
		\losEntry{$\wa=(\arc_1,\dots,\arc_k)$}{walk consisting of edges $\arc_j=(v_j,v_{j+1})$}
		\losEntry{$\wa[j]$}{$j$-th edge on walk $\wa$}
		
		\losEntry{$I$}{(finite) set of commodities}
		\losEntry{$\gamma_i > 0$}{value-of-time parameter of commodity~$i$}
		\losEntry{$\inflow_i \in L_+(\hori)$}{network inflow rate of commodity $i$}

		\hline\multicolumn{2}{l}{\textbf{Flows}}\\\hline
		\losEntry{$\wir$}{set of admissible walk-inflows}		
		\losEntry{$h \in \wir\subseteq L_+(\hori)^{\Routes}$}{multi commodity walk-inflow} 
		\losEntry{$\g = \ell(h) \in L_+(\hori)^\GA$}{aggregated edge flow of $h$}
 	\losEntry{ $\ell: \wir\to L_+(\hori)^\GA$}{network loading mapping from $h$ to $\g$} 
        \losEntry{$\trav_\arc(\g,t)$}{edge traversal time under $\g$ when entering edge $\arc$ at time $t$ under $\g$ (absolutely continuous)}
		\losEntry{$\exit_\arc(\g,t)\in \hori$}{edge exit time when entering edge $\arc$ at time $t$ under $\g$: $\exit_\arc(\g,t) \coloneqq t+\trav_\arc(\g,t)$ (non-decreasing)}
		\losEntry{$\arr_{\wa,j}(\g,t) \in \hori$}{arrival time in front of the $j$-th edge of walk $\wa$ when entering this walk at time $t$ under $\g$}
		
		\losEntry{$\Psi_{(\hat{\wa},i)}(\g,t)$}{weighted total travel time for particles of commodity $i$ entering walk $\hat{\wa}$ at time $t$ under $\g$}
		
		\losEntry{$\prices: \hori \to \Rnn^A$}{time dependent edge-tolls}
		\losEntry{$\Pf^\prices_\wa(\g,t)$}{total toll along walk $\wa$ when entering at time $t$ under $\g$ and $\prices$}
		\losEntry{$\g^i \in L_+(\hori)^\GA$}{edge flow of commodity $i$ under $h$}

        \hline\multicolumn{2}{l}{\textbf{Parameterized  Flows}}\\\hline
        \losEntry{$\wir^u$}{set of admissible walk-inflows for the parameterized network loading}
		\losEntry{$h \in \wir^u\subseteq L_+(\hori)^{\Routes}$}{multi commodity walk-inflow} 
		\losEntry{$\g = \ell^u(h) \in L_+(\hori)^\GA$}{aggregated edge flow of $h$ under $u$}
		\losEntry{$\ell^u: \wir^u \to L_+(\hori)^\GA$}{parametrized network loading mapping from $h$ to parameterized flow $\g$} 
        \losEntry{$\ell^u_{\wa,j}( h_\wa)$}{the flow on the $j$-th edge on walk $\wa$ induced by inflow into that walk under $ h_\wa$ under $u$, without aggregating over multiple occurrences of that edge}
		\losEntry{$\ell^u_{\wa,\arc}( h_\wa)$}{the flow on edge $\arc$ on walk $\wa$ induced by inflow into that walk under $ h_\wa$ under $u$,  aggregated over multiple occurrences of that edge}
		
 	\end{longtable}
}

\end{document}